\documentclass[%
 aip,
 jmp,
 amsmath,amssymb,
 reprint, onecolumn
]{revtex4-1}

\usepackage{graphicx}
\usepackage{dcolumn}
\usepackage{bm}
\usepackage[utf8]{inputenc}
\usepackage[T1]{fontenc}
\usepackage{mathptmx}
\usepackage{etoolbox}
\usepackage{dsfont}

\usepackage{amsthm}
\usepackage{enumitem}
\usepackage{mathtools}
\usepackage{braket}
\usepackage{tikz}
\usetikzlibrary{arrows.meta,decorations.pathreplacing}
\usepackage{pgfplots}
\pgfplotsset{compat=1.18}
\usepackage{hyperref}
\hypersetup{colorlinks=true,linkcolor=blue,citecolor=blue,urlcolor=blue,
  pdftitle={Newton-Cartan limit of Klein-Gordon AQFT: gravitational atoms and the loss of vacuum entanglement},
  pdfauthor={Leonardo A. Pach\'on}}

\newtheorem{theorem}{Theorem}
\newtheorem{proposition}{Proposition}
\newtheorem{lemma}{Lemma}
\newtheorem{corollary}{Corollary}
\newtheorem{alemma}{Lemma}

\theoremstyle{definition}
\newtheorem{definition}{Definition}
\theoremstyle{remark}
\newtheorem{remark}{Remark}

\makeatletter
\def\@email#1#2{%
 \endgroup
 \patchcmd{\titleblock@produce}
  {\frontmatter@RRAPformat}
  {\frontmatter@RRAPformat{\produce@RRAP{*#1\href{mailto:#2}{#2}}}\frontmatter@RRAPformat}
  {}{}
}%
\makeatother

\newcommand{\R}{\mathbb{R}}
\newcommand{\C}{\mathbb{C}}
\newcommand{\Z}{\mathbb{Z}}

\newcommand{\F}{\mathcal{F}}
\newcommand{\Hil}{\mathcal{H}}
\newcommand{\Op}{\mathcal{O}}
\newcommand{\K}{\mathcal{K}}
\newcommand{\one}{\mathds{1}}
\newcommand{\rmd}{\mathrm{d}}
\newcommand{\rmi}{\mathrm{i}}
\newcommand{\rme}{\mathrm{e}}
\newcommand{\bv}{\mathbf{v}}
\DeclareMathOperator{\supp}{supp}

\begin{document}

\preprint{}

\title{Newton--Cartan limit of Klein--Gordon AQFT: gravitational atoms
and the loss of vacuum entanglement}

\author{Leonardo A.\ Pach\'{o}n}
\affiliation{guane Enterprises, R+D+I Unit, Medell\'{i}n 050010, Colombia}
\email{leonardo@guane.com.co}

\date{\today}

\begin{abstract}
We take the non-relativistic limit $c\to\infty$ of the free
Klein--Gordon field on static spacetimes at three levels: one-particle
resolvents, quasi-free states and the time-zero net. After subtraction
of the rest energy, the static one-particle Hamiltonian is an explicit
function of a Schr\"odinger-type operator whose potential is the exact
Newtonian potential energy on the Schwarzschild exterior. On regular
static stars the rest-energy-subtracted Hamiltonian converges to the
Newtonian Schr\"odinger operator in norm-resolvent sense at rate
$c^{-2}$, with convergence of eigenvalues and an explicit first
post-Newtonian correction; for metrics analytic in $c^{-2}$ the simple
eigenvalues, and all levels of spherical stars, are analytic in
$c^{-2}$. On the Schwarzschild exterior the
Boulware one-particle spectrum has no eigenvalues for any $c$; the
convergence is strong but not in norm, the limit selects the Friedrichs
realisation of the gravitational hydrogen atom, and its bound states
emerge by spectral concentration. In both cases the rescaled two-point
functions converge to those of the Schr\"odinger--Fock vacuum. The
hydrogenic states are limits of resonances: in each partial wave there
is exactly one near each hydrogenic level, its real part carries the
first post-Newtonian correction, and its width, proportional to
$c^{-(4l+3)}$, is the one found for gravitational atoms by matched
asymptotics. Strong convergence and spectral concentration extend to the
Reissner--Nordstr\"om exterior, extremal or not. At the level of nets the
time-zero Weyl algebra is the same for every $c$, and the vacuum and the
dynamics converge on it, yet for smooth metrics the algebra of each
bounded region with smooth boundary is a type III$_1$ factor at finite
$c$ and a type I factor with a pure product vacuum in the limit. What
collapses is the entanglement of the vacuum; thermal states in flat space
and on stars, with chemical potential shifted by the rest energy, survive
the limit. With positivity of the energy imposed modulo the central
charge, the axioms of Galilean nets hold in the limit for every static
star and for Schwarzschild, and the vacuum is not separating.
\end{abstract}

\maketitle

\section{Introduction}\label{sec:intro}

\subsection{The problem}

A free massive scalar field on a static spacetime has, as $c\to\infty$,
a non-relativistic limit in which the rest energy is removed and the
lapse becomes a Newtonian potential. On a formal level the limit is
standard: expanding the mode frequencies in $c^{-2}$ gives the
Schr\"odinger equation with potential energy $m\Phi$~\cite{Kuchar1980,
Lammerzahl1995,PadmanabhanPadmanabhan2011,
GiuliniGrossardt2012,SchwartzGiulini2019,HartongObersOling2023}, and the
expansion of the Wheeler--DeWitt equation in powers of the Planck mass
is its analogue in quantum gravity~\cite{KieferSingh1991}.
On the Schwarzschild exterior the formal limit hides a change of
spectral type. For every $c$ the Boulware one-particle Hamiltonian has
spectrum $[0,\infty)$, because the effective potential of every partial
wave vanishes at the horizon, and it has no eigenvalues; the limiting
operator, the gravitational hydrogen atom, has infinitely many
eigenvalues. There is
no mode-by-mode correspondence, and the limit, if it exists, cannot be
uniform. This paper proves that it exists, determines its type, and
follows it from the one-particle level to the level of local algebras,
where the companion
papers~\cite{pachon2026algebraicquantumkinematicssrselection,
pachon2026galileanreehschliederobstruction} formulate Galilean quantum
field theory.

\subsection{Results}

The starting point is an exact identity (Lemma~\ref{lem:reduction}).
For a static metric with lapse $N_c$, the operator $\hbar\sqrt{K_c}-
mc^2$ that generates the rest-energy-rotated dynamics is a universal
function $F_c$ of a Schr\"odinger-type operator $S_c$ whose potential
is $m\Phi_c$, $\Phi_c=c^2(N_c^2-1)/2$, and the two resolvents differ by
$O(c^{-2})$ in norm. On the Schwarzschild exterior $\Phi_c=-GM/r$
exactly. Gravity enters through the lapse, which multiplies the mass
term, so $K_c\geq0$ and none of the Krein-space difficulties of the
electrostatic case arise.

\paragraph{Stars.} For static stars whose metric is Lipschitz and of
post-Newtonian form, the rest-energy-subtracted Hamiltonian converges to
$H_S=-\frac{\hbar^2}{2m}\Delta+m\Phi$ in norm-resolvent sense at rate
$c^{-2}$; eigenvalues converge at the same rate, the rescaled two-point
functions converge, and the first post-Newtonian correction is explicit
(Theorem~\ref{thm:stars}, Proposition~\ref{prop:PN}). The remainder of
that correction is $O(c^{-4})$ when the post-Newtonian coefficients of
the metric converge at rate $c^{-2}$, and the simple eigenvalues, and all
levels of spherical stars, are analytic in $c^{-2}$ when the metric
is. For a star of uniform density the
coefficients $c^2(E(c)-E)$ of the numerically computed levels approach
the post-Newtonian prediction at rate $c^{-2}$
(Table~\ref{tab:stars}).

\paragraph{Black holes.} On the Schwarzschild exterior the convergence
is strong and not in norm. The limit selects the Friedrichs
realisation of the gravitational hydrogen atom, and the hydrogenic bound
states arise through spectral concentration of the continuous Boulware
spectrum (Theorem~\ref{thm:schw}, Corollary~\ref{cor:metastability}).
At finite $c$ they decay through the horizon. They are limits of
resonances: in each partial wave there is exactly one near each
hydrogenic level, and its real part and width are those of the
gravitational atoms of the physics literature, with width proportional
to $c^{-(4l+3)}$ (Theorem~\ref{thm:resonances}). The proof matches the
solution that is ingoing at the horizon with the hydrogenic solution
that is regular at the origin, obtains the width exactly from flux
conservation, and shows that the real part is the post-Newtonian
formula of Proposition~\ref{prop:PN} on the Schwarzschild metric. Strong
convergence and spectral concentration extend to the Reissner--Nordstr\"om
exterior, extremal or not (Corollary~\ref{cor:RN}).

\paragraph{Nets.} The time-zero algebras of the Klein--Gordon field and
of its limit are the same Weyl algebra for every $c$, identified through
a local complex field. The vacuum and the dynamics converge on it, but
the local algebras change type: for smooth metrics and bounded regions
with smooth boundary, a type III$_1$ factor at every finite $c$, a type I
factor with a pure product vacuum in the limit (Theorem~\ref{thm:net},
Proposition~\ref{prop:net-static}). This is the precise sense in which
vacuum modular structure collapses. Thermal states in flat space and on
stars survive when the chemical potential includes the rest energy
(Proposition~\ref{prop:thermal}). In the vacuum representation of the
limit nets of flat space, of static stars and of the Schwarzschild
exterior the algebras of open space-time regions are all
$\mathcal{B}(\Hil)$; in flat space this is the
$c\to\infty$ shadow of the timelike tube theorem
(Proposition~\ref{prop:open-regions}).

\paragraph{Axioms.} Imposed as positivity modulo the central charge,
$\hat H+\lambda\hat M\geq0$, the spectral axioms hold for the limit nets
of every static star and of Schwarzschild, and with them all the axioms
of Galilean nets in Weyl form; imposed as $\hat H\geq0$ they fail as
soon as there is a bound state (Proposition~\ref{prop:positivity} and
the discussion after it). The
vacuum of such nets is not separating (Proposition~\ref{prop:nonseparating}).
In the limit theory the short-distance singularity of two-point
functions is a property of the Hamiltonian, not of the state, and the
relevant regularity condition on states is the smoothness of the
one-body density (Proposition~\ref{prop:structure},
Definition~\ref{def:NC-regular}).

\subsection{Relation to previous work}

\paragraph{Rigorous non-relativistic limits.} Titchmarsh and Hunziker
studied the non-relativistic limit of the Dirac spectrum with
electrostatic potentials~\cite{Titchmarsh1962,Hunziker1975}, and
Veseli\'c that of the Klein--Gordon and Dirac spectra, proving
analyticity in $c^{-1}$ and then in $c^{-2}$~\cite{Veselic1970,
Veselic1971,Veselic1983}; see also
Refs.~\onlinecite{GesztesyGrosseThaller1985,ItoYamada2005,Thaller1992}.
Schoene proved convergence of solutions in flat
space~\cite{Schoene1979}, and Cirincione and Chernoff on ultrastatic
backgrounds whose spatial manifold is complete~\cite{CirincioneChernoff1981}.
Hassell, Jia, Sussman and Vasy developed a microlocal calculus for the
non-relativistic limit with time-dependent coefficients and proved
convergence of solutions of the Cauchy problem, at rate $c^{-2}$ on
bounded time intervals, for spacetimes that include horizon-free
astrophysical metrics in which the mass enters as a Newtonian
potential~\cite{HassellJiaSussmanVasy2025a,HassellJiaSussmanVasy2025b}.
They exclude black-hole exteriors, and they pose as open an analogue of
Veseli\'c's analyticity results for the Cauchy problem.
Theorem~\ref{thm:stars} is a spectral statement for static stars, and
Proposition~\ref{prop:PN}(iii) is the gravitational counterpart, for bound
states, of Veseli\'c's analyticity: the simple levels of stars, and all
levels of spherical stars, are analytic in $c^{-2}$ when the metric
is. The analogue for the Cauchy problem remains
open. Without a
horizon, Sussman proved that static bodies bind infinitely many states
of the massive field~\cite{Sussman2024}. On the Schwarzschild exterior
the massive field has no eigenvalues at any $c$ and its scattering is
asymptotically complete~\cite{Bachelot1994}; Theorem~\ref{thm:schw}
treats the limit $c\to\infty$, and we have found no previous result of
this kind for a black-hole exterior.

\paragraph{Gravitational atoms.} Quasi-bound states of massive fields
around black holes, their hydrogenic spectrum and their widths were
computed by Detweiler~\cite{Detweiler1980} and, with Leaver's continued
fraction~\cite{Leaver1985}, by Dolan~\cite{Dolan2007}; they underlie the
superradiance and ultralight-boson literature~\cite{ArvanitakiDubovsky2011,
BaumannChiaPorto2019,BaumannChiaStoutTerHaar2019,BritoCardosoPani2020,
Banerjee2020}. The fine structure of the levels is in
Ref.~\onlinecite{BaumannChiaPorto2019}, and a factor $2$ in Detweiler's
widths was corrected in Ref.~\onlinecite{PaniCardosoGualtieriBertiIshibashi2012}
and confirmed in Ref.~\onlinecite{BaoXuZhang2022}; all these results are
obtained by matched asymptotic expansions and numerics. The mathematical
theory of resonances of black holes is reviewed in
Ref.~\onlinecite{DyatlovZworski2019} (see also
Ref.~\onlinecite{SaBarretoZworski1997}). For massive fields,
exponentially growing modes on Kerr~\cite{ShlapentokhRothman2014}, decay
rates~\cite{PasqualottoShlapentokhRothmanVandeMoortel2023}, and
quasimodes at large angular momentum on Reissner--Nordstr\"om and
Schwarzschild~\cite{ShlapentokhRothmanVandeMoortel2023} are rigorous; for a
uniform analysis of small Schwarzschild--de~Sitter black holes, with the
field mass at the cosmological scale, see Ref.~\onlinecite{HintzXie2022}.
At fixed angular momentum, the regime of gravitational atoms, we know of
no rigorous result that locates the quasi-bound resonances of a massive
field on a black-hole exterior in the limit $\alpha\to0$. Theorems~\ref{thm:schw} and~\ref{thm:resonances}
place these results in the non-relativistic limit.

\paragraph{Algebraic quantum field theory.} The static ground state and
its one-particle structure are due to Kay~\cite{Kay1978}; the Boulware
state~\cite{Boulware1975} is singular on the horizon, since stationary
quasi-free states that are regular across it are thermal at the Hawking
temperature~\cite{KayWald1991}.
Local algebras of the free field are type III$_1$
factors~\cite{Araki1964b,BuchholzDAntoniFredenhagen1987,Verch1997}, and
the timelike tube theorem goes back to Borchers and Araki; see
Refs.~\onlinecite{Borchers1961,StrohmaierWitten2023} and the references
in the latter. Falcone and Conti argued that the non-locality associated
with the Reeh--Schlieder theorem is suppressed under non-relativistic
conditions~\cite{FalconeConti2024}. The Galilean axioms used here are those of
Refs.~\onlinecite{pachon2026algebraicquantumkinematicssrselection,
pachon2026galileanreehschliederobstruction}; the cyclicity of the vacuum
for the algebras of open space-time regions of Galilean theories is due
to Requardt~\cite{Requardt1982}.

\paragraph{What is new.} Lemma~\ref{lem:reduction} is an elementary
identity; at the level of modes the Schr\"odinger form of the static
Klein--Gordon equation is common in the literature on gravitational
atoms, and what the lemma adds is the operator statement with explicit
constants, which turns the problem into perturbation theory of
Schr\"odinger operators. Theorem~\ref{thm:stars}, with its rate, is a
spectral counterpart, for static stars, of the Cauchy-problem results of
Refs.~\onlinecite{HassellJiaSussmanVasy2025a,HassellJiaSussmanVasy2025b};
the post-Newtonian formula of Proposition~\ref{prop:PN} is first-order
perturbation theory for the formal post-Newtonian Hamiltonian, and what
is new is its proof, with a rate, for Lipschitz metrics, and the
analyticity of the levels. Theorem~\ref{thm:schw}, the selection of the
Friedrichs extension and the spectral concentration on the Schwarzschild
exterior appear to be new, and so do parts (b) and (c) of
Theorem~\ref{thm:resonances}, which prove the formulas of the physics
literature for the widths and the fine structure of gravitational atoms
in the non-relativistic limit; part (a) is standard. At the level of nets, the
identification of the time-zero algebras for all $c$ uses an elementary
local complex field; the disjointness of types, the thermal statement and the timelike-tube
interpretation are new, as are positivity modulo the central charge for
the limit nets and the reformulation of the regularity condition.

The paper is organised as follows. Section~\ref{sec:setting} sets up the
one-particle structures and proves the reduction. Sections~\ref{sec:stars}
and~\ref{sec:schw} treat stars and the Schwarzschild exterior,
Sec.~\ref{sec:net} the nets, and Sec.~\ref{sec:axioms} the Galilean
axioms. Section~\ref{sec:discussion} summarises and lists open problems.
The appendices contain the proof of the resolvent bound, the estimates
for the Schwarzschild exterior, the proof of the resonance theorem and
the numerical methods. Figure~\ref{fig:structure} shows how the results depend on each other.

\begin{figure}[htb]
\centering
\begin{tikzpicture}[font=\footnotesize, >={Stealth[length=2mm]},
  box/.style={rounded corners=2pt, align=center, inner sep=3pt,
    minimum height=9mm, text width=36mm,
    execute at begin node={\hyphenpenalty=10000\exhyphenpenalty=10000}},
  hyp/.style={box, draw=gray!70!black, dashed, fill=gray!10},
  der/.style={box, draw, fill=green!12},
  thm/.style={box, draw, line width=0.8pt, fill=blue!12, font=\footnotesize\bfseries},
  arr/.style={->, gray!70!black, line width=0.5pt},
  bus/.style={->, black, line width=0.9pt}]
\def\figXa{2.0}\def\figXb{6.2}\def\figXc{10.75}\def\figXd{14.95}
\def\figYa{0}\def\figYb{-1.55}\def\figYc{-3.1}\def\figYd{-4.65}\def\figYe{-6.2}\def\figYf{-7.75}
\node[der] (l1) at (\figXa,\figYa)
  {Lemma~\ref{lem:reduction}: $P_c=F_c(S_c)$, resolvents at distance
   $O(c^{-2})$};
\node[hyp] (hs) at (\figXa,\figYb)
  {Static star: $\Sigma_c=\R^3$, metric of post-Newtonian form (S1)};
\node[thm] (t1) at (\figXa,\figYc)
  {Theorem~\ref{thm:stars}\\ \mdseries norm resolvent, rate $c^{-2}$};
\node[thm] (p1) at (\figXa,\figYd)
  {Proposition~\ref{prop:PN}\\ \mdseries post-Newtonian levels, analytic
   in $c^{-2}$; Table~\ref{tab:stars}};
\node[hyp] (hb) at (\figXb,\figYa)
  {Schwarzschild or Reissner--Nordstr\"om exterior (extremal or not),
   Boulware state};
\node[der] (b13) at (\figXb,\figYb)
  {Lemmas~\ref{lem:partial-waves}--\ref{lem:approximants}: partial
   waves, graph limit, approximants};
\node[thm] (t2) at (\figXb,\figYc)
  {Theorem~\ref{thm:schw}, Corollaries~\ref{cor:metastability}
   and~\ref{cor:RN}\\ \mdseries strong, not norm; Table~\ref{tab:schw}};
\node[der] (b47) at (\figXb,\figYd)
  {Lemmas~\ref{lem:continuation}--\ref{lem:far}: Jost solutions, near
   zone, matching, far zone};
\node[thm] (t3) at (\figXb,\figYe)
  {Theorem~\ref{thm:resonances} (Schwarzschild)\\ \mdseries resonances;
   Table~\ref{tab:atoms}};
\node[der] (aqft) at (\figXc,0.15)
  {Type III$_1$ (Araki, BDF, Verch); timelike tube (Borchers,
   Strohmaier--Witten); time slice (Dimock)};
\node[thm] (t4) at (\figXc,\figYb)
  {Theorem~\ref{thm:net}\\ \mdseries Minkowski: states, dynamics,
   III$_1$ vs.\ I, cylinder};
\node[thm] (p2) at (\figXc,\figYc)
  {Proposition~\ref{prop:net-static}\\ \mdseries static backgrounds;
   states at rate $c^{-2}$};
\node[thm] (p3) at (\figXc,\figYd)
  {Proposition~\ref{prop:open-regions}\\ \mdseries open regions:
   $\mathcal{B}(\Hil)$};
\node[thm] (p4) at (\figXc,\figYe)
  {Proposition~\ref{prop:thermal}\\ \mdseries thermal states in flat
   space and on stars};
\node[der] (pI) at (\figXd,\figYa)
  {Ref.~\onlinecite{pachon2026algebraicquantumkinematicssrselection},
   Remark~4: the shift by $\lambda\hat M$ keeps the Bargmann relations};
\node[thm] (p5) at (\figXd,\figYb)
  {Proposition~\ref{prop:positivity}\\ \mdseries positivity modulo the
   central charge};
\node[der] (pII) at (\figXd,\figYc)
  {Ref.~\onlinecite{pachon2026galileanreehschliederobstruction},
   Theorem~1: Fock rigidity};
\node[thm] (p6) at (\figXd,\figYd)
  {Proposition~\ref{prop:nonseparating}\\ \mdseries vacuum not
   separating};
\node[der] (a7) at (\figXd,\figYe)
  {Ref.~\onlinecite{pachon2026algebraicquantumkinematicssrselection},
   Lemma~A7: unique continuation};
\node[thm] (p7) at (\figXd,\figYf)
  {Proposition~\ref{prop:structure}, Definition~\ref{def:NC-regular},
   Remark~\ref{rem:coulomb}\\ \mdseries two-point functions};
\draw[arr] (hs) -- (t1);
\draw[arr] (t1) -- (p1);
\draw[arr] (hb) -- (b13);
\draw[arr] (b13) -- (t2);
\draw[arr] (b47) -- (t3);
\draw[arr, dashed] (p1.south) |- node[pos=0.75, below, black]
  {$D_1$} (t3.west);
\draw[arr] (aqft) -- (t4);
\draw[arr] (t4) -- (p2);
\draw[arr] (pI) -- (p5);
\draw[arr] (pII) -- (p6);
\draw[arr, dashed] (p5.east) -- ++(0.3,0) |- (p6.east);
\draw[arr] (a7.west) -- (p3.south east);
\def\figBus{8.475}\def\figBusQ{8.29}\def\figBusY{-3.87}
\draw[bus, -] ([xshift=6mm]t1.south) |- (\figBus,\figBusY);
\draw[bus, -] (t2.east) -| (\figBus,\figBusY);
\draw[bus] (\figBus,\figBusY) |- (p2.west);
\draw[bus] (\figBus,\figBusY) |- (p3.west);
\draw[bus] (\figBusQ,\figBusY) |- (p4.west);
\end{tikzpicture}
\caption{Main dependencies between the results. Dashed boxes:
hypotheses; green: lemmas and results used from elsewhere; bold boxes:
results. Lemma~\ref{lem:reduction} enters every one-particle result and
Propositions~\ref{prop:net-static} and~\ref{prop:thermal}, and
Lemma~\ref{lem:partial-waves} also enters Theorem~\ref{thm:resonances}
and Lemmas~\ref{lem:continuation} and~\ref{lem:far}. The thick lines
carry the operators $H_S$ of Theorems~\ref{thm:stars}
and~\ref{thm:schw}, with their convergence, to
Proposition~\ref{prop:net-static}, the operators alone to
Proposition~\ref{prop:open-regions}, and the norm-resolvent convergence
of Theorem~\ref{thm:stars} to Proposition~\ref{prop:thermal}, which has
no Schwarzschild case. Apart from Lemma~\ref{lem:reduction}, they are
the only link from the one-particle results to the nets, and
Theorem~\ref{thm:net} uses no one-particle result.
Corollary~\ref{cor:RN} goes the other way: it repeats the proofs of
Theorem~\ref{thm:schw}, Corollary~\ref{cor:metastability} and
Proposition~\ref{prop:net-static} on the Reissner--Nordstr\"om
exterior. Proposition~\ref{prop:thermal} also uses the proofs of
Theorem~\ref{thm:net} and Proposition~\ref{prop:net-static}, and the
convergence of the thermal modular groups after it the proof of
Theorem~\ref{thm:net}(b). Dashed
arrows: the real part in Theorem~\ref{thm:resonances} is the functional
$D_1$ of Proposition~\ref{prop:PN} on the Schwarzschild metric, and
Proposition~\ref{prop:positivity}(ii) supplies hypotheses of
Proposition~\ref{prop:nonseparating} for the limit nets.}
\label{fig:structure}
\end{figure}

\section{Setting}\label{sec:setting}

\subsection{Static spacetimes and their one-particle structures}
\label{sec:one-particle}

Throughout, $\epsilon:=c^{-2}$. We consider static metrics
\begin{equation}\label{eq:static-metric}
g_c=-N_c^2\,c^2\,\rmd t^2+h_{c,ij}\,\rmd x^i\rmd x^j
\qquad\text{on }\R\times\Sigma_c,\quad \Sigma_c\subset\R^3\text{ open},
\end{equation}
with lapse $N_c>0$ on $\Sigma_c$ and Riemannian metric $h_c$. For the
Klein--Gordon field of mass $m>0$ put
\begin{equation}\label{eq:Kc}
K_c:=c^2\Bigl[A_c+N_c^2\frac{m^2c^2}{\hbar^2}\Bigr],\qquad
A_c:=-\frac{N_c}{\sqrt{|h_c|}}\,\partial_i\Bigl(N_c\sqrt{|h_c|}\,
h_c^{ij}\,\partial_j\,\cdot\,\Bigr),
\end{equation}
so that $(\Box_{g_c}-m^2c^2/\hbar^2)\varphi=0$ reads $\partial_t^2\varphi+
K_c\varphi=0$. The operator $A_c$ is symmetric and non-negative on
$C_c^\infty(\Sigma_c)$ in
\begin{equation}\label{eq:Kc-space}
\K_c:=L^2(\Sigma_c,\mu_c),\qquad \rmd\mu_c:=\rho_c\,\rmd^3x,\qquad
\rho_c:=N_c^{-1}\sqrt{|h_c|},
\end{equation}
since $\braket{u,A_cu}_{\K_c}=\int N_c\sqrt{|h_c|}\,h_c^{ij}\,
\partial_i\bar u\,\partial_ju\,\rmd^3x$. In the two classes of
spacetimes treated below, $K_c$ is essentially self-adjoint on
$C_c^\infty(\Sigma_c)$ (Secs.~\ref{sec:stars} and~\ref{sec:schw}), and we
denote its closure by the same symbol. Since $A_c\geq0$,
\begin{equation}\label{eq:K-lower}
K_c\;\geq\;\frac{m^2c^4}{\hbar^2}\,N_c^2
\end{equation}
in the sense of forms.

For time-zero Cauchy data write $\varphi(u)$ and $\pi(v)$, $u,v$ real in
$C_c^\infty(\Sigma_c)$, for the field and $\pi=\partial_t\varphi$ smeared
against $\mu_c$. We normalise the action as $\frac12\int((\partial_t
\varphi)^2-\varphi\,K_c\varphi)\,\rmd\mu_c\,\rmd t$, so that $\pi$ is the
canonical momentum and $[\varphi(u),\pi(v)]=\rmi\hbar\braket{u,v}_{\K_c}$.
The \emph{static ground state} $\omega_c$ is the quasi-free state with
\begin{equation}\label{eq:ground-state}
\omega_c\bigl(\varphi(u)^2\bigr)=\tfrac{\hbar}{2}\braket{u,K_c^{-1/2}u},
\qquad
\omega_c\bigl(\pi(v)^2\bigr)=\tfrac{\hbar}{2}\braket{v,K_c^{1/2}v},
\qquad
\omega_c\bigl(\varphi(u)\pi(v)\bigr)=\tfrac{\rmi\hbar}{2}\braket{u,v},
\end{equation}
all inner products in $\K_c$. It exists whenever these quantities are
finite for test data, which by~\eqref{eq:K-lower} holds for $u$ with
compact support in $\Sigma_c$, on which $N_c^{-1}$ is bounded:
$\braket{u,K_c^{-1/2}u}\leq(\hbar/mc^2)\braket{u,N_c^{-1}u}$.
Its GNS representation is the Fock representation over the one-particle
space $\K_c$, with
\[
\varphi(u)=\sqrt{\hbar/2}\,\bigl(a(K_c^{-1/4}u)+a^*(K_c^{-1/4}u)\bigr),\qquad
\pi(v)=\rmi\sqrt{\hbar/2}\,\bigl(a^*(K_c^{1/4}v)-a(K_c^{1/4}v)\bigr),
\]
and one-particle Hamiltonian $\hbar\sqrt{K_c}$~\cite{Kay1978}. On a
star (Sec.~\ref{sec:stars}) this is the ground state of the static
Killing flow; on the Schwarzschild exterior it is the Boulware
state~\cite{Boulware1975}. It is singular on the horizon: a stationary
quasi-free state that is regular across the horizon is thermal at the
Hawking temperature~\cite{KayWald1991}, and $\omega_c$ is a ground state.

\subsection{Rest-energy subtraction and the rescaled field}
\label{sec:rescaled-field}

Put
\begin{equation}\label{eq:Pc-Sc}
P_c:=\hbar\sqrt{K_c}-mc^2,\qquad
S_c:=\frac{\hbar^2K_c-m^2c^4}{2mc^2},\qquad
\gamma_c:=\Bigl(\frac{mc^2}{\hbar\sqrt{K_c}}\Bigr)^{1/2}
=\Bigl(1+\frac{P_c}{mc^2}\Bigr)^{-1/2}.
\end{equation}
The number operator $\hat N=\rmd\Gamma(\one)$ of the free field commutes
with the dynamics, and with $\hat M:=m\hat N$ the relativistic
Hamiltonian splits as
\begin{equation}\label{eq:rest-energy-split}
\rmd\Gamma(\hbar\sqrt{K_c})=\rmd\Gamma(P_c)+c^2\hat M .
\end{equation}
The \emph{rescaled field} is the rest-energy-rotated positive-frequency
part of $\varphi$. We smear it antilinearly, as the time-zero fields of
Ref.~\onlinecite{pachon2026algebraicquantumkinematicssrselection}; in
the one-particle language,
\begin{equation}\label{eq:rescaled-field}
\hat\psi_c(t,u):=\sqrt{\tfrac{2mc^2}{\hbar}}\;\rme^{\rmi mc^2t/\hbar}
\varphi^+(t,\bar u)=\sqrt\hbar\;a\bigl(\gamma_c\,\rme^{\rmi tP_c/\hbar}\,u\bigr),
\end{equation}
where $\varphi^+(t,w)=\sqrt{\hbar/2}\,a(\rme^{\rmi t\sqrt{K_c}}
K_c^{-1/4}\bar w)$ is the annihilation part of $\varphi(t,w)$ and $a$ is
antilinear in its argument. It satisfies $\hat\psi_c(t,u)\Omega_c=0$ and
\[
\bigl[\hat\psi_c(t,u),\hat\psi_c(t,v)^*\bigr]=\hbar\,
\braket{\gamma_cu,\gamma_cv}_{\K_c},
\]
For $u=\mathcal{U}_c^*f$ and $v=\mathcal{U}_c^*g$, with $f,g\in
L^2(\R^3)$ fixed (and supported in $\{r>R_0\}$ on the Schwarzschild
exterior), this tends to the canonical value $\hbar\braket{f,g}$ as
$c\to\infty$, because $\mathcal{U}_c\gamma_c\mathcal{U}_c^*\to\one$
strongly (Theorem~\ref{thm:stars}(b) and the proof of
Theorem~\ref{thm:schw}(d)). At finite $c$ the
field
$\hat\psi_c$ is not local, because $\gamma_c$ is a non-local operator
(a non-local complex field of this kind is used for the non-relativistic
limit of scalar field theories in Ref.~\onlinecite{NamjooGuthKaiser2018});
the time-zero field that is both local and canonical at every $c$ is
$\hat\psi_{0,c}$ of Sec.~\ref{sec:net}, and the two have the same limit
(Theorem~\ref{thm:net} and Proposition~\ref{prop:net-static}).
The limiting object is the free Schr\"odinger field on the Fock space
$\Gamma(L^2(\R^3))$ with one-particle Hamiltonian $H_S$,
\begin{equation}\label{eq:limit-field}
\hat\psi(t,u)=\sqrt\hbar\;a\bigl(\rme^{\rmi tH_S/\hbar}\,u\bigr),\qquad
[\hat\psi(t,u),\hat\psi(t,v)^*]=\hbar\braket{u,v},
\end{equation}
which carries the Bargmann charge $\hat M=m\hat N$ in the sense of
Ref.~\onlinecite{pachon2026algebraicquantumkinematicssrselection}
(Definition~3 there): $\rme^{\rmi\theta\hat M/\hbar}\hat\psi(t,u)
\rme^{-\rmi\theta\hat M/\hbar}=\rme^{-\rmi\theta m/\hbar}\hat\psi(t,u)$.

To compare one-particle spaces we use the density map
\begin{equation}\label{eq:Uc}
\mathcal{U}_c:\K_c\to L^2(\R^3),\qquad (\mathcal{U}_cu)(x):=
\rho_c(x)^{1/2}u(x)\ (x\in\Sigma_c),\qquad 0\ (x\notin\Sigma_c).
\end{equation}
It is an isometry onto $L^2(\Sigma_c)$, unitary when $\Sigma_c=\R^3$.
For $f\in C_c^\infty(\R\times\R^3)$ with $\supp f(t,\cdot)\subset
\Sigma_c$ we smear with
\begin{equation}\label{eq:smeared-vectors}
\hat\psi_c(f):=\sqrt\hbar\;a\bigl(F^{(c)}_f\bigr),\quad
F^{(c)}_f:=\int\rmd t\;\gamma_c\,\rme^{\rmi tP_c/\hbar}\,
\mathcal{U}_c^*f(t,\cdot),\qquad
F^{(\infty)}_f:=\int\rmd t\;\rme^{\rmi tH_S/\hbar}\,f(t,\cdot),
\end{equation}
and $\hat\psi(f):=\sqrt\hbar\,a(F^{(\infty)}_f)$. Then
\begin{equation}\label{eq:two-point}
\omega_c\bigl(\hat\psi_c(f)\hat\psi_c(g)^*\bigr)=\hbar\braket{F_f^{(c)},
F_g^{(c)}}_{\K_c},\qquad
\omega_\infty\bigl(\hat\psi(f)\hat\psi(g)^*\bigr)=\hbar\braket{F_f^{(\infty)},
F_g^{(\infty)}},
\end{equation}
and all other two-point functions of these fields vanish in the
respective vacua. Convergence of the rescaled theory therefore reduces
to
\begin{equation}\label{eq:H2}
\mathcal{U}_cF_f^{(c)}\longrightarrow F_f^{(\infty)}\quad\text{in }
L^2(\R^3),
\end{equation}
which Theorems~\ref{thm:stars} and~\ref{thm:schw} establish for stars
and for the Schwarzschild exterior.

\subsection{Exact reduction to a Schr\"odinger-type operator}
\label{sec:reduction}

\begin{lemma}[Exact reduction]\label{lem:reduction}
In the setting of Sec.~\ref{sec:one-particle},
\begin{equation}\label{eq:Sc}
S_c=\frac{\hbar^2}{2m}A_c+m\Phi_c,\qquad
\Phi_c:=\frac{c^2}{2}\bigl(N_c^2-1\bigr),\qquad
S_c\geq-\frac{mc^2}{2}.
\end{equation}
Moreover $P_c=F_c(S_c)$, i.e.\ $S_c=P_c+P_c^2/(2mc^2)$, with
\begin{equation}\label{eq:Fc}
F_c(s):=mc^2\Bigl(\sqrt{1+\tfrac{2s}{mc^2}}-1\Bigr),\qquad
s\geq-\tfrac{mc^2}{2},
\end{equation}
so that $F_c(s)=s-s^2/(2mc^2)+O(c^{-4})$ for fixed $s$, and for every $z\in\C\setminus\R$ and every $c$ with $mc^2\geq8|z|$,
\begin{equation}\label{eq:reduction-bound}
\bigl\|(P_c-z)^{-1}-(S_c-z)^{-1}\bigr\|\leq
\frac{1}{mc^2}\Bigl(\frac{16}{3}+\frac{2|z|^2}{|\mathrm{Im}\,z|^2}\Bigr).
\end{equation}
On the Schwarzschild exterior, $N_c^2=1-2GM/(c^2r)$, one has
$\Phi_c=-GM/r$ exactly, for every $c$.
\end{lemma}

\begin{proof}
$\hbar^2K_c-m^2c^4=\hbar^2c^2A_c+m^2c^4(N_c^2-1)$; dividing by $2mc^2$
gives~\eqref{eq:Sc}, and $S_c\geq-mc^2/2$ because $A_c\geq0$ and
$N_c^2\geq0$. Since $\hbar\sqrt{K_c}=(m^2c^4+2mc^2S_c)^{1/2}$ by the
spectral theorem, $P_c=F_c(S_c)$. The difference of resolvents is
$d(S_c)$ with $d(s)=[s-F_c(s)]/[(F_c(s)-z)(s-z)]$, and
$\sup_{s\geq-mc^2/2}|d(s)|$ is bounded in Appendix~\ref{app:reduction}.
The last statement is $c^2(N_c^2-1)/2=-GM/r$.
\end{proof}

The lapse multiplies the rest-mass term in~\eqref{eq:Kc}, so gravity
enters the square root of $K_c$ as a potential of $S_c$ while $K_c$
stays non-negative. An electrostatic potential instead enters as
$(\rmi\hbar\partial_t-eA_0)^2$: the field equation is quadratic in
$\partial_t$ with a first-order term, its conserved charge form is
indefinite and the dynamics lives in a Krein space, and for strong
potentials the energy is indefinite as well, which is the Klein
paradox~\cite{Veselic1971,Veselic1983}. The non-relativistic limit is thus
a question about the Schr\"odinger-type operators $S_c$, whose potential
$m\Phi_c$ is the exact Newtonian potential energy of the Schwarzschild
exterior and differs from it by $O(c^{-2})$ inside a star. The relativistic
kinematics sits in the universal function $F_c$, whose second term gives
the kinematic part of the first post-Newtonian correction
(Proposition~\ref{prop:PN}).

A statement about $P_c$ is a statement about the Boulware or static
one-particle Hamiltonian $\hbar\sqrt{K_c}$ shifted by $mc^2$, and
Lemma~\ref{lem:reduction} transfers resolvent estimates from $S_c$ to
$P_c$ at the cost $O(c^{-2})$ in norm. An exact transfer is also
available. Since $p+p^2/2mc^2-s=(p-z)(p+z+2mc^2)/2mc^2$ for
$s:=z+z^2/2mc^2$, partial fractions give, for $\mathrm{Im}\,z>0$,
\begin{equation}\label{eq:P-S-resolvent}
(P_c-z)^{-1}=\Bigl(1+\frac{z}{mc^2}\Bigr)(S_c-s)^{-1}+(P_c+z+2mc^2)^{-1},
\qquad s=z+\frac{z^2}{2mc^2},
\end{equation}
and the last term is analytic in $z$ off $(-\infty,-mc^2]$. Spectral
statements about $S_c$ at $s$ are thus statements about $P_c$ at $z$
(Sec.~\ref{sec:schw}).

\section{Regular static stars}\label{sec:stars}

\subsection{Norm-resolvent convergence}

A static star has no horizon: the lapse is bounded below and
$\Sigma_c=\R^3$. We parametrise the metric by its Newtonian potential
$\Phi$ and $c$-dependent remainders.

\begin{theorem}[Regular static stars]\label{thm:stars}
Let $\Sigma_c=\R^3$ and assume
\begin{enumerate}[label=(S\arabic*)]
\item $N_c^2=1+2\epsilon\Phi+\epsilon^2\varrho_c$ and $h_c=\delta+\epsilon
k_c$, where $\Phi\in W^{1,\infty}(\R^3)$ is real with $\Phi(x)\to0$ as
$|x|\to\infty$, and the functions $\varrho_c$ and the symmetric
matrices $k_c$ are bounded in $W^{1,\infty}(\R^3)$ uniformly in
$c\geq c_0$.
\end{enumerate}
Let $H_S:=-\frac{\hbar^2}{2m}\Delta+m\Phi$ on $H^2(\R^3)$. There is
$c_1$ such that for $c\geq c_1$ the lapse satisfies $N_c^2\geq\frac12$,
$K_c$ is essentially self-adjoint on $C_c^\infty(\R^3)$, the static
ground state $\omega_c$ exists, and for every $z\in\C\setminus\R$
\begin{equation}\label{eq:norm-resolvent}
\bigl\|\,\mathcal{U}_c(P_c-z)^{-1}\mathcal{U}_c^*-(H_S-z)^{-1}\bigr\|
\leq C_z\,c^{-2}.
\end{equation}
Consequently:
\begin{enumerate}[label=(\alph*)]
\item every isolated eigenvalue $E$ of $H_S$, of multiplicity $\mu$, is
the limit of exactly $\mu$ eigenvalues of $P_c$, counted with
multiplicity, each at distance $O(c^{-2})$ from $E$, and the
corresponding spectral projections converge in norm;
\item $\mathcal{U}_c\rme^{-\rmi tP_c/\hbar}\mathcal{U}_c^*\to
\rme^{-\rmi tH_S/\hbar}$ strongly, uniformly for $t$ in compact sets, and
$\mathcal{U}_c\gamma_c\mathcal{U}_c^*\to\one$ strongly. Hence
\eqref{eq:H2} holds for every $f\in C_c^\infty(\R\times\R^3)$, and the
two-point functions~\eqref{eq:two-point} of $\hat\psi_c$ in $\omega_c$
converge to those of $\hat\psi$ in the Fock vacuum;
\item if $g_c$ is smooth, $\omega_c$ is a Hadamard
state~\cite{SahlmannVerch2000}.
\end{enumerate}
\end{theorem}

The flat case $\Phi=\varrho_c=0$, $k_c=0$ is included: then
$S_c=-\frac{\hbar^2}{2m}\Delta$ for every $c$, and
\eqref{eq:norm-resolvent} is Lemma~\ref{lem:reduction} alone.

\begin{proof}
\emph{Uniform bounds.} By (S1), $\sup|N_c^2-1|\leq2\epsilon\|\Phi\|_\infty+
\epsilon^2\sup_c\|\varrho_c\|_\infty$ and $\sup|h_c-\delta|\leq\epsilon
\sup_c\|k_c\|_\infty$, so for $c\geq c_1$ we have $\frac12\leq N_c^2\leq2$ and
$\frac12\delta\leq h_c\leq\frac32\delta$. Then $\rho_c=N_c^{-1}
\sqrt{|h_c|}$ is bounded above and below, and
\[
\nabla\ln\rho_c=-\frac{2\epsilon\nabla\Phi+\epsilon^2\nabla\varrho_c}
{2N_c^2}+\frac{\epsilon}2\operatorname{tr}\bigl(h_c^{-1}\nabla k_c\bigr)
=O(\epsilon)\quad\text{in }L^\infty .
\]

\emph{The form of $\mathcal{U}_cS_c\mathcal{U}_c^*$.} Writing $u=
\rho_c^{-1/2}v$, the form $\braket{u,S_cu}_{\K_c}$ of
Lemma~\ref{lem:reduction} becomes, on $v\in H^1(\R^3)$,
\begin{equation}\label{eq:qc}
q_c(v)=\frac{\hbar^2}{2m}\int N_c^2h_c^{ij}\,
\overline{\bigl(\partial_iv-\tfrac12v\,\partial_i\ln\rho_c\bigr)}
\bigl(\partial_jv-\tfrac12v\,\partial_j\ln\rho_c\bigr)\,\rmd^3x
+\int m\Phi_c|v|^2\,\rmd^3x ,
\end{equation}
because $N_c\sqrt{|h_c|}/\rho_c=N_c^2$. The coefficients are bounded and
uniformly elliptic, so $q_c$ is closed on $H^1(\R^3)$. On $\K_c$ the form
of $S_c$ is $\frac{\hbar^2}{2m}\int N_c\sqrt{|h_c|}\,h_c^{ij}\,\partial_i
\bar u\,\partial_ju\,\rmd^3x+\int m\Phi_c\rho_c|u|^2\,\rmd^3x$ on
$H^1(\R^3)$, and its operator is $-\frac{\hbar^2}{2m}\rho_c^{-1}\partial_i
(N_c\sqrt{|h_c|}\,h_c^{ij}\partial_ju)+m\Phi_cu$. Its principal
coefficients are Lipschitz and no lower-order term stands inside the
divergence, so by elliptic regularity~\cite{GilbargTrudinger2001}
(Theorem~8.8) its domain is $H^2(\R^3)$, in which $C_c^\infty(\R^3)$ is
dense in the graph norm; the interior estimate of that theorem holds on
unit balls with constants independent of the centre, since the
ellipticity and Lipschitz bounds are uniform, and summing it over a
covering by unit balls gives the global statement. Hence $K_c=(2mc^2S_c+m^2c^4)/\hbar^2$ is
essentially self-adjoint on $C_c^\infty(\R^3)$, and $\mathcal{U}_cS_c
\mathcal{U}_c^*$, which is unitarily equivalent to $S_c$, is the operator
of $q_c$. Its domain $\rho_c^{1/2}H^2(\R^3)$ differs from $H^2(\R^3)$ when
$\nabla\rho_c$ jumps, as at the surface of the star of
Sec.~\ref{sec:uniform-star}. Since $\hbar\sqrt{K_c}\geq N_c\,mc^2\geq mc^2/\sqrt2$
by~\eqref{eq:K-lower}, $\omega_c$ exists and $0\leq\gamma_c\leq2^{1/4}$.

\emph{Comparison of forms.} Let $q(v):=\frac{\hbar^2}{2m}\|\nabla v\|^2+
\int m\Phi|v|^2$ be the form of $H_S$. By (S1),
$N_c^2h_c^{ij}-\delta^{ij}=O(\epsilon)$ and $m\Phi_c-m\Phi=\frac m2
\epsilon\varrho_c=O(\epsilon)$ in $L^\infty$. Together with the bound
on $\nabla\ln\rho_c$ this gives $|q_c(v)-q(v)|\leq C_1\epsilon\,
(\|\nabla v\|^2+\|v\|^2)$. With $\lambda_0:=m\|\Phi\|_\infty+1$ one has
$\|\nabla v\|^2+\|v\|^2\leq C_2\,(q(v)+\lambda_0\|v\|^2)$, hence
\begin{equation}\label{eq:form-bound}
|q_c(v)-q(v)|\leq\delta_c\,\bigl(q(v)+\lambda_0\|v\|^2\bigr),\qquad
\delta_c:=C_1C_2\,\epsilon .
\end{equation}

\emph{Resolvents.} Put $Q:=H_S+\lambda_0\geq1$ and $Q_c:=\mathcal{U}_c
S_c\mathcal{U}_c^*+\lambda_0$. By~\eqref{eq:form-bound} the form
$w\mapsto(q_c-q)(Q^{-1/2}w)$ is bounded, with operator $B_c$ of norm at
most $\delta_c$, and $Q_c=Q^{1/2}(\one+B_c)Q^{1/2}$ as operators
associated with forms. For $\delta_c<1$,
\begin{equation}\label{eq:factorised-resolvent}
Q_c^{-1}-Q^{-1}=Q^{-1/2}\bigl[(\one+B_c)^{-1}-\one\bigr]Q^{-1/2},
\qquad
\bigl\|Q^{1/2}\bigl(Q_c^{-1}-Q^{-1}\bigr)\bigr\|\leq\frac{\delta_c}{1-\delta_c}.
\end{equation}
For $z\in\C\setminus\R$ the second resolvent identity,
$(A-z)^{-1}-(B-z)^{-1}=(A-z)^{-1}(A+\lambda_0)\,[(A+\lambda_0)^{-1}-(B+
\lambda_0)^{-1}]\,(B+\lambda_0)(B-z)^{-1}$, and $\|(A+\lambda_0)(A-z)^{-1}\|
\leq1+|z+\lambda_0|/|\mathrm{Im}\,z|$ for self-adjoint $A$, give
$\|(\mathcal{U}_cS_c\mathcal{U}_c^*-z)^{-1}-(H_S-z)^{-1}\|\leq C_z'\,
\epsilon$. Adding~\eqref{eq:reduction-bound}, conjugated by the unitary
$\mathcal{U}_c$, proves~\eqref{eq:norm-resolvent}.

\emph{(a)} Norm-resolvent convergence with rate $c^{-2}$ at one
non-real point gives the same rate uniformly on compact subsets of the
resolvent set of $H_S$. If $\Gamma$ is a circle about $E$ that encloses
no other point of $\sigma(H_S)$, the Riesz projections of
$\tilde P_c:=\mathcal{U}_cP_c\mathcal{U}_c^*$ and $H_S$ for $\Gamma$ differ
by $O(c^{-2})$ in norm, so they have the same rank $\mu$ for $c$ large.
Let $z_0=E+\rmi\eta$, with $\eta$ larger than the radius of $\Gamma$.
The normal operators $(\tilde P_c-z_0)^{-1}$ and $(H_S-z_0)^{-1}$ differ
by $O(c^{-2})$, and the spectrum of a normal operator moves by at most
the norm of such a perturbation, so $(E_j(c)-z_0)^{-1}$ lies within
$O(c^{-2})$ of $\sigma((H_S-z_0)^{-1})$ for each eigenvalue $E_j(c)$ of
$\tilde P_c$ inside $\Gamma$. The map $w\mapsto(w-z_0)^{-1}$ sends the
closed disc bounded by $\Gamma$ and the closed set $(\sigma(H_S)\setminus
\{E\})\cup\{\infty\}$ to disjoint compact sets, so for $c$ large the
nearest point is $(E-z_0)^{-1}$. Hence $|(E_j(c)-z_0)^{-1}-(E-z_0)^{-1}|=
O(c^{-2})$, that is, $|E_j(c)-E|=O(c^{-2})$.

\emph{(b)} Norm-resolvent convergence implies strong convergence of the
unitary groups, uniformly for $t$ in compact sets~\cite{ReedSimon1972}.
For $\gamma_c$ write $\mathcal{U}_c\gamma_c\mathcal{U}_c^*-\one=
\eta_c(\tilde P_c)$ with $\eta_c(p)=(1+p/mc^2)^{-1/2}-1$, which is
bounded by $1+2^{1/4}$ on $\sigma(\tilde P_c)$ and by $2\Lambda/mc^2$ on
$[-\Lambda,\Lambda]$ once $mc^2\geq4\Lambda$. For $v\in L^2(\R^3)$ and
$\pm\Lambda\notin\sigma_{\rm p}(H_S)$,
\[
\|\eta_c(\tilde P_c)v\|\leq\frac{2\Lambda}{mc^2}\|v\|+(1+2^{1/4})\,
\|\one_{\{|\tilde P_c|>\Lambda\}}v\|
\;\xrightarrow{c\to\infty}\;(1+2^{1/4})\,\|\one_{\{|H_S|>\Lambda\}}v\|,
\]
since spectral projections of intervals whose end points are not
eigenvalues converge strongly. Letting $\Lambda\to\infty$ gives
$\mathcal{U}_c\gamma_c\mathcal{U}_c^*\to\one$. The operators
$\mathcal{U}_c\gamma_c\rme^{\rmi tP_c/\hbar}\mathcal{U}_c^*$ are bounded
uniformly in $c$ and converge strongly, uniformly for $t$ in the support
of $f$, and $t\mapsto f(t,\cdot)$ is continuous into $L^2$;
\eqref{eq:H2} follows by dominated convergence, and the two-point
functions converge by~\eqref{eq:two-point}.

\emph{(c)} For smooth $g_c$ the spacetime is static with a Killing
field that is timelike everywhere, since $N_c^2\geq\frac12$, and it is
globally hyperbolic, because $N_c$ is bounded above and below and $h_c$
is uniformly equivalent to $\delta$, hence complete; $\omega_c$ is its
quasi-free ground state. For static spacetimes whose lapse is bounded
away from zero this state is Hadamard by
Ref.~\onlinecite{FullingNarcowichWald1981}; more generally, ground states
and KMS states of stationary spacetimes satisfy the microlocal spectrum
condition~\cite{SahlmannVerch2000}, which characterises Hadamard
states~\cite{Radzikowski1996}.
\end{proof}

The theorem uses the Newtonian potential only through
$\Phi\in W^{1,\infty}$ and its decay: the star is arbitrary, and so is its
equation of state, provided the metric is Lipschitz and has the
post-Newtonian form (S1). The rate $c^{-2}$ is sharp in general, as the
next proposition shows.

\subsection{First post-Newtonian correction}

\begin{proposition}[First post-Newtonian correction]\label{prop:PN}
Assume (S1), and in addition $\varrho_c\to\varrho$ and $k_c\to k$ in
$L^\infty$. Let $E$ be a simple isolated eigenvalue of $H_S$ with real
normalised eigenfunction $u$, and $E(c)$ the eigenvalue of $P_c$ near
$E$.
\begin{enumerate}[label=(\roman*)]
\item One has
\begin{equation}\label{eq:PN}
E(c)=E+\frac{1}{c^2}\Bigl(D_1(u)-\frac{E^2}{2m}\Bigr)+o(c^{-2}),
\end{equation}
where
\begin{align*}
D_1(u)={}&\frac{\hbar^2}{2m}\int\Bigl[\bigl(\Phi+\tfrac12\operatorname{tr}k\bigr)
|\nabla u|^2-k^{ij}\partial_iu\,\partial_ju\Bigr]\rmd^3x
+\int\Bigl[\tfrac12m\varrho-m\Phi^2+\tfrac12m\Phi\operatorname{tr}k
\Bigr]u^2\,\rmd^3x\\
&-E\int\bigl(\tfrac12\operatorname{tr}k-\Phi\bigr)u^2\,\rmd^3x .
\end{align*}
\item If moreover $\|\varrho_c-\varrho\|_\infty+\|k_c-k\|_\infty=
O(c^{-2})$, the remainder in~\eqref{eq:PN} is $O(c^{-4})$.
\item If $\epsilon=c^{-2}\mapsto(\varrho_c,k_c)$ extends holomorphically
to a complex neighbourhood of $\epsilon=0$ with values in
$W^{1,\infty}(\R^3)$, then for $c$ large $E(c)$ is the sum of a
convergent power series in $c^{-2}$.
\end{enumerate}
For rotation-invariant metrics the statement applies in each
angular-momentum sector, where the eigenvalues of $H_S$ are simple.
\end{proposition}

\begin{proof}
Work in $\K_c$, whose elements are functions on $\R^3$ with the inner
product $\braket{v,w}_\epsilon=\int\rho_c\bar vw$. The form of $S_c$ is
$\tilde q_\epsilon(v,w)=\frac{\hbar^2}{2m}\int a^{ij}_\epsilon\partial_i
\bar v\,\partial_jw+\int b_\epsilon\bar vw$ on $H^1(\R^3)$, with
$a_\epsilon=N_c\sqrt{|h_c|}\,h_c^{-1}$ and $b_\epsilon=m\Phi_c\rho_c$. By
(S1) and the additional hypothesis, in $L^\infty$,
\[
a_\epsilon=\delta+\epsilon a'+o(\epsilon),\quad
b_\epsilon=m\Phi+\epsilon b'+o(\epsilon),\quad
\rho_c=1+\epsilon\rho'+o(\epsilon),
\]
with $a'=(\Phi+\frac12\operatorname{tr}k)\delta-k$, $b'=\frac12m\varrho+
m\Phi(\frac12\operatorname{tr}k-\Phi)$ and $\rho'=\frac12
\operatorname{tr}k-\Phi$. By the proof of Theorem~\ref{thm:stars}, $S_c$
has, for $c$ large, a single eigenvalue $s(\epsilon)$ near $E$, with
$s(\epsilon)\to E$. Moreover~\eqref{eq:factorised-resolvent} shows that
the resolvents of $\mathcal{U}_cS_c\mathcal{U}_c^*$ converge to those of
$H_S$ in the norm of $\mathcal{B}(L^2,H^1)$, uniformly on a circle about
$E$: with $R_c(z)$ and $R(z)$ these resolvents,
$R_c(z)-R(z)=(Q_c^{-1}-Q^{-1})[\one+(z+\lambda_0)R_c(z)]+(z+\lambda_0)
Q^{-1}(R_c(z)-R(z))$, and $Q^{-1}$ is bounded from $L^2$ to $H^2$. The
Riesz projections converge in the same norm, and
$\mathcal{U}_c^{\pm1}$ converge to the identity in $\mathcal{B}(H^1)$
because $\rho_c\to1$ in $W^{1,\infty}$. Hence there are eigenvectors
$w_\epsilon$ of $S_c$ for $s(\epsilon)$ with $w_\epsilon\to u$ in $H^1$.
For all $v\in H^1$, $\tilde q_\epsilon(v,w_\epsilon)=s(\epsilon)
\braket{v,w_\epsilon}_\epsilon$, while $\tilde q_0(u,w)=E\braket{u,w}_0$
for all $w\in H^1$. Subtracting, with $v=u$,
\[
\bigl(s(\epsilon)-E\bigr)\braket{u,w_\epsilon}_\epsilon
=(\tilde q_\epsilon-\tilde q_0)(u,w_\epsilon)
-E\int(\rho_c-1)\,u\,w_\epsilon\,\rmd^3x .
\]
Dividing by $\epsilon$ and letting $\epsilon\to0$ gives
$(s(\epsilon)-E)/\epsilon\to D_1(u)$. Finally
$E(c)=F_c(s(\epsilon))=s(\epsilon)-s(\epsilon)^2/(2mc^2)+O(c^{-4})$ by
Lemma~\ref{lem:reduction}. This proves (i).

(ii) Under the stronger hypothesis the three expansions above hold with
remainders $O(\epsilon^2)$ in $L^\infty$. Since $\delta_c=O(\epsilon)$,
the resolvent identity above makes the Riesz projections converge at rate
$O(\epsilon)$ in $\mathcal{B}(L^2,H^1)$, and with
$\|\rho_c^{\pm1/2}-1\|_{W^{1,\infty}}=O(\epsilon)$ this gives
$\|w_\epsilon-u\|_{H^1}=O(\epsilon)$. The
identity above then gives $(s(\epsilon)-E)(1+O(\epsilon))=\epsilon D_1(u)
+O(\epsilon^2)$, hence $s(\epsilon)=E+\epsilon D_1(u)+O(\epsilon^2)$.

(iii) In the proof of Theorem~\ref{thm:stars}, $\mathcal{U}_cS_c
\mathcal{U}_c^*$ is the operator of the form~\eqref{eq:qc} on the fixed
domain $H^1(\R^3)$, whose coefficients are built from $N_c^2h_c^{ij}$,
$\nabla\ln\rho_c$ and $\Phi_c$. Under the hypothesis they are holomorphic
in $\epsilon$ with values in $L^\infty$, and we write the form for complex
$\epsilon$ without conjugating them, so that it is holomorphic. Since
the coefficients differ from those of $H_S$ by $O(\epsilon)$ in
$L^\infty$, it is sectorial and closed on $H^1(\R^3)$ with vertex and
angle uniform in $\epsilon$ near $0$, and the associated operators form a
holomorphic family of type~(B)~\cite{Kato1995} (Sec.~VII.4). The simple
isolated eigenvalue $s(\epsilon)$ near $E$ is then holomorphic near
$\epsilon=0$ (Ref.~\onlinecite{Kato1995}, Sec.~VII.1.3), and so is
$E(c)=F_c(s)=(m/\epsilon)((1+2\epsilon s/m)^{1/2}-1)$, holomorphic in
$(\epsilon,s)$ near $(0,E)$ with a removable singularity at
$\epsilon=0$.
\end{proof}

Equation~\eqref{eq:PN} separates two corrections. The term $-E^2/2m$ is
kinematic and universal: it comes from $F_c$ and does not depend on the
star. The term $D_1(u)$ is geometric: it collects the $O(c^{-2})$ parts
of the lapse, of the spatial metric and of the one-particle measure,
evaluated on the Newtonian bound state. Since $D_1$ involves no
derivative of the metric, the formula also covers metrics that are only
Lipschitz at the surface of the star. For isotropic spatial metrics,
$h_c=(1-2\epsilon\Phi)\delta$, the gradient term of $D_1$ vanishes and
$D_1(u)=4E\braket{u,\Phi u}-3m\braket{u,\Phi^2u}+\frac m2\braket{u,
(\varrho-2\Phi^2)u}$. When $\varrho=2\Phi^2$, as outside a static body in
isotropic coordinates, \eqref{eq:PN} is first-order perturbation theory
for the formal first post-Newtonian Hamiltonian $H_S+c^{-2}H_1$, with
$H_1=-p^4/8m^3+3\{p^2,\Phi\}/4m+m\Phi^2/2$; formal derivations of such
Hamiltonians are in Refs.~\onlinecite{Lammerzahl1995,SchwartzGiulini2019}.
The ordering of $p^2$ and $\Phi$ is not a choice: it is the one produced
by the density map, formally $\mathcal{U}_cP_c\mathcal{U}_c^*=H_S+c^{-2}H_1
+O(c^{-4})$, and other orderings differ by multiples of
$\hbar^2\Delta\Phi/mc^2$, which does not vanish inside a star. The
functional $D_1$ is invariant under the change of coordinates
$x\mapsto x+\epsilon\xi$, which shifts $k$ by $\partial\xi+
(\partial\xi)^{\rm T}$ and $\varrho$ by $2\xi\cdot\nabla\Phi$: the
variation of the integrand is a divergence when $u$ is an eigenfunction.
If $E$ has multiplicity $\mu$, the same argument on the range of the
Riesz projection shows that the $\mu$ eigenvalues of $P_c$ near $E$ are
$E+c^{-2}(\lambda_j-E^2/2m)+o(c^{-2})$, where the $\lambda_j$ are the
eigenvalues of the matrix of the polarisation of $D_1$ on a real
orthonormal basis of the eigenspace, with remainder $O(c^{-4})$ under
the hypothesis of (ii).

\subsection{Example: a star of uniform density}\label{sec:uniform-star}

The interior Schwarzschild solution of uniform density, of mass $M$ and
radius $R$, glued to the exterior at $r=R$, has
\[
N_c=\tfrac32\bigl(1-\tfrac{2GM}{c^2R}\bigr)^{1/2}-\tfrac12\bigl(1-
\tfrac{2GMr^2}{c^2R^3}\bigr)^{1/2},\qquad
h_{c,rr}=\bigl(1-\tfrac{2GMr^2}{c^2R^3}\bigr)^{-1}\qquad(r\leq R).
\]
The lapse and its derivative are continuous at $r=R$, the metric is
Lipschitz there, and (S1) holds with $\Phi=-GM(3R^2-r^2)/(2R^3)$ inside
and $\Phi=-GM/r$ outside; the $O(c^{-2})$ part of $\Phi_c$ inside is
$3G^2M^2(R^2-r^2)^2/(8R^6)$. The Buchdahl bound $2GM/(c^2R)<8/9$ holds
for $c$ large. The metric satisfies the hypotheses of
Proposition~\ref{prop:PN}(ii) and (iii): $N_c$ and $h_c$ are holomorphic
in $\epsilon$ for $|\epsilon|<R/2GM$, $N_c$ has no zeros for
$|\epsilon|<4R/9GM$, where $2GM\epsilon/R=8/9$ is the Buchdahl value, and
the identities $\varrho_c(R)=\partial_r\varrho_c(R)=0$ and the continuity
of $k_c$ at $r=R$ hold for complex $\epsilon$, so $(\varrho_c,k_c)$ is
holomorphic with values in $W^{1,\infty}(\R^3)$. By
Proposition~\ref{prop:PN}(iii) its levels, in each angular-momentum
sector, are analytic in $c^{-2}$. Table~\ref{tab:stars} compares the bound-state energies
of $P_c$, computed by shooting in $r$ (Appendix~\ref{app:numerics}), with
Proposition~\ref{prop:PN}. In units
$\hbar=m=G=M=1$ the gravitational Bohr radius is $1$, the point-mass
levels are $-1/(2n^2)$, and $c$ is the inverse of the gravitational
fine-structure constant $\alpha=GMm/(\hbar c)$.

\begin{table}[htb]
\caption{Uniform-density star of radius $R$, with $a_0=\hbar^2/GMm^2$
the gravitational Bohr radius and $\alpha=GMm/\hbar c$. $E$ is the
Newtonian level of $H_S$ for angular momentum $l$ and radial index
$n_r$. The next four columns give $(E(c)-E)/mc^2\alpha^4$ for the bound
states of $P_c$, and the last column the coefficient
$(D_1(u)-E^2/2m)/mc^4\alpha^4$ of Proposition~\ref{prop:PN}, their limit;
$mc^4\alpha^4=G^4M^4m^5/\hbar^4$ does not depend on $c$. All entries are
dimensionless; in units $\hbar=m=G=M=1$, $a_0=1$, $mc^2\alpha^2=1$ and
$\alpha^{-1}=c$.}
\label{tab:stars}
\begin{ruledtabular}
\begin{tabular}{rrrrrrrrr}
$R/a_0$ & $l$ & $n_r$ & $E/mc^2\alpha^2$ &
\multicolumn{4}{c}{$(E(c)-E)/mc^2\alpha^4$} & Prop.~\ref{prop:PN}\\
 & & & & $\alpha^{-1}=16$ & $32$ & $64$ & $128$ & \\
\hline
0.5 & 0 & 0 & $-0.449735$ & $-1.73172$ & $-1.70556$ & $-1.69914$ & $-1.69754$ & $-1.69701$\\
0.5 & 0 & 1 & $-0.118499$ & $-0.32653$ & $-0.32199$ & $-0.32088$ & $-0.32060$ & $-0.32051$\\
0.5 & 1 & 0 & $-0.124960$ & $-0.13309$ & $-0.13218$ & $-0.13196$ & $-0.13190$ & $-0.13189$\\
0.5 & 1 & 1 & $-0.055541$ & $-0.05104$ & $-0.05070$ & $-0.05062$ & $-0.05060$ & $-0.05059$\\
2 & 0 & 0 & $-0.293497$ & $-0.27730$ & $-0.27663$ & $-0.27647$ & $-0.27642$ & $-0.27641$\\
2 & 0 & 1 & $-0.094906$ & $-0.09009$ & $-0.08980$ & $-0.08973$ & $-0.08971$ & $-0.08970$\\
2 & 1 & 0 & $-0.121336$ & $-0.10128$ & $-0.10089$ & $-0.10079$ & $-0.10076$ & $-0.10076$\\
2 & 1 & 1 & $-0.054313$ & $-0.03973$ & $-0.03958$ & $-0.03954$ & $-0.03953$ & $-0.03953$\\
\end{tabular}
\end{ruledtabular}
\end{table}

The coefficients $c^2(E(c)-E)$ converge to the values of
Proposition~\ref{prop:PN}, and their distance from them falls by a factor
close to $4$ each time $c$ doubles. This is consistent with the
remainder $O(c^{-4})$ of Proposition~\ref{prop:PN}(ii): by (iii), $E(c)$
is a power series in $c^{-2}$, and the factor tends to $4$ when the
coefficient of $c^{-4}$ does not vanish. At $c=128$ they differ from them by
at most $5.3\times10^{-4}$, a relative $3\times10^{-4}$. The Richardson
extrapolation $(4f(128)-f(64))/3$ of the last two columns, computed
from the unrounded values, reproduces the
column of Proposition~\ref{prop:PN} to $2\times10^{-6}$. This is a check
of the proposition, not a consistency check, because that column comes
from a quadrature of $D_1(u)$ and the others from the spectrum of $P_c$.
Every level lies above the point-mass level with the same $l$ and $n_r$,
because $\Phi\geq-GM/r$, with strict inequality inside the star; the
shift is largest for $l=0$. The relativistic correction deepens every
level.

\section{The Schwarzschild exterior}\label{sec:schw}

\subsection{Strong but not norm convergence}

On the Schwarzschild exterior,
\[
N_c^2=1-\frac{r_{\rm s}}{r},\qquad r_{\rm s}:=\frac{2GM}{c^2},\qquad
h_c=N_c^{-2}\rmd r^2+r^2\rmd\Omega^2,\qquad
\Sigma_c=\{x\in\R^3:\ |x|>r_{\rm s}\},
\]
in Schwarzschild coordinates, read as polar coordinates on $\R^3$. Then
$\rho_c=N_c^{-2}$, the density map is $\mathcal{U}_cu=N_c^{-1}u$, and
$\mathcal{U}_c^*g=N_cg|_{\Sigma_c}$. By Lemma~\ref{lem:reduction},
$S_c=\frac{\hbar^2}{2m}A_c-GMm/r$ with the exact Newtonian potential;
$S_c$ differs from the gravitational hydrogen Hamiltonian only through
the kinetic operator $A_c$ and the domain $\Sigma_c$. The limiting
operator is
\begin{equation}\label{eq:HS-schw}
H_S:=-\frac{\hbar^2}{2m}\Delta-\frac{GMm}{r}\quad\text{on }H^2(\R^3),
\qquad \sigma(H_S)=\Bigl\{E_n=-\frac{G^2M^2m^3}{2\hbar^2n^2}\Bigr\}_{n\geq1}
\cup[0,\infty).
\end{equation}
Its restriction to $C_c^\infty(\R^3\setminus\{0\})$ has deficiency
indices $(1,1)$, from the $s$-wave; its self-adjoint extensions add a
point interaction at the origin, and $H_S$ is the Friedrichs
extension~\cite{Kato1995}. The next theorem shows that the limit
selects it.

\begin{theorem}[Schwarzschild exterior]\label{thm:schw}
For the Schwarzschild exterior and $H_S$ as in~\eqref{eq:HS-schw}:
\begin{enumerate}[label=(\alph*)]
\item $\mathcal{U}_c(P_c-z)^{-1}\mathcal{U}_c^*\to(H_S-z)^{-1}$ strongly,
for every $z\in\C\setminus\R$. Hence $\mathcal{U}_c\,f(P_c)\,
\mathcal{U}_c^*\to f(H_S)$ strongly for every bounded continuous $f$,
and $\mathcal{U}_c\rme^{-\rmi tP_c/\hbar}\mathcal{U}_c^*\to
\rme^{-\rmi tH_S/\hbar}$ strongly, uniformly for $t$ in compact sets.
\item For every $c$, $\sigma(P_c)=[-mc^2,\infty)$ and $P_c$ has no
eigenvalues. The convergence in (a) is not in norm.
\item (Spectral concentration.) If $I$ is a bounded open interval with
$E_n\in I$ and $\bar I\cap\sigma(H_S)=\{E_n\}$, then $\mathcal{U}_c\one_I(P_c)
\mathcal{U}_c^*\to\Pi_n$ strongly, where $\Pi_n$ is the eigenprojection
of $H_S$ for $E_n$.
\item For $f\in C_c^\infty(\R\times(\R^3\setminus\{0\}))$ and $c$ large,
$F_f^{(c)}$ is defined, \eqref{eq:H2} holds, and the two-point functions
of $\hat\psi_c$ in the Boulware state converge to those of $\hat\psi$ in
the Fock vacuum of $\rmd\Gamma(H_S)$.
\end{enumerate}
\end{theorem}

\begin{proof}
\emph{(a)} By Lemma~\ref{lem:reduction} it suffices to prove the claim
for $S_c$. Appendix~\ref{app:schw} shows that $K_c$ is essentially
self-adjoint on $C_c^\infty(\Sigma_c)$ (Lemma~\ref{lem:partial-waves})
and constructs, for every $\varphi\in C_c^\infty(\R^3)$, vectors
$u_c\in C_c^\infty(\Sigma_c)$ with
\begin{equation}\label{eq:graph-limit}
\|\mathcal{U}_cu_c-\varphi\|\to0,\qquad
\|\mathcal{U}_cS_cu_c-H_S\varphi\|\to0
\end{equation}
(Lemma~\ref{lem:approximants}). The approximants are
$u_c=\theta(r_*/L_c)\,\varphi$, where $r_*$ is the tortoise coordinate,
$\theta$ is a smooth step and $L_c=(\ell r_{\rm s})^{1/2}$ for a fixed
length $\ell$; the cut sits exponentially close to the horizon, where
$A_c$ degenerates. Since
$C_c^\infty(\R^3)$ is a core for $H_S$ and $\mathcal{U}_c$ is an isometry,
the graph-limit criterion of Lemma~\ref{lem:graph-limit} gives
$\mathcal{U}_c(S_c-z)^{-1}\mathcal{U}_c^*\to(H_S-z)^{-1}$ strongly. For
the functional calculus let $A'_c:=\mathcal{U}_cP_c\mathcal{U}_c^*\oplus0$
on $L^2(\Sigma_c)\oplus L^2(\{r\leq r_{\rm s}\})$. Then
$(A'_c-z)^{-1}=\mathcal{U}_c(P_c-z)^{-1}\mathcal{U}_c^*-z^{-1}
\one_{\{r\leq r_{\rm s}\}}$, and $\one_{\{r\leq r_{\rm s}\}}\to0$
strongly, so $A'_c\to H_S$ in the strong resolvent sense on
$L^2(\R^3)$. By Ref.~\onlinecite{ReedSimon1972} (Theorems VIII.20 and
VIII.21), $f(A'_c)\to f(H_S)$ strongly for bounded continuous $f$ and
$\rme^{-\rmi tA'_c/\hbar}\to\rme^{-\rmi tH_S/\hbar}$ strongly, uniformly
for $t$ in compact sets. Finally $f(A'_c)=\mathcal{U}_cf(P_c)
\mathcal{U}_c^*+f(0)\one_{\{r\leq r_{\rm s}\}}$.

\emph{(b)} Lemma~\ref{lem:partial-waves} shows $\sigma(S_c)=
[-mc^2/2,\infty)$ without eigenvalues; for the massive field on the
Schwarzschild exterior the absence of eigenvalues is due to
Bachelot~\cite{Bachelot1994}. Since $F_c$ is a strictly
increasing bijection of $[-mc^2/2,\infty)$ onto $[-mc^2,\infty)$,
the same holds for $P_c$ with $[-mc^2,\infty)$. Let
$\lambda\in(E_1,E_2)$, $\eta>0$ and $z=\lambda+\rmi\eta$, and let $c$ be
so large that $\lambda>-mc^2$. Since
$\lambda\in\sigma(P_c)$ and $\mathcal{U}_c$ is an isometry,
$\|\mathcal{U}_c(P_c-z)^{-1}\mathcal{U}_c^*\|=\|(P_c-z)^{-1}\|=1/\eta$,
whereas $\|(H_S-z)^{-1}\|=(\eta^2+d^2)^{-1/2}$ with $d=
\operatorname{dist}(\lambda,\sigma(H_S))>0$. The difference of the two
operators therefore has norm at least $1/\eta-(\eta^2+d^2)^{-1/2}>0$ for
every such $c$.

\emph{(c)} The end points of $I$ are not eigenvalues of $H_S$, so
$\one_I(A'_c)\to\one_I(H_S)=\Pi_n$ strongly~\cite{ReedSimon1972}
(Theorem VIII.24), and $\one_I(A'_c)-\mathcal{U}_c\one_I(P_c)
\mathcal{U}_c^*=\one_I(0)\one_{\{r\leq r_{\rm s}\}}\to0$.

\emph{(d)} Let $\supp f\subset\R\times\{r>R_0\}$ and $c$ so large that
$N_c^2\geq\frac12$ on $\{r>R_0\}$. For $g\in L^2(\{r>R_0\})$,
\eqref{eq:K-lower} and operator monotonicity give
\begin{equation}\label{eq:gamma-bound}
\|\gamma_c\,\mathcal{U}_c^*g\|^2=\frac{mc^2}{\hbar}\braket{\mathcal{U}_c^*g,
K_c^{-1/2}\mathcal{U}_c^*g}\leq\braket{\mathcal{U}_c^*g,N_c^{-1}
\mathcal{U}_c^*g}_{\K_c}=\int N_c^{-1}|g|^2\,\rmd^3x\leq2^{1/2}\|g\|^2,
\end{equation}
and likewise $\|\gamma_c^2\,\mathcal{U}_c^*g\|^2\leq2\|g\|^2$. Hence
$F_f^{(c)}$ is defined. To show $\mathcal{U}_c\gamma_c\mathcal{U}_c^*g
\to g$, split the spectral measure $\mu_g$ of $\mathcal{U}_c^*g$ for
$P_c$ into $|p|\leq\Lambda$, $p>\Lambda$ and $p<-\Lambda$. On the first
part $|\gamma_c-1|\leq2\Lambda/mc^2$ once $mc^2\geq4\Lambda$. On the
second, $0\leq\gamma_c\leq1$, and $\mu_g(p>\Lambda)\leq\braket{g,
\mathcal{U}_ch_\Lambda(P_c)\mathcal{U}_c^*g}\to\braket{g,h_\Lambda(H_S)g}$
by (a), for a continuous $h_\Lambda$ with $\one_{(\Lambda,\infty)}\leq
h_\Lambda\leq\one_{(\Lambda/2,\infty)}$; the limit tends to $0$ as
$\Lambda\to\infty$. On the third, by Cauchy--Schwarz and the bound on
$\gamma_c^2$, $\int_{p<-\Lambda}\gamma_c^2\,\rmd\mu_g\leq(2\|g\|^2)^{1/2}
\,\mu_g(p<-\Lambda)^{1/2}$, and $\mu_g(p<-\Lambda)\to0$ by the same
argument, with a continuous function between $\one_{(-\infty,-\Lambda)}$
and $\one_{(-\infty,-\Lambda/2)}$, once $-\Lambda/2<E_1$, because
$\sigma(H_S)\subset[E_1,\infty)$. Letting first $c\to\infty$ and then $\Lambda\to\infty$ gives
$\mathcal{U}_c\gamma_c\mathcal{U}_c^*g\to g$, uniformly for $g$ in
compact subsets of $L^2(\{r>R_0\})$, on which these operators are
uniformly bounded. Together with (a), and $\mathcal{U}_c^*\mathcal{U}_c=
\one$, dominated convergence in $t$ gives~\eqref{eq:H2}; the two-point
functions converge by~\eqref{eq:two-point}.
\end{proof}

\begin{corollary}[Spectral concentration and metastability]
\label{cor:metastability}
Let $u$ be a normalised eigenfunction of $H_S$ for $E_n$. Then
\[
\braket{\mathcal{U}_c^*u,\rme^{-\rmi tP_c/\hbar}\mathcal{U}_c^*u}
\longrightarrow\rme^{-\rmi tE_n/\hbar}
\]
uniformly for $t$ in compact sets, and the spectral measure of $\mathcal{U}_c^*u$ for $P_c$ concentrates at $E_n$:
its mass in every open interval containing $E_n$ tends to $1$. More
precisely, with $\alpha=GMm/\hbar c$ and the gravitational Bohr time
$T_{\rm B}:=\hbar^3/G^2M^2m^3$, there is $C$ such that for $c$ large and
all $t$
\[
\bigl|\braket{\mathcal{U}_c^*u,\rme^{-\rmi tP_c/\hbar}\mathcal{U}_c^*u}
-\rme^{-\rmi tF_c(E_n)/\hbar}\bigr|\leq C\bigl(1+|t|/T_{\rm B}\bigr)\,
\alpha^{1/2},
\]
and the spectral measure of $\mathcal{U}_c^*u$ for $P_c$ gives mass at
most $C\alpha(mc^2\alpha^2/\delta)^2+C\alpha^2$ to the complement of
$[F_c(E_n)-\delta,F_c(E_n)+\delta]$.
For every fixed $c$, on the other hand,
\[
\lim_{T\to\infty}\frac1T\int_0^T\bigl|\braket{\mathcal{U}_c^*u,
\rme^{-\rmi tP_c/\hbar}\mathcal{U}_c^*u}\bigr|^2\rmd t=0 .
\]
\end{corollary}

\begin{proof}
The first two statements are Theorem~\ref{thm:schw}(a) and (c). For the
bounds, the proof of Lemma~\ref{lem:approximants} applies to
$\varphi=u$, for which $|u|+|\partial_ru|+|\partial_r^2u|+|\Delta u|\leq
C(1+r^{-1})\rme^{-r/2na_0}$, with $u_c:=\theta(r_*/L)\,u\in D(S_c)$ in
place of an element of $C_c^\infty(\Sigma_c)$; it gives
$\|\mathcal{U}_cu_c-u\|=O(r_{\rm s})$ and
$\|(S_c-E_n)u_c\|=\|\mathcal{U}_cS_cu_c-E_n\mathcal{U}_cu_c\|=
O(r_{\rm s}^{1/4})=O(c^{-1/2})$. Since
$|F_c(s)-F_c(E_n)|\leq2|s-E_n|/(1+2E_n/mc^2)^{1/2}$ for $s\geq-mc^2/2$,
also $\|(P_c-F_c(E_n))u_c\|=O(c^{-1/2})$, and Duhamel's formula gives
$|\braket{u_c,\rme^{-\rmi tP_c/\hbar}u_c}-\rme^{-\rmi tF_c(E_n)/\hbar}
\|u_c\|^2|\leq|t|\,\|(P_c-F_c(E_n))u_c\|\,\|u_c\|/\hbar$. Finally
$\|u_c-\mathcal{U}_c^*u\|\leq\|\mathcal{U}_cu_c-u\|$ and
$\|u_c\|^2=1+O(r_{\rm s})$. The spectral mass of $u_c$ outside the
interval is at most $\|(P_c-F_c(E_n))u_c\|^2/\delta^2$, and that of
$\mathcal{U}_c^*u$ differs from it by $O(r_{\rm s})$. For $l\geq1$ the
approximants do better, $u\sim r^l$ making the Newtonian term near the
horizon smaller, and the rate is $\alpha$. The last statement is Wiener's
theorem, since $P_c$ has no eigenvalues.
\end{proof}

At every finite $c$ a hydrogenic state leaks through the horizon; in
the limit it does not. The limits $c\to\infty$ and $t\to\infty$ do not
commute, and the bound states of $H_S$ arise from states of the
Boulware one-particle space whose lifetime diverges with $c$: the bound
proves a lifetime of order at least $c^{1/2}$, and
Theorem~\ref{thm:resonances} gives the width, of order $c^{-(4l+3)}$
(Sec.~\ref{sec:atoms}).

\begin{remark}[Which extension, and why]\label{rem:friedrichs}
The self-adjoint extensions of $H_S$ on $C_c^\infty(\R^3\setminus\{0\})$
differ in the behaviour of $s$-waves at the origin. At finite $c$ there
is no choice: $K_c$ is essentially self-adjoint, because the horizon
is at infinite tortoise distance. The mechanism that selects the
Friedrichs extension is visible in~\eqref{eq:Ac-tortoise}: in the
tortoise coordinate $A_c$ is $-r^{-2}\partial_{r_*}r^2\partial_{r_*}$
plus an angular term with coefficient $N_c^2$, so functions constant near
the horizon are annihilated by the radial part, and the static solution
regular at the horizon matches the Coulomb solution regular at the
origin. The approximants of Lemma~\ref{lem:approximants} are of this
kind. Numerically, for $s$-waves the resolvent error of
Table~\ref{tab:schw} is carried by the part of the solution absorbed at
the horizon, not by the region near the origin; for $p$-waves the
absorbed part is negligible (Sec.~\ref{sec:schw-numerics}).
\end{remark}

\begin{remark}[Order of limits near the horizon]\label{rem:order-limits}
On a fixed region $\{r>R\}$ the potential $m\Phi_c=-GMm/r$ is
$O(1)$ and the kinetic coefficients converge, so Theorem~\ref{thm:schw}
describes the limit $c\to\infty$ at fixed $r$, in which the horizon
shrinks to the point $r=0$ and the limiting spatial manifold is
$\R^3\setminus\{0\}$. At fixed $c$, instead, $|m\Phi_c|/mc^2=r_{\rm s}/2r\to
\frac12$ as $r\to r_{\rm s}$, and no Newtonian description exists near
the horizon. The proof reflects this: the approximants are cut at
$r_*\in[-2L_c,-L_c]$, with $L_c=(\ell r_{\rm s})^{1/2}$, and the error
they make there is what limits the rate. The near-horizon thermal
structure of the Hartle--Hawking state, which requires the bifurcate
Killing horizon~\cite{KayWald1991}, has no limit
(Sec.~\ref{sec:thermal}).
\end{remark}

The charge of a Reissner--Nordstr\"om black hole does not enter the
limit of a neutral field. On its exterior, of mass $M$ and charge $Q$,
\[
N_c^2=1-\frac{r_{\rm s}}{r}+\frac{r_Q^2}{r^2}=\Bigl(1-\frac{r_+}{r}\Bigr)
\Bigl(1-\frac{r_-}{r}\Bigr),\qquad r_Q^2=\frac{GQ^2}{4\pi\varepsilon_0c^4},
\]
with $h_c=N_c^{-2}\rmd r^2+r^2\rmd\Omega^2$ on $\Sigma_c=\{|x|>r_+\}$. The
ratio $r_Q/r_{\rm s}$ does not depend on $c$. If $Q^2\leq4\pi\varepsilon_0
GM^2$, then $r_\pm$ are fixed multiples of $r_{\rm s}$ and $q:=r_-/r_+
\leq1$ is fixed; $q=1$ is the extremal case.

\begin{corollary}[Reissner--Nordstr\"om exterior]\label{cor:RN}
Let $Q^2\leq4\pi\varepsilon_0GM^2$. For the neutral Klein--Gordon field on
the Reissner--Nordstr\"om exterior, Theorem~\ref{thm:schw},
Corollary~\ref{cor:metastability} and Proposition~\ref{prop:net-static}
hold, with the same $H_S$ and with the Boulware state of this exterior.
\end{corollary}

\begin{proof}
Lemma~\ref{lem:reduction} applies, with $\Phi_c=-GM/r+c^2r_Q^2/2r^2$. Here
$c^2r_Q^2=q_*GMr_{\rm s}/2\leq GMr_+$, with $q_*:=Q^2/4\pi\varepsilon_0GM^2
\leq1$, so the second term is $O(c^{-2})$ at fixed $r$. Take
$r_*=r+\frac{r_+^2}{r_+-r_-}\ln\frac{r-r_+}{r_+}-\frac{r_-^2}{r_+-r_-}
\ln\frac{r-r_-}{r_+}$ if $q<1$ and $r_*=r+2r_+\ln\frac{r-r_+}{r_+}-
\frac{r_+^2}{r-r_+}$ if $q=1$; then $\rmd r_*/\rmd r=N_c^{-2}$, $r_*$ is an
increasing bijection of $(r_+,\infty)$ onto $\R$, and $r_*(2r_+)=O(r_+)$.
The form~\eqref{eq:Ac-tortoise} of $A_c$ holds for every metric
$N_c^{-2}\rmd r^2+r^2\rmd\Omega^2$, and in Lemma~\ref{lem:partial-waves}
$W_l=\frac{\hbar^2}{2m}N_c^2\bigl(l(l+1)/r^2+(N_c^2)'/r\bigr)+m\Phi_c$, so
$W_l+mc^2/2=N_c^2[\frac{\hbar^2}{2m}(l(l+1)/r^2+(N_c^2)'/r)+
\frac{mc^2}{2}]\geq0$, since $(N_c^2)'\geq0$ for $r\geq r_+$. As
$r_*\to-\infty$, $N_c^2=O(\rme^{2\kappa_+r_*})$ with surface gravity
$\kappa_+=(r_+-r_-)/2r_+^2$ if $q<1$, and $N_c^2=O(r_+^2/r_*^2)$ if $q=1$.
In both cases $\eta_l:=\frac{2m}{\hbar^2}(W_l+\frac{mc^2}2)$ is integrable at
$-\infty$, so for $E>-mc^2/2$ every solution is asymptotic to a
combination of plane waves~\cite{Weidmann1987} and is not square
integrable at $-\infty$. At $E=-mc^2/2$ an eigenfunction $w$ of the
partial-wave operator, which lies in $H^2(\R)$ because $W_l$ is bounded,
would satisfy $\frac{\hbar^2}{2m}\|w'\|^2+\int(W_l+\frac{mc^2}2)|w|^2\,
\rmd r_*=0$; since $W_l+\frac{mc^2}2\geq0$, this gives $w'=0$ and $w=0$.
So Lemma~\ref{lem:partial-waves} holds; for $q<1$ the absence of bound
states also follows from the decay results of
Ref.~\onlinecite{PasqualottoShlapentokhRothmanVandeMoortel2023} (see
Ref.~\onlinecite{ShlapentokhRothmanVandeMoortel2023}, Sec.~1.2). In
Lemma~\ref{lem:approximants} replace $r_{\rm s}$ by $r_+$ in $L$ and in
the constants, which now depend on $q$. The kinetic terms hold with
$r_{\rm s}/r$ and $r_{\rm s}/r^2$ replaced by $1-N_c^2$ and $(N_c^2)'$,
which satisfy $0\leq1-N_c^2\leq r_{\rm s}/r$ and $0\leq(N_c^2)'\leq
r_{\rm s}/r^2$ for $r\geq r_+$, with $r_{\rm s}\leq2r_+$. On
$\{r_*\leq-L\}$, $N_c^2\leq C\exp(-(1-q)L/r_+)$ if $q<1$ and $N_c^2\leq
Cr_+^2/L^2$ if $q=1$; in both cases the terms of the commutator that
carry a factor $N_c^2$ are at most $Cr_+^{3/2}$ in square norm. The only
new term in $N_c^{-1}S_c\varphi-H_S\varphi$ is
$(mc^2r_Q^2/2r^2)N_c^{-1}\varphi$: its norm is $O(r_+^{1/2})$ on
$\{r\geq2r_+\}$, and on $\{r_+<r<2r_+\}$ its square norm weighted by
$\theta^2$ is at most $C\int\theta^2\,\rmd r_*\leq CL$, as for the
Newtonian term. Hence $\|\mathcal{U}_cS_cu_c-H_S\varphi\|=O(r_+^{1/4})$,
and the proofs of Theorem~\ref{thm:schw} and of
Corollary~\ref{cor:metastability} go through. The exterior is static and
globally hyperbolic, with $\partial_t$ timelike everywhere, which is all
that the proof of Proposition~\ref{prop:net-static} uses.
\end{proof}

\subsection{Numerical illustration}\label{sec:schw-numerics}

Table~\ref{tab:schw} shows $\mathcal{R}(c)=\|\mathcal{U}_c(S_c-z)^{-1}
\mathcal{U}_c^*g-(H_S-z)^{-1}g\|/\|g\|$ for a smooth bump $g$ supported
in $0.5\leq r/a_0\leq3$. The Green's function of each partial wave is
integrated numerically and compared with the resolvent of the
gravitational hydrogen atom (Appendix~\ref{app:numerics}). The error
decreases as $c^{-2}$ for $p$-waves. For $s$-waves the local exponent
drifts towards $-3/2$: for $z=-0.3+0.1\rmi$ it is $-1.72$, $-1.62$,
$-1.56$, $-1.53$ and $-1.51$ between successive rows, and $-1.506$
between $c=384$ and $c=768$. The $s$-wave error is dominated by the part
of the solution that is absorbed at the horizon, and flux conservation
gives its norm. If $v$ is the radial function of $(H_S-z)^{-1}g$ and
$\varpi:=kr_{\rm s}$ with $k$ as in Lemma~\ref{lem:continuation},
Lemmas~\ref{lem:matching} and~\ref{lem:far} and Green's formula give,
in units $\hbar=m=1$,
$|v(0)|\,(\mathrm{Re}\,\varpi\;r_{\rm s}/2\,\mathrm{Im}\,z)^{1/2}(1+o(1))$
for that
norm, which is proportional to $c^{-3/2}$; it agrees with the computed
absorbed part to $2\%$ for $c\geq48$ and both values of $z$ (for $z=\rmi$
and $c=192$, $7.50\times10^{-4}$ against $7.51\times10^{-4}$). For
$p$-waves the same balance gives an absorbed part proportional to
$c^{-7/2}$. For $z=-0.3+0.1\rmi$, which lies in the gap between $E_1$ and
$E_2$, the resolvent of $H_S$ is large and the limit is approached
later, but at the same rate.

\begin{table}[htb]
\caption{Relative resolvent error $\mathcal{R}(c)$ of $S_c$ on the
Schwarzschild exterior, for the $s$-wave ($l=0$) and the $p$-wave
($l=1$). $\mathcal{R}(c)$ is dimensionless, $z$ is in units of
$mc^2\alpha^2=G^2M^2m^3/\hbar^2$, in which $E_1=-\frac12$ and
$E_2=-\frac18$, and $\alpha^{-1}=\hbar c/GMm$. The last row is the local
exponent between $\alpha^{-1}=192$ and $384$.}
\label{tab:schw}
\begin{ruledtabular}
\begin{tabular}{rllll}
$\alpha^{-1}$ & $l=0$, $z=-0.3+0.1\rmi$ & $l=0$, $z=\rmi$ & $l=1$, $z=-0.3+0.1\rmi$ & $l=1$, $z=\rmi$\\
\hline
12 & $1.16$ & $7.07\times10^{-2}$ & $3.98\times10^{-2}$ & $9.18\times10^{-3}$\\
24 & $3.53\times10^{-1}$ & $2.05\times10^{-2}$ & $9.80\times10^{-3}$ & $2.28\times10^{-3}$\\
48 & $1.15\times10^{-1}$ & $6.55\times10^{-3}$ & $2.44\times10^{-3}$ & $5.69\times10^{-4}$\\
96 & $3.89\times10^{-2}$ & $2.21\times10^{-3}$ & $6.10\times10^{-4}$ & $1.42\times10^{-4}$\\
192 & $1.35\times10^{-2}$ & $7.65\times10^{-4}$ & $1.52\times10^{-4}$ & $3.55\times10^{-5}$\\
384 & $4.74\times10^{-3}$ & $2.68\times10^{-4}$ & $3.81\times10^{-5}$ & $8.88\times10^{-6}$\\
\hline
exponent & $-1.51$ & $-1.51$ & $-2.00$ & $-2.00$\\
\end{tabular}
\end{ruledtabular}
\end{table}

The same Green's function gives the spectral density $\pi^{-1}\,
\mathrm{Im}\braket{\mathcal{U}_c^*g,(S_c-E-\rmi0)^{-1}
\mathcal{U}_c^*g}$ of Corollary~\ref{cor:metastability} near $E_1$. It is
well fitted by a Lorentzian plus a linear background
(Appendix~\ref{app:numerics}). For $c=10$ the centre and half-width are
$-0.546304$ and
$1.5265\times10^{-2}$; for $c=20$ they are $-0.510618$ and
$1.2367\times10^{-3}$. Its weight tends to $|\braket{u_{1s},g}|^2$
($0.292$ at $c=20$, against $0.298$); Fig.~\ref{fig:atom}(b) shows
both. These are the values of $S_c$; the $(1,0)$ resonances at these
values of $c$, computed independently with the continued fraction of
the next subsection and converted with $s=z+z^2/2mc^2$, give $-0.546665$ and $1.5257\times10^{-2}$ at
$c=10$, and $-0.510620$ and $1.2367\times10^{-3}$ at $c=20$.

\subsection{Gravitational atoms}\label{sec:atoms}

Corollary~\ref{cor:metastability} says that the hydrogenic states
become infinitely long-lived as $c\to\infty$. The rate is the subject of
a well-developed physics literature on ultralight bosons around black
holes~\cite{Detweiler1980,Dolan2007,ArvanitakiDubovsky2011,
BaumannChiaPorto2019,BaumannChiaStoutTerHaar2019,BritoCardosoPani2020},
where the quasi-bound states are poles of the continued Green's
function. In the present language they are resonances of $P_c$. Let
$\Lambda_l$ be the orthogonal projection onto angular momentum $l$; it
commutes with $P_c$, $\mathcal{U}_c$ and $H_S$.

\begin{theorem}[Resonance limit]\label{thm:resonances}
Fix $l\geq0$ and a non-zero $\chi\in C_c^\infty(\R^3\setminus\{0\})$, and
let $c$ be so large that $\supp\chi\subset\Sigma_c$.
\begin{enumerate}[label=(\alph*)]
\item The function $z\mapsto\chi\,\mathcal{U}_c(P_c-z)^{-1}\Lambda_l\,
\mathcal{U}_c^*\chi$, with values in the Hilbert--Schmidt operators on
$L^2(\R^3)$, continues meromorphically from $\Omega_c\cap\{\mathrm{Im}\,z
>0\}$ across $(-mc^2,0)$ to
\[
\Omega_c:=\Bigl\{z\in\C:\ \mathrm{Re}\,z>-mc^2,\
\mathrm{Im}\,z>-\frac{\hbar c^3}{4GM}\Bigr\}\setminus[0,\infty).
\]
Its poles in $\Omega_c$, the resonances of $P_c$ with angular momentum
$l$, do not depend on $\chi$.
\item Let $n\geq l+1$ and $0<\delta<\frac12\operatorname{dist}\bigl(E_n,
\sigma(H_S)\setminus\{E_n\}\bigr)$. For $c$ large there is exactly one
resonance $z_{nl}(c)$ with $|z_{nl}(c)-E_n|<\delta$; it is a pole of
order one, whose residue has rank $2l+1$.
\item With $\alpha:=GMm/(\hbar c)$,
\begin{align}
\mathrm{Re}\,z_{nl}(c)&=E_n+\frac{G^4M^4m^5}{\hbar^4c^2}\Bigl[
-\frac{1}{8n^4}+\frac{1}{n^3}\Bigl(\frac2n-\frac{6}{2l+1}\Bigr)\Bigr]
+o(c^{-2}),\label{eq:res-real}\\
\mathrm{Im}\,z_{nl}(c)&=-4\,C_{nl}\,(l!)^2\,mc^2\,\alpha^{4l+5}\,
\bigl(1+o(1)\bigr),\label{eq:res-imag}
\end{align}
with $C_{nl}=\dfrac{2^{4l+1}(n+l)!}{n^{2l+4}(n-l-1)!}\Bigl[\dfrac{l!}
{(2l)!\,(2l+1)!}\Bigr]^2$.
\end{enumerate}
\end{theorem}

Equations~\eqref{eq:res-real} and~\eqref{eq:res-imag} are the
Schwarzschild case of the spectrum and decay rates of
Ref.~\onlinecite{BaumannChiaStoutTerHaar2019}; the fine structure
in~\eqref{eq:res-real} is that of Ref.~\onlinecite{BaumannChiaPorto2019}.
Detweiler's original estimate of the rates~\cite{Detweiler1980} is larger
by a factor $2$, as noticed in
Ref.~\onlinecite{PaniCardosoGualtieriBertiIshibashi2012} and confirmed in
Ref.~\onlinecite{BaoXuZhang2022}; the theorem settles the normalisation.
The width scales as $c^{-(4l+3)}$.

The proof, in Appendix~\ref{app:resonances}, follows the scheme of
Ashbaugh and Harrell for eigenvalues of a limit problem that become
resonances under a singular perturbation of the
domain~\cite{AshbaughHarrell1982}: Jost solutions, a Wronskian and a
flux identity. Part~(a) is the construction of Jost solutions in each
partial wave; meromorphic continuations of this kind are standard, see
Ref.~\onlinecite{BachelotMotetBachelot1993} for massless fields on the
Schwarzschild exterior and Ref.~\onlinecite{DyatlovZworski2019}. In the tortoise coordinate the potential of $S_c$ is
exponentially close to $-mc^2/2$ towards the horizon, and energies near
$E_n<0$ lie below the continuum of the Coulomb tail. The central fact is
a matching lemma (Lemma~\ref{lem:matching}, Fig.~\ref{fig:atom}(a)): the radial function that is
ingoing at the horizon, where it behaves as $(r/r_{\rm s}-1)^{-\rmi
\varpi}$ with $\varpi$ as in Sec.~\ref{sec:schw-numerics}, divided by $\frac{(2l)!}{(l!)^2}r_{\rm s}^{-l}$, converges to
the radial function of the hydrogen equation that is regular at the
origin. This is the mechanism of Remark~\ref{rem:friedrichs} at the level
of solutions, and Hurwitz's theorem turns it into~(b). The width then
follows exactly from flux conservation. At the real zero of the real
part of the Jost function, its imaginary part is fixed by the flux of
the ingoing solution through the horizon, and with the normalisation of the matching lemma
this gives~\eqref{eq:res-imag} with its constant. In this form the width
is the imaginary part of a complex scattering length of the $l$-th
partial wave, which is of order $\varpi r_{\rm s}^{2l+1}$, times the
square of the coefficient of $r^l$ in the normalised hydrogenic radial
function $R_{nl}$. This is the structure of the level shifts
of Coulomb systems by short-range absorptive interactions, made rigorous
in Ref.~\onlinecite{AlbeverioGesztesyHoeghKrohnStreit1983}. The real
part follows from a Wronskian identity between the real standing wave at
the zero of the real part of the Jost function and the hydrogenic state,
whose leading term is the functional of Proposition~\ref{prop:PN}:
\begin{equation}\label{eq:res-real-D1}
\mathrm{Re}\,z_{nl}(c)=E_n+\frac{1}{c^2}\Bigl(D_1(u_{nl})-\frac{E_n^2}
{2m}\Bigr)+o(c^{-2}),
\end{equation}
where $u_{nl}$ is a real normalised eigenfunction of $H_S$ for $E_n$ with
angular momentum $l$ and $D_1$ is evaluated on the Schwarzschild metric.
The proposition does not apply to it, since $\Phi=-GM/r$ is singular,
but its formula does. In Schwarzschild coordinates $N_c^2=1+2\epsilon\Phi$
with $\Phi=-GM/r$, so $\varrho=0$, and $k=(2GM/r)\,\hat x\hat x$ with
$\hat x:=x/r$; then $\Phi+\frac12\operatorname{tr}k=0$ and
\begin{equation}\label{eq:D1-schw}
D_1(u)=-\frac{\hbar^2GM}{m}\int\frac{|\partial_ru|^2}{r}\,\rmd^3x
-2m\int\Phi^2u^2\,\rmd^3x+2E\int\Phi u^2\,\rmd^3x .
\end{equation}
For $u=R_{nl}(r)Y(\hat x)$, with $Y$ a real spherical harmonic of
degree $l$ normalised on $S^2$, and $E=E_n$, in units with $\hbar=m=GM=1$, one
integration by parts in the radial equation, together with the
hydrogenic values of Ref.~\onlinecite{BetheSalpeter1957},
$\langle r^{-1}\rangle=n^{-2}$, $\langle r^{-2}\rangle=2/((2l+1)n^3)$,
$\langle r^{-3}\rangle=2/(l(l+1)(2l+1)n^3)$ for $l\geq1$ and
$R_{n0}(0)^2=4/n^3$, gives
$\int_0^\infty rR_{nl}'(r)^2\,\rmd r=2/((2l+1)n^3)-1/n^4$. Hence
$D_1(u_{nl})-E_n^2/2$ is the bracket of~\eqref{eq:res-real} for every
$n>l\geq0$. Isotropic coordinates, with $\varrho=2\Phi^2$ and
$k=(2GM/r)\,\delta$, give the same value, as they must, since $D_1$ is
invariant under changes of coordinates $x\mapsto x+\epsilon\xi$
(Sec.~\ref{sec:stars}). The integrals converge although $\Phi$ is
singular at the origin. The horizon does not contribute to the real part
at order $c^{-2}$; it produces the width.

\begin{table}[htb]
\caption{Quasi-bound states of the massive scalar field on
Schwarzschild, $\omega$ in units of $\mu=mc^2/\hbar$.
$E_{\rm num}=(\mathrm{Re}\,\omega/\mu-1)/\alpha^2$ tends to $-1/2n^2$;
$\mathrm{FS}_{\rm num}=(E_{\rm num}+1/2n^2)/\alpha^2$ is compared with the
bracket of~\eqref{eq:res-real} ($\mathrm{FS}_{\rm pred}$). The width
$\Gamma:=4C_{nl}\prod_{j=1}^{l}(j^2+4r_{\rm s}^2(\mathrm{Re}\,\omega)^2
/c^2)\,\hbar\,\mathrm{Re}\,\omega\,\alpha^{4l+5}>0$ is minus the leading
term of~\eqref{eq:res-imag}, with $(l!)^2$ replaced by the product and $mc^2$
by $\hbar\,\mathrm{Re}\,\omega$, as in
Ref.~\onlinecite{BaumannChiaStoutTerHaar2019}; the last column is the
local exponent of $|\mathrm{Im}\,\omega|/\mu$ in $\alpha$ between
consecutive rows, expected $4l+5$. The rows $(3,2)$ are computed with
$40$ significant digits (Appendix~\ref{app:numerics}). All entries are
dimensionless.}
\label{tab:atoms}
\begin{ruledtabular}
\begin{tabular}{rrrrrrr}
$(n,l)$ & $\alpha$ & $E_{\rm num}$ & $\mathrm{FS}_{\rm num}$ &
$\mathrm{FS}_{\rm pred}$ & $-\hbar\,\mathrm{Im}\,\omega/\Gamma$ & exponent\\
\hline
(1,0) & 0.010 & $-0.500414$ & $-4.140$ & $-4.125$ & 1.013 & ---\\
(1,0) & 0.020 & $-0.501671$ & $-4.177$ & $-4.125$ & 1.044 & 5.04\\
(1,0) & 0.040 & $-0.506877$ & $-4.298$ & $-4.125$ & 1.158 & 5.15\\
(2,1) & 0.020 & $-0.125053$ & $-0.1329$ & $-0.1328$ & 1.018 & ---\\
(2,1) & 0.030 & $-0.125120$ & $-0.1331$ & $-0.1328$ & 1.039 & 9.07\\
(2,1) & 0.040 & $-0.125213$ & $-0.1333$ & $-0.1328$ & 1.066 & 9.13\\
(3,1) & 0.030 & $-0.055601$ & $-0.0510$ & $-0.0509$ & 1.039 & ---\\
(3,1) & 0.040 & $-0.055637$ & $-0.0511$ & $-0.0509$ & 1.067 & 9.13\\
(3,2) & 0.060 & $-0.055632$ & $-0.0214$ & $-0.0213$ & 1.091 & ---\\
(3,2) & 0.080 & $-0.055693$ & $-0.0214$ & $-0.0213$ & 1.157 & 13.39\\
\end{tabular}
\end{ruledtabular}
\end{table}

Table~\ref{tab:atoms} compares the theorem with Leaver's continued
fraction~\cite{Leaver1985} in the form of Ref.~\onlinecite{Dolan2007},
solved for the decay constant $q=-(\mu^2-\omega^2)^{1/2}$ to avoid
cancellations (Appendix~\ref{app:numerics}). The real parts converge to
the hydrogenic levels with the fine-structure coefficient
of~\eqref{eq:res-real}. The ratios of the widths to $\Gamma$ tend to $1$,
and the local exponents decrease towards $4l+5$ as $\alpha$ decreases;
the one exponent for $l=2$ is taken at larger $\alpha$. The slow convergence of $\mathrm{FS}_{\rm num}$
for $s$-waves has a definite source: in the proof of the real part the
term with $\varpi_0^2$ in~\eqref{eq:real-identity}, which comes from the
horizon, contributes through the overlap zone a term of relative order
$r_{\rm s}\ln(1/r_{\rm s})$,
which in these units shifts $\mathrm{FS}$ by about
$-4a_{n0}^2\alpha^2\ln(1/2\alpha^2)$, with $a_{n0}^2=4/n^3$; for $(1,0)$
this is $-0.014$, $-0.046$ and $-0.147$ at the three values of $\alpha$,
against $-0.015$, $-0.052$ and $-0.173$ in the table. The Lorentzians of
Sec.~\ref{sec:schw-numerics},
obtained without continued fractions, agree to $0.1\%$ in position and
width with the $(1,0)$ poles at $\alpha=0.1$ and $0.05$, computed in the
same way and converted to $S_c$ as there; the shift from $E_1$ agrees to $0.8\%$ at $c=10$
and to $2\times10^{-4}$ at $c=20$. Figure~\ref{fig:atom} shows the scales of the problem and the
spectral concentration.

\begin{figure}[htb]
\centering
\begin{tikzpicture}[font=\footnotesize, >={Stealth[length=2mm]}]
\begin{scope}
\def\lx#1{(#1+3)*1.3}   
\draw[->] ({\lx{-3}},0) -- ({\lx{2.35}},0) node[right] {$r/a_0$};
\foreach \k in {-3,-2,-1,0,1,2} {\draw ({\lx{\k}},0) -- ({\lx{\k}},-0.08);}
\foreach \k in {-3,-2,-1,0,1,2} {\node[below, font=\scriptsize] at ({\lx{\k}},-0.08) {$10^{\k}$};}
\draw[very thick] ({\lx{-2.301}},0) -- ({\lx{-2.301}},0.35) node[above] {$r_{\rm s}$};
\draw ({\lx{-1.301}},0) -- ({\lx{-1.301}},0.35) node[above] {$\hbar/mc$};
\draw ({\lx{0}},0) -- ({\lx{0}},0.35) node[above] {$a_0$};
\draw ({\lx{0.602}},0) -- ({\lx{0.602}},0.35) node[above] {$4a_0$};
\draw[decorate, decoration={brace, amplitude=4pt, mirror}]
  ({\lx{-2}},-0.5) -- ({\lx{-1.7}},-0.5)
  node[midway, below=5pt, align=center, font=\scriptsize] {near\\ zone\\ (B5)};
\draw[decorate, decoration={brace, amplitude=4pt, mirror}]
  ({\lx{-1.55}},-0.5) -- ({\lx{0}},-0.5)
  node[midway, below=5pt, align=center, font=\scriptsize] {matching (B6)\\ $r_{\rm s}\ll r\ll a_0$};
\draw[decorate, decoration={brace, amplitude=4pt, mirror}]
  ({\lx{0.602}},-0.5) -- ({\lx{2}},-0.5)
  node[midway, below=5pt, align=center, font=\scriptsize] {far zone (B7)\\ $r>R_1$};
\node[anchor=west, font=\footnotesize] at ({\lx{-3}},2.75) {(a) scales at $\alpha=0.05$};
\end{scope}
\begin{scope}[xshift=9.4cm, yshift=-1.3cm]
\begin{axis}[width=8.2cm, height=5.2cm, ymode=log,
  ytick={0.01,0.1,1,10,100},
  xmin=-0.6, xmax=-0.45, ymin=0.01, ymax=200,
  xlabel={$E/mc^2\alpha^2$}, ylabel={density $\times\,mc^2\alpha^2$},
  xtick={-0.6,-0.55,-0.5,-0.45},
  xticklabel style={/pgf/number format/fixed, /pgf/number format/precision=2},
  legend style={font=\scriptsize, draw=none, fill=none, at={(0.02,0.98)}, anchor=north west},
  tick label style={font=\scriptsize}, label style={font=\footnotesize},
  title={(b) spectral density near $E_1$}, title style={font=\footnotesize}]
\addplot[domain=-0.6:-0.45, samples=400, thick, dashed]
  {0.27254/pi*0.015265/((x+0.546304)^2+0.015265^2)};
\addlegendentry{$\alpha=0.1$}
\addplot[domain=-0.6:-0.45, samples=800, thick]
  {0.29244/pi*0.0012367/((x+0.510618)^2+0.0012367^2)};
\addlegendentry{$\alpha=0.05$}
\draw[->, very thick, gray] (axis cs:-0.5,0.012) -- (axis cs:-0.5,150);
\end{axis}
\end{scope}
\end{tikzpicture}
\caption{The gravitational atom. (a)~The scales $r_{\rm s}=2\alpha^2a_0$,
$\hbar/mc=\alpha a_0$ and $a_0$, in units of the Bohr radius
$a_0=\hbar^2/GMm^2$, for $\alpha=GMm/\hbar c=0.05$. In the proof of
Theorem~\ref{thm:resonances} the solution that is ingoing at the
horizon is close, for $2r_{\rm s}\leq r\leq Xr_{\rm s}$, to the Legendre
polynomial in $2r/r_{\rm s}-1$ (Lemma~\ref{lem:near}), and across the
overlap $r_{\rm s}\ll r\ll a_0$ it is matched with the regular
hydrogenic solution (Lemma~\ref{lem:matching}); the solution that
decays at infinity is controlled for $r>R_1>4n^2a_0$
(Lemma~\ref{lem:far}; $4a_0$ is this bound for $n=1$), and the two are
compared at $r=a_0$. (b)~Spectral density
$\pi^{-1}\mathrm{Im}\braket{\mathcal{U}_c^*g,(S_c-E-\rmi0)^{-1}
\mathcal{U}_c^*g}$ in the $s$-wave, for the test function $g$ of
Appendix~\ref{app:numerics}, with $\|g\|^2=0.539$: the Lorentzians
fitted in Sec.~\ref{sec:schw-numerics}, without their linear
background, at $\alpha=0.1$ (dashed) and $\alpha=0.05$ (solid), and the
limit, a Dirac measure at $E_1=-\frac12$ of weight
$|\braket{u_{1s},g}|^2=0.298$ (grey arrow). Their centres and
half-widths agree to $0.1\%$, and their shifts from $E_1$ to $0.8\%$ at
$\alpha=0.1$, with the $(1,0)$ resonances at the same $\alpha$,
computed with the continued fraction of Sec.~\ref{sec:atoms} (they are
not in Table~\ref{tab:atoms}) and converted to $S_c$ with
$s=z+z^2/2mc^2$. The energy $E$ is the spectral variable of $S_c$, in
units of $mc^2\alpha^2=G^2M^2m^3/\hbar^2$.}
\label{fig:atom}
\end{figure}

The theorem locates the poles but does not control the continued
resolvent away from them uniformly in $c$. That control is what an
exponential decay law of the kind proved in
Refs.~\onlinecite{Hunziker1990,SofferWeinstein1998} would need: for $u$ of definite angular
momentum $l$, the amplitude of Corollary~\ref{cor:metastability} would be
$\rme^{-\rmi tz_{nl}(c)/\hbar}$ up to errors that vanish as $c\to\infty$
uniformly in $t\geq0$, including the times $t\sim\hbar/|\mathrm{Im}\,
z_{nl}(c)|$ at which the state decays. Rigorous analyses of the same
radial equation, for other questions, are the late-time tails of
Ref.~\onlinecite{PasqualottoShlapentokhRothmanVandeMoortel2023} and the
quasimodes at large angular momentum of
Ref.~\onlinecite{ShlapentokhRothmanVandeMoortel2023}; Hintz and Xie treat
a two-scale limit of the same kind for small Schwarzschild--de~Sitter
black holes, with the field mass at the cosmological
scale~\cite{HintzXie2022}.

\section{The limit at the level of nets}\label{sec:net}

\subsection{The time-zero net is independent of $c$}

Let $\mathfrak W$ be the Weyl algebra over the symplectic space
$(\mathcal S(\R^3),\hbar\,\mathrm{Im}\braket{\cdot,\cdot})$, with
generators $W(h)$ and relations
\[
W(h)W(h')=\rme^{-\rmi\hbar\,\mathrm{Im}\braket{h,h'}/2}\,W(h+h'),
\]
and local subalgebras $\mathfrak W(B)$ generated by the $W(h)$ with
$h\in C_c^\infty(B)$, for $B\subset\R^3$ open. In a regular
representation write $W(h)=\rme^{\rmi\Phi_{\rm S}(h)}$ and
$\Phi_{\rm S}(h)=(\hat\psi_0(h)+\hat\psi_0(h)^*)/\sqrt2$, so that
$[\hat\psi_0(h),\hat\psi_0(h')^*]=\hbar\braket{h,h'}$: $\mathfrak W(B)$ is
the time-zero algebra of a canonical field in the sense of
Ref.~\onlinecite{pachon2026algebraicquantumkinematicssrselection}
(Definition~1 there).

The Klein--Gordon time-zero fields provide such a field for every $c$.
With $\kappa:=mc^2/\hbar$ put
\begin{equation}\label{eq:local-field}
\hat\psi_{0,c}(h):=\sqrt{\kappa/2}\,\bigl(\varphi(\bar h)+\rmi\kappa^{-1}
\pi(\bar h)\bigr),
\end{equation}
where $\varphi$ and $\pi$ are extended complex-linearly. Then
$[\hat\psi_{0,c}(h),\hat\psi_{0,c}(h')^*]=\hbar\braket{h,h'}$ and
$[\hat\psi_{0,c}(h),\hat\psi_{0,c}(h')]=0$, and $\hat\psi_{0,c}(h)$ is
local: it involves only $\varphi$ and $\pi$ smeared in $\supp h$; it is
the time-zero upper component of the Feshbach--Villars
decomposition~\cite{FeshbachVillars1958}, normalised canonically. For
$h=u+\rmi v$ with $u,v$ real, $\Phi_{\rm S}(h)=\sqrt\kappa\,\varphi(u)+
\pi(v)/\sqrt\kappa$. We transport every state of the Klein--Gordon field
to $\mathfrak W$ through this identification; in Minkowski space
$\K_c=L^2(\R^3)$ and no density map is needed.

\begin{theorem}[Time-zero limit of the free Klein--Gordon net]
\label{thm:net}
Let $\omega_c$ be the Minkowski vacuum of the free Klein--Gordon field of
mass $m$, transported to $\mathfrak W$, and let $\omega_\infty$ be the
Fock vacuum of the free Schr\"odinger field, $\omega_\infty(W(h))=
\rme^{-\hbar\|h\|^2/4}$. The relativistic Heisenberg evolution acts by
$W(h)\mapsto W(\tilde T^{(c)}_th)$ for a real-linear symplectic map
$\tilde T^{(c)}_t$ of $\mathcal S(\R^3)$; let $T^{(c)}_t:=\rme^{-\rmi\kappa
t}\tilde T^{(c)}_t$, the evolution followed by the phase rotation of the
local field~\eqref{eq:local-field}, and write $T^{(c)}_t h=A^{(c)}_th+B^{(c)}_t\bar h$ with complex-linear
$A^{(c)}_t$, $B^{(c)}_t$. Then:
\begin{enumerate}[label=(\alph*)]
\item $|\omega_c(W(h))-\omega_\infty(W(h))|\leq\dfrac{\hbar^3}{8m^2c^2}\,
\|\nabla h\|^2$ for every $h\in\mathcal S(\R^3)$;
\item $A^{(c)}_th\to\rme^{\rmi tH_0/\hbar}h$ and $B^{(c)}_th\to0$ in
$L^2(\R^3)$, uniformly for $t$ in compact sets, for every
$h\in\mathcal S(\R^3)$, where $H_0=-\frac{\hbar^2}{2m}\Delta$; the limit
is the automorphism $W(h)\mapsto W(\rme^{\rmi tH_0/\hbar}h)$ of
$\mathfrak W$, the evolution of the free Schr\"odinger field;
\item for every finite $c$ and every non-empty bounded open
$B\subset\R^3$ with smooth boundary, the restrictions of $\omega_c$ and $\omega_\infty$ to
$\mathfrak W(B)$ are disjoint: $\pi_{\omega_c}(\mathfrak W(B))''$ is a
factor of type III$_1$, whereas $\pi_{\omega_\infty}(\mathfrak W(B))''$ is a
factor of type I on which $\omega_\infty$ is pure;
\item for $R,T>0$, the von Neumann algebra of the cylinder
$B_R\times(-T,T)$ in the vacuum representation of the Klein--Gordon field
is $\pi_{\omega_c}(\mathfrak W(B_{R+cT}))''$. In particular, for
every bounded $B\subset\R^3$ it contains $\pi_{\omega_c}(\mathfrak W(B))$
as soon as $c$ is so large that $B\subset B_{R+cT}$.
\end{enumerate}
\end{theorem}

\begin{proof}
In the vacuum, $\varphi(u)=\sqrt{\hbar/2}\,(a(K^{-1/4}u)+a^*(K^{-1/4}u))$
and $\pi(v)=\rmi\sqrt{\hbar/2}\,(a^*(K^{1/4}v)-a(K^{1/4}v))$ with
$K^{1/2}=(\kappa^2-c^2\Delta)^{1/2}$. Hence $\Phi_{\rm S}(h)=\sqrt{\hbar/2}\,
(a(k_h)+a^*(k_h))$ with
\begin{equation}\label{eq:kh}
k_h=\gamma\,u+\rmi\gamma^{-1}v,\qquad \gamma:=\sqrt\kappa\,K^{-1/4}
=\bigl(1-\hbar^2\Delta/m^2c^2\bigr)^{-1/4},
\end{equation}
and $\omega_c(W(h))=\exp(-\frac\hbar4\|k_h\|^2)$ with $\|k_h\|^2=
\braket{u,\gamma^2u}+\braket{v,\gamma^{-2}v}$.

(a) In momentum space $\gamma^{\pm2}=(1+\hbar^2p^2/m^2c^2)^{\mp1/2}$, and
$|\gamma^{\pm2}-1|\leq\hbar^2p^2/(2m^2c^2)$. Hence $|\|k_h\|^2-\|h\|^2|
\leq\frac{\hbar^2}{2m^2c^2}\|\nabla h\|^2$, and $|\rme^{-x}-\rme^{-y}|\leq
|x-y|$ for $x,y\geq0$.

(b) With $\nu_\pm:=\frac12(\gamma\pm\gamma^{-1})$, even real functions of
$p$ with $\nu_+^2-\nu_-^2=1$, the map $\mathcal J h:=k_h=\nu_+h+\nu_-\bar
h$ has inverse $k\mapsto\nu_+k-\nu_-\bar k$, and $\hat\psi_{0,c}(h)=
\sqrt\hbar\,(a(\nu_+h)+a^*(\nu_-\bar h))$ is a Bogoliubov transformation of the vacuum annihilation
operators. These evolve as $a(k)\mapsto a(\rme^{\rmi t\omega}k)$ with
$\omega(p)=(\kappa^2+c^2p^2)^{1/2}$, so $\tilde T^{(c)}_t=\mathcal J^{-1}
\rme^{\rmi t\omega}\mathcal J$, and in momentum space
\[
A^{(c)}_t=\nu_+^2\rme^{\rmi t(\omega-\kappa)}-\nu_-^2
\rme^{-\rmi t(\omega+\kappa)},\qquad
B^{(c)}_t=2\rmi\,\nu_+\nu_-\,\sin(\omega t)\,\rme^{-\rmi\kappa t}.
\]
As $c\to\infty$, $\omega-\kappa\to\hbar p^2/2m$, $\nu_+\to1$ and
$\nu_-=O(\hbar^2p^2/m^2c^2)$ pointwise, and $\nu_+^2+\nu_-^2\leq
1+\hbar^2p^2/(m^2c^2)$ gives a polynomial bound uniform in $c\geq c_0$.
Dominated convergence gives (b).

(c) By the Reeh--Schlieder theorem the vacuum is cyclic for
$\pi_{\omega_c}(\mathfrak W(B))''$, so this is the GNS von Neumann algebra
of $\omega_c|\mathfrak W(B)$. It is the von Neumann algebra of the domain
of dependence of $B$~\cite{Araki1963}, which for the free field is a
factor of type III~\cite{Araki1964b}, and of type
III$_1$~\cite{BuchholzDAntoniFredenhagen1987,Verch1997}. In the Fock
representation of $\omega_\infty$, $\Gamma(L^2(\R^3))=\Gamma(L^2(B))
\otimes\Gamma(L^2(\R^3\setminus B))$ and $\pi_{\omega_\infty}(\mathfrak
W(B))''=\mathcal{B}(\Gamma(L^2(B)))\otimes\one$; the cyclic subspace of
the vacuum is $\Gamma(L^2(B))\otimes\Omega$, so the GNS von Neumann
algebra of $\omega_\infty|\mathfrak W(B)$ is $\mathcal{B}(\Gamma(L^2(B)))$,
on which $\omega_\infty$ is pure. Two factor states whose GNS von Neumann
algebras have different types are not quasi-equivalent, hence
disjoint.

(d) A point $(x,t)$ is spacelike to the whole open cylinder iff
$|x|-R\geq c(|t|+T)$, which is also the causal complement of the double
cone $D:=\{|x|+c|t|<R+cT\}$ over $B_{R+cT}$; so $D$ is the causal
completion of the cylinder. By the timelike tube theorem of Borchers and
Araki~\cite{Borchers1961,StrohmaierWitten2023}, the algebra of the
cylinder is that of its timelike envelope, the region swept by the
timelike curves in the cylinder when they are deformed with fixed
endpoints. Here the envelope is $D\cap\{|t|<T\}$, which is strictly
smaller than $D$ for all $R,T>0$: it misses the caps $T\leq|t|<T+R/c$,
which are small only when $R\ll cT$. It is a neighbourhood of the Cauchy
surface $B_{R+cT}\times\{0\}$ of $D$, and the free field has the
time-slice property~\cite{Dimock1980}, so its algebra is that of $D$. By the proof of (c),
this is $\pi_{\omega_c}(\mathfrak W(B_{R+cT}))''$.
\end{proof}

Theorem~\ref{thm:net} makes precise in which sense the Galilean theory
is a limit of the relativistic one. The time-zero algebras are the same
for every $c$; the states converge on every Weyl operator, and the
dynamics converges on the one-particle level. Yet the local algebras
change type in the limit: at every finite $c$ the vacuum restricted to
$\mathfrak W(B)$ is a faithful state of a type III$_1$ factor, and in the
limit it is a pure product state of a type I factor. The vacuum
entanglement across $\partial B$, which the type III$_1$ property
encodes, is lost. The convergence in (a) is pointwise on the Weyl
operators, not in the norm of states: by Shale's
criterion~\cite{Shale1962} the
Bogoliubov transformation with coefficients $\nu_-(p)$, which is not
Hilbert--Schmidt, is not unitarily implementable, so $\omega_c$ and
$\omega_\infty$ are inequivalent on the whole of $\mathfrak W$ for every
$c$. Figure~\ref{fig:types} summarises the situation.

\begin{figure}[htb]
\centering
\begin{tikzpicture}[font=\footnotesize, >={Stealth[length=2mm]},
  box/.style={rounded corners=2pt, align=center, inner sep=3pt,
    text width=40mm, draw,
    execute at begin node={\hyphenpenalty=10000\exhyphenpenalty=10000}},
  arr/.style={->, gray!70!black, line width=0.5pt}]
\node[box, fill=gray!5] (w) at (4.6,0)
  {$\mathfrak W(B)$, the same Weyl algebra for every $c$, generated by the
   local field~\eqref{eq:local-field}};
\node[box, fill=blue!12, line width=0.8pt] (rel) at (1.9,-2.3)
  {$c<\infty$: type III$_1$ for the vacuum and, in Minkowski space, for
   the thermal states; vacuum faithful, entangled across $\partial B$};
\node[box, fill=blue!12, line width=0.8pt] (nr) at (7.3,-2.3)
  {$c=\infty$: type I for the vacuum and, in flat space, for the
   thermal states; vacuum pure, not separating; thermal states faithful};
\draw[arr] (w) -- node[left, font=\scriptsize] {GNS} (rel);
\draw[arr] (w) -- node[right, font=\scriptsize] {GNS} (nr);
\draw[arr, line width=0.8pt] (rel.south) -- ++(0,-0.5) -| node[pos=0.25, below,
  font=\scriptsize, align=center, text=black]
  {$\omega_c\to\omega_\infty$, $\omega_c^{\beta,\mu}\to\omega^{\beta,\mu}$ on
   every $W(h)$, dynamics converge;\\
   the vacua are disjoint for every finite $c$} (nr.south);
\node[anchor=west] at (0,0.95) {(a)};
\begin{scope}[xshift=11.2cm, yshift=-1.3cm]
\def\R{0.9}\def\T{0.8}\def\cT{1.35}  
\def\cTT{2.5}                          
\draw[->] (-3.6,0) -- (3.7,0) node[right] {$x$};
\draw[->] (0,-1.85) -- (0,1.95) node[below left] {$t$};
\draw[dashed] (-\R-\cT,0) -- (0,{\T+\R*\T/\cT}) -- (\R+\cT,0) -- (0,{-\T-\R*\T/\cT}) -- cycle;
\fill[blue!12] (-\R,\T) -- (\R,\T) -- (\R+\cT,0) -- (\R,-\T) -- (-\R,-\T) -- (-\R-\cT,0) -- cycle;
\draw[line width=0.6pt] (-\R,\T) -- (\R,\T) -- (\R+\cT,0) -- (\R,-\T) -- (-\R,-\T) -- (-\R-\cT,0) -- cycle;
\fill[gray!35] (-\R,-\T) rectangle (\R,\T);
\draw (-\R,-\T) rectangle (\R,\T);
\draw[line width=1.6pt, gray!55] (-\R-\cTT,0) -- (\R+\cTT,0);
\draw[line width=1.6pt] (-\R-\cT,0) -- (\R+\cT,0);
\node[font=\scriptsize] at (0,0.3) {$B_R\times(-T,T)$};
\node[font=\scriptsize, anchor=south west] at (\R+\cT,0.02) {$R+cT$};
\node[font=\scriptsize, anchor=north, gray] at (\R+\cTT,-0.08) {$R+c'T$};
\node[anchor=west] at (-3.3,2.25) {(b)};
\end{scope}
\end{tikzpicture}
\caption{Time-zero algebras and their types (Theorem~\ref{thm:net},
Propositions~\ref{prop:open-regions} and~\ref{prop:thermal}). (a)~The
Weyl algebra of a ball $B$ is the same for every $c$. The vacuum states,
the thermal states with the chemical potential shifted by the rest
energy, and the dynamics converge on it, yet the von Neumann algebras
it generates in these representations change type, from III$_1$ to I,
in the vacuum and, in flat space, in the thermal states alike
(Sec.~\ref{sec:thermal}). What is special to
the vacuum is that its limit is pure, hence not faithful, so its
entanglement across $\partial B$ is lost; the thermal states stay
faithful and converge together with their global modular groups. The
vacuum states are disjoint for every finite $c$. (b)~Schematic, in one
space dimension, with edges of slope $1/c$: the cylinder
$B_R\times(-T,T)$ (grey); its timelike envelope (blue hexagon), whose
algebra is that of the cylinder by the timelike tube theorem and which
contains the base $B_{R+cT}\times\{0\}$; and the double cone over
$B_{R+cT}$ (dashed), whose algebra is that of the hexagon by the
time-slice property of the free field. So the time-zero algebra of
$B_{R+cT}$ is that of the cylinder (Theorem~\ref{thm:net}(d)). For a
larger speed of light $c'$ the base $B_{R+c'T}$ is larger (grey line).
As $c\to\infty$ it exhausts $\R^3$, and in the vacuum representation of
the limit net the algebra of every open region is $\mathcal{B}(\Hil)$
(Proposition~\ref{prop:open-regions}); in the thermal representations
in flat space it is the global factor, of type III$_1$
(Ref.~\onlinecite{pachon2026galileanreehschliederobstruction},
Theorem~6).}
\label{fig:types}
\end{figure}

The dynamics generated by $\rmd\Gamma(P_c)$ (Sec.~\ref{sec:rescaled-field})
acts on $\mathfrak W$ through $\mathcal J^{-1}\rme^{\rmi t(\omega-\kappa)}
\mathcal J$, and the argument of (b) shows that it converges to the same
limit. Both it and $T^{(c)}_t$ are the relativistic evolution followed by
a rotation by the rest energy: of the particles in the first case, of
the local field, $h\mapsto\rme^{-\rmi\kappa t}h$, in the second. The
rotation of the local field preserves every $\mathfrak W(B)$ but, for
$\sin\kappa t\neq0$, is not unitarily implementable in the vacuum
representation: its antilinear part, $2\rmi\nu_+\nu_-\sin(\kappa t)$, is
unbounded. The convergence in (b) is therefore a convergence of the
one-particle maps in $L^2(\R^3)$, vector by vector, not of unitary
groups in the Klein--Gordon representation. It gives strong convergence
of $W(T^{(c)}_th)$ in every representation in which $h\mapsto W(h)$ is
continuous for the $L^2$ norm, such as the Fock representation of
$\omega_\infty$; in the C$^*$-norm of $\mathfrak W$ distinct Weyl
operators are at distance $2$.

The same holds on static backgrounds.

\begin{proposition}[Static backgrounds]\label{prop:net-static}
For the stars of Theorem~\ref{thm:stars} and for the Boulware state of
Theorem~\ref{thm:schw}, define $\hat\psi_{0,c}$ by~\eqref{eq:local-field}
on $\K_c$ and transport it to $L^2(\R^3)$ with $\mathcal{U}_c$. Then (a)
of Theorem~\ref{thm:net} holds in the form $|\omega_c(W(h))-
\omega_\infty(W(h))|\leq C_hc^{-2}$, and (b) holds without rate, with
$H_S$ in place of $H_0$, for $h\in C_c^\infty(\R^3)$ in the case of stars
and for $h\in C_c^\infty(\R^3\setminus\{0\})$ on the Schwarzschild
exterior; in (b) the limit is the automorphism $W(h)\mapsto W(\rme^{\rmi tH_S/\hbar}h)$
of the Weyl algebra over $L^2(\R^3)$. If the metric is smooth, (c) holds
as well, for bounded open $B$ with smooth boundary and $\bar B\subset
\Sigma_c$.
\end{proposition}

\begin{proof}
Now $\gamma=\gamma_c=(1+P_c/mc^2)^{-1/2}$ and $\gamma_c^{-2}=
(1+2S_c/mc^2)^{1/2}$. Everything reduces to the strong convergence to
$h$ of $\mathcal{U}_c\gamma_c^{j}\mathcal{U}_c^*h$ for $j=\pm1,\pm2$.
For $j=1$ this is the argument in the proofs of
Theorems~\ref{thm:stars}(b) and~\ref{thm:schw}(d), and so is $j=2$ for
stars, where $\gamma_c$ is bounded. On the Schwarzschild exterior, with
$w=\mathcal{U}_c^*h$, $\|(\gamma_c^2-\one)w\|^2=\|\gamma_c^2w\|^2-2\|
\gamma_cw\|^2+\|w\|^2$. By~\eqref{eq:K-lower} and operator monotonicity
of $t\mapsto-t^{-1}$, $\|\gamma_c^2w\|^2\leq\int N_c^{-2}|h|^2\,\rmd^3x
\to\|h\|^2$; by the case $j=1$, $\|\gamma_cw\|\to\|h\|$; and $\|w\|=
\|h\|$. Hence $(\gamma_c^2-\one)w\to0$. For the negative powers, $|(1+2x)^{1/2}-1|^2\leq2|x|$ and $|(1+2x)^{1/4}-1|^2\leq2|x|$ for
$x\geq-\frac12$, so with $w=\mathcal{U}_c^*h$
\[
\|(\gamma_c^{-2}-\one)w\|^2+\|(\gamma_c^{-1}-\one)w\|^2\leq
\frac{4}{mc^2}\braket{w,|S_c|w}.
\]
The right-hand side tends to $0$: for stars $S_c\geq-C$ uniformly, so
$\braket{w,|S_c|w}\leq q_c(\mathcal{U}_cw)+2C\|w\|^2$ is bounded; on the
Schwarzschild exterior $\braket{w,|S_c|w}\leq\|w\|\,\|S_cw\|$, which is
bounded because $h$ is supported away from the horizon and the
coefficients of $\mathcal{U}_cS_c\mathcal{U}_c^*$ converge there. This
gives (a), and the same bounds give its rate. For $h=u+\rmi v$ put
$w_u:=\mathcal{U}_c^*u$ and $w_v:=\mathcal{U}_c^*v$; as in the proof of
Theorem~\ref{thm:net}, $\omega_c(W(h))=\exp(-\frac\hbar4(\braket{w_u,
\gamma_c^2w_u}+\braket{w_v,\gamma_c^{-2}w_v}))$, and $\omega_\infty(W(h))=
\exp(-\frac\hbar4\|h\|^2)$. Since $\gamma_c^{-2}-\one=P_c/mc^2$ and
$|P_c|\leq2|S_c|$, the second term differs from $\|v\|^2$ by at most
$2\braket{w_v,|S_c|w_v}/mc^2=O(c^{-2})$. For the first,
$\gamma_c^2-\one\geq-(P_c)_+/mc^2\geq-|S_c|/mc^2$, while $\gamma_c^2\leq
N_c^{-1}$ by~\eqref{eq:K-lower} and operator monotonicity, so
$-O(c^{-2})\leq\braket{w_u,(\gamma_c^2-\one)w_u}\leq\int(N_c^{-1}-1)|u|^2
\,\rmd^3x=O(c^{-2})$, because $N_c^{-1}-1=O(c^{-2})$ on the support of
$h$. For (b), $A^{(c)}_t=\nu_{c,+}^2\rme^{\rmi tP_c/\hbar}-
\nu_{c,-}^2\rme^{-\rmi t(P_c+2mc^2)/\hbar}$ and $B^{(c)}_t$ is
$\nu_{c,+}\nu_{c,-}$ times $\rme^{\rmi tP_c/\hbar}-\rme^{-\rmi t(P_c+
2mc^2)/\hbar}$, an operator of norm at most $2$ that commutes with it,
with
$\nu_{c,\pm}=\frac12(\gamma_c\pm\gamma_c^{-1})$; the powers above give
$\nu_{c,-}^2w\to0$, $\nu_{c,+}\nu_{c,-}w\to0$ and $\nu_{c,+}^2w-w\to0$,
and
Theorems~\ref{thm:stars}(b) and~\ref{thm:schw}(a) give the convergence
of $\mathcal{U}_c\rme^{\rmi tP_c/\hbar}\mathcal{U}_c^*$. For smooth
metrics the ground state is Hadamard (Theorem~\ref{thm:stars}(c); for
the Boulware state on the exterior region see
Ref.~\onlinecite{SahlmannVerch2000}), and its local algebras are type
III$_1$ factors~\cite{Verch1997}. The time-zero algebra of $B$ is that of
its domain of dependence by the time-slice property of the free
field~\cite{Dimock1980}.
The argument of (c) then applies without the Reeh--Schlieder theorem: if
$\mathcal M=\pi_{\omega_c}(\mathfrak W(B))''$ is a factor and $P\in
\mathcal M'$ is the projection onto the cyclic subspace of the vacuum,
then $x\mapsto xP$ is an isomorphism of $\mathcal M$ onto $\mathcal M_P$,
because the central support of $P$ is $\one$, so the GNS von Neumann
algebra of the restricted state has the type of $\mathcal M$.
\end{proof}

\subsection{Space-time algebras of the limit net}\label{sec:open-regions}

The space-time algebras of the limit nets are maximal, in flat space and
on the static backgrounds of Secs.~\ref{sec:stars} and~\ref{sec:schw}.

\begin{proposition}[Open regions of the Schr\"odinger net]
\label{prop:open-regions}
Let $H$ be $H_0$, the operator $H_S$ of Theorem~\ref{thm:stars}, or the
operator $H_S$ of Theorem~\ref{thm:schw}. For the free Schr\"odinger
field on $\Gamma(L^2(\R^3))$ with Hamiltonian $\rmd\Gamma(H)$, let
$\F(\Op)$ be the von Neumann algebra generated by the closures of
$\hat\psi(f)$ and $\hat\psi(f)^*$, $f\in C_c^\infty(\Op)$. Then
$\F(\Op)=\mathcal{B}(\Gamma(L^2(\R^3)))$ for every non-empty open
$\Op\subset\R\times\R^3$.
\end{proposition}

\begin{proof}
Let $(t_0-T,t_0+T)\times B\subset\Op$, with $B\subset\R^3$ an open
ball; by time-translation covariance we may take $t_0=0$. For $\chi\in
C_c^\infty((-T,T))$ and $g\in C_c^\infty(B)$, $\hat\psi(\chi\otimes g)=
\sqrt\hbar\,a(k)$ with $k=\int\chi(t)\,\rme^{\rmi tH/\hbar}g\,\rmd t$.
Suppose $w\in L^2(\R^3)$ is orthogonal to all such $k$. Then the
continuous function $t\mapsto\braket{w,\rme^{\rmi tH/\hbar}g}=
\braket{\rme^{-\rmi tH/\hbar}w,g}$ vanishes on $(-T,T)$ for every $g\in
C_c^\infty(B)$, so $\rme^{-\rmi tH/\hbar}w$ vanishes almost everywhere on
$B$ for $|t|<T$. The operator $L:=2mH/\hbar^2$ is self-adjoint on
$H^2(\R^3)$ and bounded below, $\rme^{-\rmi tH/\hbar}=\rme^{-\rmi(\hbar
t/2m)L}$, and $Lu=-\Delta u+Vu$ for $u\in H^2(\R^3)$, with $V=0$ for
$H_0$, $V=2m^2\Phi/\hbar^2$, which is bounded, for the stars, and $V=
-2GMm^2/(\hbar^2|x|)$, which is bounded outside every ball centred at the
origin, on the Schwarzschild exterior. Lemma~A7 of
Ref.~\onlinecite{pachon2026algebraicquantumkinematicssrselection}, whose
$h$ and $W$ are here $L$ and $V$, applies with $U=B$ and with $\Omega_0=
\R^3$ in the first two cases and $\Omega_0=\R^3\setminus\{0\}$ in the
third, and gives $w=0$. So the vectors $k$ span a dense subspace. The Segal fields
$(a(k)+a^*(k))/\sqrt2$ are affiliated with $\F(\Op)$, hence so are the
Weyl operators $W(k)$, and by the Weyl relations $W(k)\in\F(\Op)$ for
$k$ in the complex span. By strong continuity of $k\mapsto W(k)$ in the
Fock representation, $\F(\Op)$ contains all Weyl operators, which act
irreducibly.
\end{proof}

For $H_0$ the step $w=0$ is elementary. The function $z\mapsto
\braket{w,\rme^{\rmi zH_0/\hbar}g}$ is holomorphic in $\{\mathrm{Im}\,
z>0\}$ and continuous up to the real axis because $H_0\geq0$, so by the
Schwarz reflection principle it vanishes identically; at $z=\rmi s$ this
says that the heat evolution $\rme^{-sH_0/\hbar}w$, a real-analytic
function, vanishes on $B$, hence everywhere. For the Lipschitz
potentials of stars this argument is not available. Lemma~A7 of
Ref.~\onlinecite{pachon2026algebraicquantumkinematicssrselection}, a form
of Masuda's unique continuation theorem for the Schr\"odinger
evolution~\cite{Masuda1967}, replaces it: it passes to an elliptic
equation in one more variable, and of the potential it needs only that
it be real and locally bounded outside a closed null set with connected
complement, here the origin in the Schwarzschild case. For $H_0$ the
density of the vectors $k$ is a special case of Requardt's
theorems~\cite{Requardt1986}. He proves results of this kind for
Schr\"odinger operators with the unique continuation property and a
generalised eigenfunction expansion, which by Agmon's theorem exists for
short-range potentials, and, for multi-time nested states, for $n$-body
operators with translation-invariant pair potentials.
Theorem~\ref{thm:stars} assumes only $\Phi\to0$, and a star of mass $M$
has $\Phi=-GM/|x|$ outside it, as in Theorem~\ref{thm:schw}. These
potentials are not short-range and not translation invariant, so neither
result, as established there, applies; Lemma~A7 needs no eigenfunction
expansion.

In flat space, Proposition~\ref{prop:open-regions} and
Theorem~\ref{thm:net}(d) describe the same fact from the two sides of
the limit (Fig.~\ref{fig:types}(b)). At finite $c$ the algebra of a fixed cylinder contains the
time-zero algebra of the ball of radius $R+cT$, which exhausts $\R^3$ as
$c\to\infty$; in the limit every open region generates everything.
Neither side has a useful limit of space-time algebras: with fixed
temporal extent every algebra of the limit net is
$\mathcal{B}(\Gamma(L^2(\R^3)))$, and by
Proposition~\ref{prop:open-regions} the same holds for the limit nets of
Theorems~\ref{thm:stars} and~\ref{thm:schw}. Local structure of
the limit theory therefore lives at sharp times, as in the axioms of
Refs.~\onlinecite{pachon2026algebraicquantumkinematicssrselection,
pachon2026galileanreehschliederobstruction}. The same conclusion, that
with positive energy algebras of open space-time regions of Galilean
canonical fields never commute, is Corollary~4 of
Ref.~\onlinecite{pachon2026algebraicquantumkinematicssrselection}.

\subsection{Thermal states}\label{sec:thermal}

What is lost in the limit is the entanglement of the vacuum, not modular
structure as such. Thermal states survive, provided the chemical
potential of the conserved particle number of the free field is shifted
by the rest energy. The change of type at fixed time is not special to
the vacuum: in Minkowski space the relativistic thermal states below are
Hadamard, since their one-particle densities decay exponentially, so
their local algebras are type III$_1$ factors~\cite{Verch1997}, while the
fixed-time algebras of their limits are of type I
(Ref.~\onlinecite{pachon2026galileanreehschliederobstruction},
Theorem~6). What is special to the vacuum is that its limit is pure on
these algebras, hence not faithful.

\begin{proposition}[Thermal states]\label{prop:thermal}
Let $\beta>0$ and $\mu\in\R$, with $\mu<0$ in Minkowski space and
$\mu<\inf\sigma(H_S)$ for the stars of Theorem~\ref{thm:stars}. For
$c$ large let $\omega_c^{\beta,\mu}$ be the gauge-invariant quasi-free
state of the Klein--Gordon field with one-particle density
$(\rme^{\beta(P_c-\mu)}-1)^{-1}$, that is, the state on the Weyl algebra
over $\K_c$ that is KMS at inverse temperature $\beta$ for the dynamics
generated by $\rmd\Gamma(\hbar\sqrt{K_c})-(mc^2+\mu)\hat N=
\rmd\Gamma(P_c-\mu)$. Transported to $\mathfrak W$,
\[
\omega_c^{\beta,\mu}(W(h))\longrightarrow\omega^{\beta,\mu}(W(h)):=
\exp\Bigl(-\tfrac\hbar4\bigl\langle h,\coth\bigl(\tfrac\beta2(H_S-\mu)
\bigr)h\bigr\rangle\Bigr),
\]
the Araki--Woods state of the free Schr\"odinger gas at inverse
temperature $\beta$ and chemical potential $\mu$. The KMS states of
$\rmd\Gamma(\hbar\sqrt{K_c})$ at inverse temperature $\beta$ with
chemical potential $0$ converge instead to the Fock vacuum.
\end{proposition}

\begin{proof}
The state is $\omega_c^{\beta,\mu}(W(h))=\exp(-\frac\hbar4\braket{k_h,
\coth(\frac\beta2(P_c-\mu))k_h})$ with $k_h$ as in~\eqref{eq:kh}, and
$P_c-\mu\geq\delta>0$ for $c$ large: in Minkowski space $P_c\geq0>\mu$,
and on stars $P_c\geq F_c(\inf\sigma(S_c))\geq-C$ uniformly in $c$, while
norm-resolvent convergence (Theorem~\ref{thm:stars}) excludes spectrum of
$P_c$ from $[-C,\inf\sigma(H_S)-\delta]$ for $c$ large. The function
$\coth(\frac\beta2(p-\mu))$ is bounded and continuous on
$[\mu+\delta,\infty)$, so the convergence follows from $\mathcal{U}_c
k_h\to h$ (proof of Proposition~\ref{prop:net-static}) and the strong
convergence of bounded continuous functions of $P_c$. Without the shift
the relevant function is $\coth(\frac\beta2\hbar\sqrt{K_c})$, which tends
to $1$ uniformly on the spectrum because $\hbar\sqrt{K_c}\geq N_{\min}
mc^2$ by~\eqref{eq:K-lower}, with $N_{\min}:=\inf N_c$, equal to $1$ in
Minkowski space and at least $2^{-1/2}$ on stars for $c\geq c_1$.
\end{proof}

The shift of the chemical potential is the thermal counterpart of the
rest-energy subtraction~\eqref{eq:rest-energy-split}, and it is
available only because $\hat N$ is conserved; with interactions the
Bargmann charge of the limit is superselected, but at finite $c$ particle
number is not conserved. The modular group of $\omega_c^{\beta,\mu}$ is
its dynamics with time rescaled by
$-\beta\hbar$~\cite{HaagHugenholtzWinnink1967,Takesaki1970}, generated on the
one-particle level by $\beta(P_c-\mu)$. On $\mathfrak W$ it acts through
$\mathcal J$ as in the proof of Theorem~\ref{thm:net}(b), and the
argument given there and in Proposition~\ref{prop:net-static} shows that
it converges to the modular group of the Araki--Woods
state~\cite{ArakiWoods1963}; this concerns the global modular group, and
the local structure of the limiting thermal states is, in flat space,
that of Ref.~\onlinecite{pachon2026galileanreehschliederobstruction},
Theorem~6. On the
Schwarzschild exterior no such state exists: $\sigma(P_c)$ reaches
$-mc^2$, because of the modes near the horizon. The other thermal states
of the relativistic theory have no thermal limit. The Unruh state is the
vacuum itself, restricted to a wedge, and converges by
Theorem~\ref{thm:net}(a); what has no limit is its KMS property for the
boosts, whose temperature $T_{\rm U}=\hbar a/(2\pi ck_{\rm B})$ at proper
acceleration $a$ tends to $0$~\cite{Unruh1976}. The Hartle--Hawking
state~\cite{HartleHawking1976} has no limit at all: the Hawking
temperature $T_{\rm H}=\hbar c^3/(8\pi GMk_{\rm B})$ diverges, and the
occupation of rest-energy modes grows like $k_{\rm B}T_{\rm H}/mc^2
\propto c$. Both thermal properties come from bifurcate Killing horizons,
which have no counterpart in the limit.

\section{Galilean axioms on static backgrounds}\label{sec:axioms}

\subsection{Positivity modulo the central charge}

The Galilean nets of
Ref.~\onlinecite{pachon2026galileanreehschliederobstruction} are
generated by time-zero canonical fields in Weyl form and by the dynamics,
and satisfy its axioms (G1)--(G6): isotony, equal-time locality, Galilean
covariance with Bargmann central extension and Bargmann charge of the
fields, canonical fields, positive energy, and a vacuum that is the
unique vector invariant under space-time translations. On a static background
with a non-constant Newtonian potential the spatial translations and
the boosts are no longer symmetries, and the axioms change in two
respects. First,
\begin{enumerate}
\item[(G3$_{\rm c}$)] there is a strongly continuous unitary
representation of the Bargmann central extension of the isometry group
of the limiting Newton--Cartan structure, containing the time
translations and the central one-parameter group
$\rme^{\rmi\theta\hat M/\hbar}$, with $\hat M$ self-adjoint and central,
under which the net is covariant and the fields carry Bargmann
charge.
\end{enumerate}
Second, positivity and uniqueness refer to the time translations only.
Because the Bargmann charge is central and superselected, energy is
defined only up to multiples of $\hat M$, and the natural form of these
two axioms is the one used in
Ref.~\onlinecite{pachon2026algebraicquantumkinematicssrselection}
(hypothesis (H4) there):
\begin{enumerate}
\item[(G5$_\lambda$)] there is $\lambda\in\R$ such that $\hat H+
\lambda\hat M\geq0$;
\item[(G6$_{\rm c}^\lambda$)] $\ker(\hat H+\lambda\hat M)=\C\Omega$.
\end{enumerate}
Since $\hat M$ commutes with $\hat H$, $\ker(\hat H+\lambda\hat M)$ is
invariant under $\rme^{\rmi\theta\hat M/\hbar}$, so by
(G6$_{\rm c}^\lambda$) and Stone's theorem $\hat M\Omega=M_0\Omega$ for
some $M_0\in\R$ and $\hat H\Omega=-\lambda M_0\Omega$. The vacuum vector
is thus invariant under the time translations up to the phase
$\rme^{-\rmi t\lambda M_0/\hbar}$, and the vacuum state is invariant; this
is all one can ask when energy is defined modulo $\hat M$.
With boosts, (G5$_\lambda$) implies $\hat M\geq0$
(Ref.~\onlinecite{pachon2026algebraicquantumkinematicssrselection},
Lemma~4), and for $\lambda=0$ the two axioms reduce to (G5) and to a
time-translation version of (G6).

\begin{proposition}[Positivity modulo the central charge]
\label{prop:positivity}
\begin{enumerate}[label=(\roman*)]
\item The shift $\hat H\mapsto\hat H+\lambda\hat M$ preserves the
Weyl-form Bargmann relations of
Ref.~\onlinecite{pachon2026algebraicquantumkinematicssrselection}
(Sec.~IV there).
\item Let $H_S$ be bounded below and let $\hat H=\rmd\Gamma(H_S)$ and
$\hat M=m\hat N$ on the bosonic Fock space over $L^2(\Sigma)$. For every
$\lambda>-\inf\sigma(H_S)/m$, (G5$_\lambda$) and (G6$_{\rm c}^\lambda$)
hold, with the Fock vacuum.
\item At finite $c$, $\rmd\Gamma(\hbar\sqrt{K_c})=\rmd\Gamma(P_c)+c^2
\hat M\geq0$: relativistic positivity is (G5$_\lambda$) for the
pair $(\rmd\Gamma(P_c),\hat M)$ with $\lambda=c^2$.
\end{enumerate}
\end{proposition}

\begin{proof}
(i) This is Remark~4 of
Ref.~\onlinecite{pachon2026algebraicquantumkinematicssrselection}:
$\hat M$ commutes with the boosts, translations and rotations, and
$\rme^{\rmi t\lambda\hat M/\hbar}$ can be moved through each relation;
for instance $B(\bv)^{-1}\rme^{\rmi t(\hat H+\lambda\hat M)/\hbar}B(\bv)=
\rme^{\rmi|\bv|^2t\hat M/2\hbar}T(-\bv t)\,\rme^{\rmi t(\hat H+\lambda
\hat M)/\hbar}$. (ii) $\rmd\Gamma(H_S)+\lambda m\hat N=\rmd\Gamma(H_S+
\lambda m)$, and on the $n$-particle space this is at least
$n(\inf\sigma(H_S)+\lambda m)$, which is positive for $n\geq1$. Its
kernel is therefore spanned by the Fock vacuum. (iii) is
\eqref{eq:rest-energy-split} and $K_c\geq0$.
\end{proof}

For every gravitational potential with a Coulomb tail $H_S$ has
negative eigenvalues. Then $\rmd\Gamma(H_S)$ is unbounded below on the
bosonic Fock space, and the no-particle vector is not its unique
invariant vector if the eigenvalues are removed by an additive shift of
$H_S$. Both failures disappear once energy is measured modulo the
central charge. By (iii) they are artefacts of subtracting the rest
energy with $\lambda=0$: the relativistic theory is positive with
$\lambda=c^2$, and the limit is positive with any finite
$\lambda>-E_1/m$. With this form of the spectral axioms the failure of
positivity and uniqueness disappears for the limit nets of Theorems~\ref{thm:stars}
and~\ref{thm:schw}, for every static star and for Schwarzschild, where
$\Sigma=\R^3\setminus\{0\}$. Proposition~\ref{prop:positivity} checks the
spectral axioms and their compatibility with the Bargmann relations. The
other axioms hold as well. In the Weyl form of
Ref.~\onlinecite{pachon2026galileanreehschliederobstruction} they carry
no domain clauses, and for the limit Fock nets and every self-adjoint
$H_S$, (G1) and (G2) hold by construction and by the Weyl relations,
(G3$_{\rm c}$) holds with $\Gamma(\rme^{\rmi tH_S/\hbar})$, the second
quantised isometries of the background and $\rme^{\rmi\theta m\hat
N/\hbar}$, under which the Weyl system is covariant and carries Bargmann
charge $m$, and (G4) is the regularity of the Fock representation. In
the creation--annihilation form of
Ref.~\onlinecite{pachon2026algebraicquantumkinematicssrselection} the
vectors with finitely many particles serve as the common domain of the
fields, as in Proposition~5 there; no invariance under $\rmd\Gamma(H_S)$
is needed.

\subsection{The vacuum is not separating}

The obstruction of
Ref.~\onlinecite{pachon2026galileanreehschliederobstruction} extends to
these axioms. For the limit nets it is immediate: $\hat\psi(f)\Omega=0$,
while $\hat\psi(f)\neq0$ is affiliated with the algebra of every region
containing $\supp f$. The next proposition does not assume that the
vacuum is the no-particle vector.

\begin{proposition}[Non-separating vacuum]\label{prop:nonseparating}
Let $\F$ be a net generated, as in
Ref.~\onlinecite{pachon2026galileanreehschliederobstruction}, by a
regular Weyl system $W$ over $L^2(\Sigma)$ that carries Bargmann charge
$m_0$, $\rme^{\rmi\theta\hat M/\hbar}W(g)\rme^{-\rmi\theta\hat M/\hbar}=
W(\rme^{\rmi\theta m_0/\hbar}g)$, and satisfying (G3$_{\rm c}$),
(G5$_\lambda$) and (G6$_{\rm c}^\lambda$) above, realised on a
representation in which $\sigma(\hat M)\subseteq\{nm_0:n\in\Z_{\geq0}\}$.
Then $\Omega$ is not separating for $\F(\Op)$, for any non-empty open
$\Op$.
\end{proposition}

\begin{proof}
By (G6$_{\rm c}^\lambda$), $\hat M\Omega=M_0\Omega$ (see above), and
$M_0=Jm_0$ with $J\in\Z_{\geq0}$ because $M_0\in\sigma(\hat M)$. Without
boosts and spatial translations $M_0$ need not vanish. Since $\hat M\geq0$,
Theorem~1 of Ref.~\onlinecite{pachon2026galileanreehschliederobstruction}
gives a unitary $V$ with $VW(g)V^{-1}=W_{\rm F}(g)\otimes\one$ and
$V\hat MV^{-1}=m_0\hat N_{\rm F}\otimes\one+\one\otimes\hat M'$, with
$\hat M'\geq0$ on a multiplicity space; its parts (i) and (ii) use neither boosts nor spatial
translations. Let $B\times\{t\}\subset\Op$; after a time translation,
$t=0$. With $\Gamma_{\rm s}(L^2(\Sigma))=\Gamma_{\rm s}(L^2(B))\otimes
\Gamma_{\rm s}(L^2(\Sigma\setminus B))$, the slice algebra $\F(B,0)$
contains $Q:=V^{-1}(E\otimes\one\otimes\one)V$, where $E$ projects onto
the states of $\Gamma_{\rm s}(L^2(B))$ with more than $J$ particles.
Then $Q\neq0$, and $Q\Omega=0$, because $V\Omega$ lies in the eigenspace
of $m_0\hat N_{\rm F}\otimes\one+\one\otimes\hat M'$ for $Jm_0$, on
which $\hat N_{\rm F}\leq J$. The proof uses neither the boosts nor the
spatial translations, nor (G5$_\lambda$).
\end{proof}

The condition on $\sigma(\hat M)$ can be weakened to $\hat M$ bounded
below. Theorem~1 of
Ref.~\onlinecite{pachon2026galileanreehschliederobstruction} then gives
$\hat M'$ bounded below, and on the eigenspace of $M_0$ one has
$\hat N_{\rm F}\leq(M_0-\inf\sigma(\hat M'))/m_0$, so the same
construction applies.

What this proposition establishes concerns the vacuum.
Theorem~\ref{thm:net} gives the finer statement for the limit of the
Klein--Gordon field: the vacuum is not merely non-separating, it is a
pure product state on each local algebra, which is of type I, while
the thermal states stay faithful and their global modular groups
converge (Proposition~\ref{prop:thermal}); in flat space their
fixed-time algebras are of type I as well (Sec.~\ref{sec:thermal}).

\subsection{Two-point functions and regular states}\label{sec:NC-regular}

In a relativistic theory the Hadamard condition selects states by the
short-distance form of their two-point function. In the limit theory the
short-distance singularity does not depend on the state.

\begin{proposition}[Structure of two-point functions]\label{prop:structure}
For every state $\omega$ of the free Schr\"odinger field with
one-particle Hamiltonian $H_S$ for which the one-body density
$\rho_\omega(t',x';t,x):=\omega(\hat\psi(t',x')^*\hat\psi(t,x))$ exists
as a distribution,
\begin{equation}\label{eq:NR-two-point}
\omega\bigl(\hat\psi(t,x)\hat\psi(t',x')^*\bigr)=\hbar\,
\bigl(\rme^{-\rmi(t-t')H_S/\hbar}\bigr)(x,x')+\rho_\omega(t',x';t,x).
\end{equation}
\end{proposition}

\begin{proof}
The commutator $[\hat\psi(t,x),\hat\psi(t',x')^*]$ is the c-number
$\hbar(\rme^{-\rmi(t-t')H_S/\hbar})(x,x')$ by~\eqref{eq:limit-field}.
\end{proof}

The whole short-distance singularity is in the first term, which is
fixed by $H_S$. A short-time expansion of the two-point function with
smooth coefficients, the natural Galilean analogue of the Hadamard
condition, holds either for every state with smooth $\rho_\omega$ or for
none; it is a property of $H_S$. Such an expansion is understood
pointwise, uniformly on compact subsets of $\Sigma\times\Sigma$: the
kernel of $\rme^{-\rmi tH_S/\hbar}$ is the free kernel times
$\sum_{j<J}a_j(x,y)t^j$, with smooth $a_j$, up to $O(t^{J-3/2})$. Smeared
with test functions the question is empty, because $\braket{f,
\rme^{-\rmi tH_S/\hbar}g}$ has a Taylor expansion in $t$ for all $f,g\in
C_c^\infty(\Sigma)$ in the domain of every power of $H_S$. On compact
spatial sections the image terms $\rme^{\rmi m|x-y+nL|^2/2\hbar t}$ of
the free propagator spoil the pointwise expansion (they also produce the
Talbot revivals at finite times~\cite{BerryKlein1996}); this too is a
property of the propagator. The condition on states that remains is the
following.

\begin{definition}[Newton--Cartan regular states]
\label{def:NC-regular}
A state $\omega$ of the free Schr\"odinger field is \emph{Newton--Cartan
regular} if its one-body density $\rho_\omega$ is a smooth function on
$(\R\times\Sigma)^2$.
\end{definition}

The Fock vacuum ($\rho_\omega=0$) is Newton--Cartan regular. For
$H_S=H_0$ so are the Araki--Woods states of
Proposition~\ref{prop:thermal}, whose densities have Schwartz Fourier
transforms, and states with finitely many particles in Schwartz wave
functions. For a potential that is smooth on $\Sigma$, such as the
Coulomb one, the same holds for wave functions in the domain of every
power of $H_S$, by elliptic regularity; for the Lipschitz potentials of
stars, smooth initial data need not give a smooth $\rho_\omega$. Regularity plays the
role that $\omega-\omega_{\rm Had}\in C^\infty$ plays relativistically.

\begin{remark}[The Coulomb propagator]\label{rem:coulomb}
For $H_S=-\frac12\Delta-Z/r$ (units $\hbar=m=1$) the propagator itself
admits no pointwise short-time expansion of the free form
$(2\pi\rmi t)^{-3/2}\rme^{\rmi|x-y|^2/2t}\sum_ja_jt^j$ with smooth
coefficients beyond $j=1$, at first order in $Z$ and at the level of the
formal computation below. At that order the kernel of
$(H_S-k^2/2-\rmi0)^{-1}$ contains
$G_1(x,y;k)=(ZR_{xy}/8\pi^2)\int_1^\infty\rme^{\rmi kR_{xy}\xi}
\mathcal{A}(\xi)\,\rmd\xi$, with $R_{xy}=|x-y|$, $\xi$ the prolate spheroidal coordinate with
foci $x,y$, and $\mathcal{A}(\xi)$ the integral of $1/|z|$ over the
ellipsoid $|x-z|+|z-y|=R_{xy}\xi$ in the measure $\rmd\eta\,\rmd\phi$. The
function $\mathcal{A}$ is smooth except at $\xi_0=(|x|+|y|)/R_{xy}$,
where the ellipsoid passes through the origin and $\mathcal{A}'$ jumps
by $-8\pi R_{xy}/(2|x||y|+2x\cdot y)$. The jump produces in $G_1$ the
term $Z\,\rme^{\rmi k(|x|+|y|)}/[\pi(2|x||y|+2x\cdot y)\,k^2]$ as
$k\to\infty$; these statements about the resolvent kernel are proved.
A formal stationary-phase evaluation of the $k$-integral, in which the
region of small $k$ and the uniformity of the remainder are not
controlled, turns this term into a term of the propagator kernel with
phase $(|x|+|y|)^2/2t$ and modulus
$Z(2\pi t)^{1/2}/[2\pi^2(|x|+|y|)(2|x||y|+2x\cdot y)]$, which is of
relative size $O(t^2)$ with respect to the free kernel. Its phase
differs from $|x-y|^2/2t$ unless the origin lies on the segment joining
$x$ and $y$. A short-time expansion with smooth coefficients therefore
fails at order $t^2$, even near the diagonal in compact subsets of
$\R^3\setminus\{0\}$, at first order in $Z$. A rigorous version of this
step would start from the Duhamel term of first order in $Z$, where the
stationary point in time and the conical singularity of $1/|z|$ at the
origin give the same order. The jump is confirmed
numerically (Appendix~\ref{app:numerics}). A proof at all orders would
start from the closed form of the Coulomb Green's
function~\cite{Hostler1964}; for the analogous diffraction by a point
interaction in one dimension see Ref.~\onlinecite{GaveauSchulman1986}.
\end{remark}

\section{Discussion}\label{sec:discussion}

The limit $c\to\infty$ of the free Klein--Gordon field on a static
spacetime has been taken at three levels, and each level keeps and loses
different things.

At the one-particle level the rest-energy-subtracted Hamiltonian is a
function of a Schr\"odinger-type operator whose potential is the
Newtonian one (Lemma~\ref{lem:reduction}). Without horizon the limit is
uniform, with rate $c^{-2}$ and a computable first correction
(Theorem~\ref{thm:stars}, Proposition~\ref{prop:PN}). With a horizon
the limit is only strong. The spectrum changes type, the limit picks
one self-adjoint realisation of the gravitational hydrogen atom, and the
bound states are metastable at every finite $c$, with the widths of
gravitational atoms (Theorem~\ref{thm:schw},
Corollary~\ref{cor:metastability}, Theorem~\ref{thm:resonances}). At
finite $c$ the horizon shows up through the absorption that makes these
states decay; in the limit it leaves only the choice of the Friedrichs
realisation. On the Reissner--Nordstr\"om exterior, extremal or not,
the strong convergence, the selection of the Friedrichs realisation and
the metastability hold as well (Corollary~\ref{cor:RN}); its resonances
are open (problem~3 below).

At the level of states and nets the relativistic theory and its limit
share the time-zero algebras, and the vacua and dynamics converge on
them (Theorem~\ref{thm:net}, Proposition~\ref{prop:net-static}). The
local algebras nevertheless change type. The entanglement of the vacuum
between a region and its complement, carried by the type III$_1$
property, is absent in the limit, where the vacuum is a product state.
This is the precise content of the statement that vacuum modular
structure collapses in the Newton--Cartan limit. It concerns the vacuum
only: thermal states survive once the chemical potential includes the
rest energy (Proposition~\ref{prop:thermal}), and so do their global
modular groups. The change of type at fixed time is not special to the
vacuum: it occurs for the thermal states of
Proposition~\ref{prop:thermal} in flat space as well
(Sec.~\ref{sec:thermal}); what is special to the vacuum is that its
limit is a pure product state. In the vacuum representation of the limit the algebras of
open space-time regions are all $\mathcal{B}(\Hil)$, in flat space and on
the static backgrounds of
Theorems~\ref{thm:stars} and~\ref{thm:schw}
(Proposition~\ref{prop:open-regions}); in flat space this is the shadow
of the timelike tube theorem. The local structure of the limit therefore
lives at sharp times.

These results fit the axioms of the companion papers. Positivity has to
be imposed modulo the central charge, as in
Ref.~\onlinecite{pachon2026algebraicquantumkinematicssrselection};
with it the spectral axioms hold for the limit nets of every static
star and of Schwarzschild (Proposition~\ref{prop:positivity}), and their
vacuum is not separating (Proposition~\ref{prop:nonseparating}), as in
Ref.~\onlinecite{pachon2026galileanreehschliederobstruction}. The
condition on states that replaces the Hadamard condition is the
smoothness of the one-body density (Definition~\ref{def:NC-regular}).

Several problems remain open.
\begin{enumerate}
\item \emph{Decay law.} Theorem~\ref{thm:resonances} locates the
resonances. An exponential decay law of the kind proved in
Ref.~\onlinecite{Hunziker1990}, with survival amplitude
$\rme^{-\rmi tz_{nl}(c)/\hbar}$ for states of definite angular momentum up
to errors that vanish as $c\to\infty$ uniformly in $t\geq0$, needs bounds
on the continued resolvent away from its poles, uniform in $c$. Complex
scaling~\cite{Simon1973} in the tortoise coordinate towards the horizon, where $r(r_*)$ is
analytic for $\mathrm{Re}\,r_*<0$, is a natural route. Corollary~\ref{cor:metastability}
gives only a lifetime of order at least $c^{1/2}$.
\item \emph{Rotating black holes.} For Kerr at fixed dimensionless spin
the spacetime is stationary but not static, the Killing field
$\partial_t$ is spacelike in the ergoregion, so there is no ground state
of the kind used here, and superradiant instabilities
appear~\cite{BritoCardosoPani2020}. The framework of
Sec.~\ref{sec:setting} does not apply as it stands.
\item \emph{Charged backgrounds and charged fields.}
Corollary~\ref{cor:RN} treats the Reissner--Nordstr\"om exterior,
including the extremal case, for a neutral field. Its resonances would
need the near zone of Appendix~\ref{app:resonances} with the two
horizons, and in the extremal case the ingoing solution is no longer
analytic in a strip. A charged field in an
electrostatic potential enters $K_c$ through the time derivative, and
Lemma~\ref{lem:reduction} has no analogue.
\item \emph{Interacting fields.} In two dimensions the non-relativistic
limit of two-particle phenomena of the $P(\varphi)_2$ models was studied
by Dimock~\cite{Dimock1977}. A limit at the level of time-zero nets for
interacting fields is open.
\item \emph{Quantitative loss of entanglement.} On ultrastatic
spacetimes the relative entanglement entropy of the free massive vacuum
between regions at distance $d$ is bounded by a constant times
$\rme^{-mcd/2\hbar}$~\cite{HollandsSanders2018}. Whether the constant can
be chosen uniformly in $c$, or grows at most polynomially in $mc/\hbar$,
which would make the bound vanish in the limit, is open. In the proof of
Ref.~\onlinecite{HollandsSanders2018} the prefactor built from the
spectral projections tends to a constant as $c\to\infty$, and what remains
is to follow the dependence on $mc/\hbar$ of the trace norms.
\end{enumerate}

\appendix
\section{Proof of the resolvent bound of Lemma~\ref{lem:reduction}}
\label{app:reduction}

Write $\mathcal{E}:=mc^2$, $y:=|\mathrm{Im}\,z|>0$, and parametrise $s\geq-\mathcal{E}/2$ by
$F=F_c(s)\in[-\mathcal{E},\infty)$, so that $s=F+F^2/2\mathcal{E}$ and $s-F=F^2/2\mathcal{E}$. The
quantity to bound is
\[
d=\frac{F^2/2\mathcal{E}}{(F-z)(s-z)},
\]
under the assumption $\mathcal{E}\geq8|z|$. We distinguish four ranges of $F$.
\begin{enumerate}[label=(\roman*)]
\item $|F|\leq2|z|$. Since $F$ and $s$ are real, $|F-z|\geq y$ and
$|s-z|\geq y$, so $|d|\leq(4|z|^2/2\mathcal{E})/y^2=2|z|^2/(\mathcal{E}y^2)$.
\item $2|z|\leq|F|\leq \mathcal{E}/2$. Then $|F-z|\geq|F|/2$. Since
$|F|/2\mathcal{E}\leq\frac14$, $s=F(1+F/2\mathcal{E})$ satisfies $|s|\geq\frac34|F|$, hence
$|s-z|\geq\frac34|F|-\frac12|F|=\frac14|F|$ and $|d|\leq(F^2/2\mathcal{E})\cdot
8/F^2=4/\mathcal{E}$.
\item $F\geq \mathcal{E}/2$. Then $|z|\leq \mathcal{E}/8\leq F/4\leq s/4$, so $|F-z|\geq
\frac34F$, $|s-z|\geq\frac34s\geq\frac34F$, and $|d|\leq\frac{16}{9}
\cdot\frac{F^2/2\mathcal{E}}{F\,s}\leq\frac{8}{9\mathcal{E}}$.
\item $-\mathcal{E}\leq F\leq-\mathcal{E}/2$. Then $s\in[-\mathcal{E}/2,-3\mathcal{E}/8]$, because $s$ is
increasing in $F$ on $[-\mathcal{E},\infty)$. Hence $|F-z|\geq \mathcal{E}/2-\mathcal{E}/8=3\mathcal{E}/8$,
$|s-z|\geq3\mathcal{E}/8-\mathcal{E}/8=\mathcal{E}/4$, and $|d|\leq(\mathcal{E}/2)\cdot\frac{32}{3\mathcal{E}^2}=
\frac{16}{3\mathcal{E}}$.
\end{enumerate}
Since $2|z|\leq \mathcal{E}/4$, the four ranges cover $[-\mathcal{E},\infty)$, and
$\sup|d|\leq(16/3+2|z|^2/y^2)/\mathcal{E}$, which is~\eqref{eq:reduction-bound}.

\section{The Schwarzschild exterior}\label{app:schw}

\subsection{Partial waves and strong convergence}\label{app:strong}

Let $r_*=r+r_{\rm s}\ln(r/r_{\rm s}-1)$ be the tortoise coordinate, an
increasing bijection of $(r_{\rm s},\infty)$ onto $\R$ with
$\rmd r_*/\rmd r=N_c^{-2}$, and write $\partial_*=\partial_{r_*}=N_c^2
\partial_r$. In these coordinates $\rmd\mu_c=r^2\,\rmd r_*\,\rmd\Omega$
and
\begin{equation}\label{eq:Ac-tortoise}
A_c=-\frac{1}{r^2}\,\partial_*\bigl(r^2\partial_*\,\cdot\,\bigr)-
\frac{N_c^2}{r^2}\,\Delta_{S^2}.
\end{equation}

\begin{alemma}[Partial waves]\label{lem:partial-waves}
The map $u=r^{-1}w(r_*)Y_{lm}$ identifies $\K_c$ unitarily with
$\bigoplus_{l,m}L^2(\R,\rmd r_*)$ and $S_c$ with $\bigoplus_{l,m}
\bigl[-\frac{\hbar^2}{2m}\partial_*^2+W_l\bigr]$, where
\[
W_l=\frac{\hbar^2}{2m}N_c^2\Bigl(\frac{l(l+1)}{r^2}+
\frac{r_{\rm s}}{r^3}\Bigr)-\frac{GMm}{r}.
\]
Consequently $K_c$ is essentially self-adjoint on
$C_c^\infty(\Sigma_c)$, $\sigma(S_c)=[-mc^2/2,\infty)$, and $S_c$ has no
eigenvalues.
\end{alemma}

\begin{proof}
With $u=r^{-1}wY_{lm}$, $\partial_*(r^2\partial_*(w/r))=r\,w''-
(\partial_*N_c^2)w$ and $\partial_*N_c^2=N_c^2r_{\rm s}/r^2$, which gives
the radial part of $W_l$; the angular part is
$N_c^2l(l+1)/r^2$. Each $W_l$ is smooth and bounded on $\R$. As
$r_*\to-\infty$, $N_c^2=(r_{\rm s}/r)\rme^{(r_*-r)/r_{\rm s}}\leq
\rme^{r_*/r_{\rm s}-1}$ and $W_l+mc^2/2=N_c^2[\frac{\hbar^2}{2m}(\frac{l(l+1)}
{r^2}+\frac{r_{\rm s}}{r^3})+\frac{mc^2}2]=O(\rme^{r_*/r_{\rm s}})$; as
$r_*\to+\infty$, $W_l=-GMm/r+O(r^{-2})\to0$. A bounded potential on the
line gives an operator that is essentially self-adjoint on
$C_c^\infty(\R)$; the functions $r^{-1}w(r_*)Y_{lm}$ with $w\in
C_c^\infty(\R)$ lie in $C_c^\infty(\Sigma_c)$, and the direct sum of the
closures is self-adjoint. This proves essential self-adjointness of
$S_c$, hence of $K_c=(2mc^2S_c+m^2c^4)/\hbar^2$.

By Lemma~\ref{lem:reduction}, $\sigma(S_c)\subseteq[-mc^2/2,\infty)$.
Conversely, for $E=-mc^2/2+\hbar^2k^2/2m$ the $s$-wave functions
$w_n(r_*)=n^{-1/2}\chi((r_*+n^2)/n)\,\rme^{\rmi kr_*}$, with $\chi\in
C_c^\infty((-1,1))$ normalised, form a Weyl sequence, because
$W_0+mc^2/2\to0$ on their supports. Finally let $E\geq-mc^2/2$ and let
$w$ solve $-\frac{\hbar^2}{2m}w''+(W_l-E)w=0$. The function $q=
\frac{2m}{\hbar^2}(W_l+mc^2/2)$ satisfies $\int_{-\infty}^0(1+|r_*|)|q|\,
\rmd r_*<\infty$. By asymptotic integration~\cite{Weidmann1987}, as
$r_*\to-\infty$ every non-trivial solution behaves like
$a\,\rme^{\rmi kr_*}+b\,\rme^{-\rmi kr_*}$ with $(a,b)\neq(0,0)$ if $k>0$,
and like $a+br_*$ with $(a,b)\neq(0,0)$ if $k=0$. In neither case is it
square integrable at $-\infty$, so $E$ is not an eigenvalue.
\end{proof}

\begin{alemma}[Graph-limit criterion]\label{lem:graph-limit}
Let $T_c$ be self-adjoint on $\K_c$, $\mathcal{U}_c:\K_c\to\Hil$
isometries, and $T$ self-adjoint on $\Hil$ with core $\mathcal C$.
Suppose that for every $\varphi\in\mathcal C$ there are $u_c\in D(T_c)$
with $\mathcal{U}_cu_c\to\varphi$ and $\mathcal{U}_cT_cu_c\to T\varphi$.
Then $\mathcal{U}_c(T_c-z)^{-1}\mathcal{U}_c^*\to(T-z)^{-1}$ strongly for
every $z\in\C\setminus\R$.
\end{alemma}

\begin{proof}
The operators on the left are bounded by $1/|\mathrm{Im}\,z|$, and
$(T-z)\mathcal C$ is dense, so it suffices to take $\psi=(T-z)\varphi$
with $\varphi\in\mathcal C$. Since $\mathcal{U}_c^*\mathcal{U}_c=\one$,
\[
\mathcal{U}_c(T_c-z)^{-1}\mathcal{U}_c^*\psi-\varphi=
\mathcal{U}_c(T_c-z)^{-1}\mathcal{U}_c^*\bigl[\psi-\mathcal{U}_c(T_c-z)u_c
\bigr]+\bigl(\mathcal{U}_cu_c-\varphi\bigr),
\]
and both terms tend to $0$.
\end{proof}

Criteria of this kind for operators on varying spaces are standard; see
Refs.~\onlinecite{StolzWeidmann1993,KuwaeShioya2003,Post2012}, and, for
the removal of small holes, Ref.~\onlinecite{RauchTaylor1975}.

\begin{alemma}[Approximants]\label{lem:approximants}
Let $\theta\in C^\infty(\R)$ with $0\leq\theta\leq1$, $\theta=0$ on
$(-\infty,-2]$ and $\theta=1$ on $[-1,\infty)$, let $L=(\ell r_{\rm s})^{1/2}$ for a fixed
length $\ell$, and for $\varphi\in C_c^\infty(\R^3)$ put $u_c:=\theta(r_*/L)\,
\varphi|_{\Sigma_c}$. Then $u_c\in C_c^\infty(\Sigma_c)$,
$\|\mathcal{U}_cu_c-\varphi\|=O(r_{\rm s})$ and $\|\mathcal{U}_cS_cu_c-
H_S\varphi\|=O(r_{\rm s}^{1/4})$, where $\mathcal{U}_cS_cu_c$ is extended
by zero to $\{r\leq r_{\rm s}\}$.
\end{alemma}

\begin{proof}
Measure lengths in units of $\ell$, so that $L=r_{\rm s}^{1/2}$. Let $C$
denote constants depending on $\varphi$, $\theta$ and $\ell$ only, and
let $r_{\rm s}\leq\frac14$. Two facts about $r_*$ are used. First,
$r_*(2r_{\rm s})=2r_{\rm s}>-L$, so $\theta(r_*/L)=1$ for $r\geq
2r_{\rm s}$. Second, where $r_*\leq-L$ one has $\ln((r-r_{\rm s})/r_{\rm s})
\leq-L/r_{\rm s}-1$, so on the support of $\theta'(r_*/L)$
\begin{equation}\label{eq:tiny-lapse}
N_c^2\leq\rme^{-1-r_{\rm s}^{-1/2}} .
\end{equation}
The region $\{-2L\leq r_*\leq2r_{\rm s}\}$, which contains the part of
$\supp u_c$ with $r<2r_{\rm s}$, has tortoise length at most $3L$.

\emph{The vectors.} $\mathcal{U}_cu_c=N_c^{-1}\theta\varphi$ on
$\Sigma_c$. On $\{r\leq r_{\rm s}\}$, $\|\varphi\|^2\leq Cr_{\rm s}^3$. On
$\{r\geq2r_{\rm s}\}$, $|N_c^{-1}-1|\leq\sqrt2\,r_{\rm s}/r$, contributing
$Cr_{\rm s}^2$ to the square of the norm. On $\{r_{\rm s}<r<2r_{\rm s}\}$
we use $\rmd^3x=N_c^2\,\rmd\mu_c$: $\int N_c^{-2}\theta^2|\varphi|^2
\rmd^3x=\int\theta^2|\varphi|^2r^2\,\rmd r_*\rmd\Omega\leq Cr_{\rm s}^2L$,
and $\int|\varphi|^2\rmd^3x\leq Cr_{\rm s}^3$.

\emph{The operator.} Write $S_cu_c=\theta S_c\varphi+\frac{\hbar^2}{2m}
[A_c,\theta]\varphi$. For smooth $\varphi$, with $\Delta_r=\partial_r^2+
\frac2r\partial_r$ and $\Delta=\Delta_r+r^{-2}\Delta_{S^2}$,
\[
N_c^{-1}S_c\varphi-H_S\varphi=-\frac{\hbar^2}{2m}\Bigl[(N_c-1)\Delta
\varphi-N_c\frac{r_{\rm s}}{r}\Delta_r\varphi+N_c\frac{r_{\rm s}}{r^2}
\partial_r\varphi\Bigr]-\frac{GMm}{r}\bigl(N_c^{-1}-1\bigr)\varphi .
\]
Here $|\Delta\varphi|\leq C$, $|\Delta_r\varphi|+|\partial_r\varphi|
\leq C(1+r^{-1})$, $0\leq1-N_c\leq r_{\rm s}/r$ and $N_c\leq1$. The
first three terms have $L^2$ norm $O(r_{\rm s}^{1/2})$ on $\Sigma_c$. The
last is $O(r_{\rm s}^{1/2})$ on $\{r\geq2r_{\rm s}\}$, and on
$\{r_{\rm s}<r<2r_{\rm s}\}$ its square norm, weighted by $\theta^2$, is
at most $C\int\theta^2N_c^{-2}\,\rmd^3x/r^2\leq C\int\theta^2\,\rmd r_*
\leq CL$. Moreover $(\theta-1)H_S\varphi$ and $H_S\varphi$ restricted to
$\{r\leq r_{\rm s}\}$ are supported in $\{r<2r_{\rm s}\}$, where
$\|H_S\varphi\|^2\leq Cr_{\rm s}$. For the commutator,
by~\eqref{eq:Ac-tortoise},
\[
[A_c,\theta]\varphi=-\frac{1}{L^2}\theta''\varphi-\frac2L\theta'\,
\partial_*\varphi-\frac{2N_c^2}{rL}\theta'\varphi,
\]
with $\theta$ and its derivatives evaluated at $r_*/L$. It is supported
in $\{-2L\leq r_*\leq-L\}$, where $r<2r_{\rm s}$ and
\eqref{eq:tiny-lapse} holds, and $\partial_*\varphi=N_c^2\partial_r
\varphi$. Its norm in $\K_c$, which is the norm of $\mathcal{U}_c$
applied to it, satisfies
\[
\bigl\|[A_c,\theta]\varphi\bigr\|^2\leq C\,r_{\rm s}^2\int_{-2L}^{-L}
\Bigl(\frac{1}{L^2}+\frac{\rme^{-r_{\rm s}^{-1/2}}}{r_{\rm s}L}\Bigr)^2
\rmd r_*\leq C\frac{r_{\rm s}^2}{L^3}=C\,r_{\rm s}^{1/2}.
\]
Collecting terms, $\|\mathcal{U}_cS_cu_c-H_S\varphi\|=O(r_{\rm s}^{1/4})$.
\end{proof}

In units of $\ell$, the cut-off length $L$ must satisfy $r_{\rm s}^{2/3}
\ll L\ll1$: the
commutator costs $r_{\rm s}^2/L^3$ and the Newtonian potential near the
horizon costs $L$. The numerically observed rate $c^{-3/2}$ for
$s$-waves (Table~\ref{tab:schw}) is better than the rate $c^{-1/2}$ that
this construction gives.

\subsection{Resonances}\label{app:resonances}

This subsection proves Theorem~\ref{thm:resonances}. We use units
$\hbar=m=GM=1$, in which the Bohr radius is $1$, $mc^2=c^2$,
$\alpha=1/c$, $r_{\rm s}=2/c^2$, $E_n=-1/2n^2$ and $H_S=-\frac12\Delta-1/r$.
Fix $l$. For $s\in\C\setminus[0,\infty)$ let
$\kappa:=(-2s)^{1/2}$ with $\mathrm{Re}\,\kappa>0$; the level $E_n$
corresponds to $\kappa_n:=1/n$. Wronskians are $W(f,g)=fg'-f'g$, in the
variable indicated by the context.

\emph{Reduction.} By~\eqref{eq:P-S-resolvent}, with $s=z+z^2/2c^2$, it
suffices to continue $\chi\,\mathcal{U}_c(S_c-s)^{-1}\Lambda_l\,
\mathcal{U}_c^*\chi$: on $\Omega_c$ one has $1+z/c^2\neq0$, the last term
of~\eqref{eq:P-S-resolvent} is analytic, and $s\notin[0,\infty)$,
because $\mathrm{Im}\,s=(1+\mathrm{Re}\,z/c^2)\,\mathrm{Im}\,z$ and
$z\mapsto s$ maps $(-c^2,0)$ onto $(-c^2/2,0)$. The last term is also
Hilbert--Schmidt between the cut-offs: on $\Omega_c$,
$(P_c+z+2c^2)^{-1}=G_z(S_c)$ with $|G_z(s)|\leq C_z(s+c^2)^{-1/2}$, and
$\chi\,\mathcal{U}_c(S_c+c^2)^{-1/2}\Lambda_l$ is Hilbert--Schmidt: its
product with its adjoint, $\chi\,\mathcal{U}_c(S_c+c^2)^{-1}\Lambda_l
\mathcal{U}_c^*\chi$, is positive and has the continuous kernel of
compact support~\eqref{eq:cutoff-kernel} below, with $G_s$ replaced by
the kernel of $(h_l+c^2)^{-1}$, so it is trace class. By Lemma~\ref{lem:partial-waves}, on the range of $\Lambda_l$ the operator
$S_c$ is $2l+1$ copies of $h_l:=-\frac12\partial_*^2+W_l$ on
$L^2(\R,\rmd r_*)$, and the operator above has the kernel
\begin{equation}\label{eq:cutoff-kernel}
\frac{\chi(x)\chi(y)}{N_c(x)N_c(y)|x||y|}\sum_{m=-l}^lY_{lm}(\hat x)
\overline{Y_{lm}(\hat y)}\;G_s\bigl(r_*(|x|),r_*(|y|)\bigr),
\end{equation}
where $G_s$ is the kernel of $(h_l-s)^{-1}$. On $\supp\chi$ the tortoise
coordinate ranges over a compact interval $I_\chi$, on which $N_c^{-1}$
and $1/r$ are bounded, so everything reduces to $G_s$ on
$I_\chi\times I_\chi$.

\begin{alemma}[Continuation]\label{lem:continuation}
Let $k:=(z+c^2)/c$, so that $k^2=c^2+2s$.
\begin{enumerate}[label=(\roman*)]
\item For $\mathrm{Im}\,k>-1/2r_{\rm s}$ the equation
$-\frac12w''+(W_l-s)w=0$ has a solution $f(\cdot,k)$, analytic in $k$,
with $\rme^{\rmi kr_*}f(r_*,k)\to1$ and $\rme^{\rmi kr_*}f'(r_*,k)\to
-\rmi k$ as $r_*\to-\infty$. For $\mathrm{Im}\,k>0$ it spans the
solutions that are square integrable near $-\infty$.
\item Every $s_0\in\C\setminus[0,\infty)$ has a neighbourhood on which the
solutions that are square integrable near $+\infty$ are the multiples of
one solution $\psi(\cdot,s)$, analytic in $s$.
\end{enumerate}
For $z\in\Omega_c$ with $\mathrm{Im}\,z>0$,
$G_s(r_*,r_*')=-2f(r_<)\psi(r_>)/W(f,\psi)$, with
$r_<:=\min(r_*,r_*')$ and $r_>:=\max(r_*,r_*')$. The right-hand side
continues meromorphically in $z$ to $\Omega_c$, uniformly for
$r_*,r_*'\in I_\chi$, and its poles are the zeros of
$z\mapsto W\bigl(f(\cdot,k(z)),\psi(\cdot,s(z))\bigr)$, with the same
orders.
\end{alemma}

\begin{proof}
(i) Put $q:=2W_l+c^2$. By the proof of Lemma~\ref{lem:partial-waves},
$0\leq q\leq C\rme^{r_*/r_{\rm s}}$ on $(-\infty,0]$. The function
$\tilde f:=\rme^{\rmi kr_*}f$ is to solve
\[
\tilde f(r_*)=1+\int_{-\infty}^{r_*}\frac{\rme^{2\rmi k(r_*-t)}-1}{2\rmi k}\,
q(t)\,\tilde f(t)\,\rmd t,
\]
whose kernel is bounded by $(r_*-t)\rme^{2b(r_*-t)}q(t)$, with
$b:=\max(0,-\mathrm{Im}\,k)$. For $b<1/2r_{\rm s}$ and every $a\in\R$,
$\int_{-\infty}^a(a-t)\rme^{2b(a-t)}q(t)\,\rmd t<\infty$, so the Volterra
series converges uniformly on $(-\infty,a]$ and locally uniformly in $k$,
with $\tilde f\to1$ and $\tilde f'\to0$ at $-\infty$. This gives $f$, analytic in $k$~\cite{Weidmann1987}.
For $\mathrm{Im}\,k>0$, $f$ decays at $-\infty$, and every solution
independent of it grows like $\rme^{\rmi kr_*}$.

(ii) If $\mathrm{Re}\,s_0<0$, choose $a$ with $W_l\geq\frac12\mathrm{Re}\,
s_0$ on $[a,\infty)$, which is possible because $W_l\to0$; otherwise any
$a$ will do. The Dirichlet realisation $h_a$ of $h_l$ on $[a,\infty)$, which is in
the limit point case at $\infty$~\cite{Weidmann1987}, is self-adjoint
and bounded below by $\inf_{[a,\infty)}W_l$, so $s_0$ and a
neighbourhood of it lie in its resolvent set. Its kernel is
$-2\varphi_a(r_<)\psi(r_>)/W(\varphi_a,\psi)$, where $\varphi_a$ solves
the equation with $\varphi_a(a)=0$. Fix $y_0>a$ with
$\varphi_a(y_0,s_0)\neq0$. For $s$ near $s_0$, the function $r_*\mapsto
(h_a-s)^{-1}(r_*,y_0)$ on $(y_0,\infty)$, continued to $\R$ as a solution,
is a non-zero multiple of $\psi$ and is analytic in $s$.

For $z\in\Omega_c$ with $\mathrm{Im}\,z>0$ one has $\mathrm{Im}\,s>0$ and $\mathrm{Im}\,k=
\mathrm{Im}\,z/c>0$, so $G_s$ has the form stated. The right-hand side
does not depend on the normalisations of $f$ and $\psi$. On $\Omega_c$,
where $\mathrm{Im}\,k>-1/2r_{\rm s}$ and $s\notin[0,\infty)$, the local
choices of (ii) make it meromorphic. Its numerator $f(r_<)\psi(r_>)$
vanishes identically for no $z$, so its poles are the zeros of $W$, with
the same orders.
\end{proof}

\begin{proof}[Proof of Theorem~\ref{thm:resonances}(a)]
Combine~\eqref{eq:P-S-resolvent}, \eqref{eq:cutoff-kernel} and
Lemma~\ref{lem:continuation}. At a pole the leading Laurent coefficient
of~\eqref{eq:cutoff-kernel} is a non-zero multiple of
$\chi(x)\chi(y)f(r_<)f(r_>)\sum_mY_{lm}(\hat x)\overline{Y_{lm}(\hat y)}
/(N_c(x)N_c(y)|x||y|)$, because $\psi$ is then a multiple of $f$. Since
$\chi\neq0$ and $f$ has isolated zeros, this is not the zero operator, so
the poles are those of $G_s$ and do not depend on $\chi$.
\end{proof}

\emph{The radial equation.} In the variable $r\in(r_{\rm s},\infty)$ put
$R:=w/r$. Since $\partial_*=N_c^2\partial_r$ and $r^2N_c^2=r(r-r_{\rm
s})$, the equation $-\frac12w''+(W_l-s)w=0$ is
\begin{equation}\label{eq:radial-R}
\bigl(r(r-r_{\rm s})R'\bigr)'+\Bigl[\frac{2r^2-\kappa^2r^3}{r-r_{\rm s}}-l(l+1)\Bigr]R=0 ,
\end{equation}
which for $r_{\rm s}=0$ is the radial hydrogen equation. With
$x:=r/r_{\rm s}$ and $\varpi:=kr_{\rm s}$, so that
$\varpi^2=2r_{\rm s}-\kappa^2r_{\rm s}^2$, it reads
\[
\partial_x\bigl(x(x-1)\partial_xR\bigr)+\Bigl[\frac{\varpi^2}{x-1}+p(x)
-l(l+1)\Bigr]R=0,\qquad p(x):=2r_{\rm s}(x+1)-\kappa^2r_{\rm s}^2(x^2+x+1).
\]
Substituting $R=(x-1)^{\tau}g$ with $\tau:=-\rmi\varpi$ gives
\begin{equation}\label{eq:g}
x(x-1)g''+\bigl[(2+2\tau)x-1\bigr]g'+\bigl[\tau-\varpi^2+p(x)-l(l+1)\bigr]
g=0,
\end{equation}
which has a regular singular point at $x=1$ with exponents $0$ and
$-2\tau$. Let $g$ be its solution that is analytic at $x=1$ with
$g(1)=1$; it exists and is analytic in $k$ for
$\mathrm{Im}\,k>-1/2r_{\rm s}$, since then $j+2\tau\neq0$ for
$j=1,2,\dots$ Put $R_{\rm in}(x):=(x-1)^{-\rmi\varpi}g(x)$. Since
$r_*=r+r_{\rm s}\ln(x-1)$, both $f$ and
$\rme^{-\rmi kr_{\rm s}}r_{\rm s}^{-1}rR_{\rm in}$ equal
$\rme^{-\rmi kr_*}(1+o(1))$ as $r_*\to-\infty$; for $\mathrm{Im}\,k>0$
this makes them equal, and by analyticity in $k$
\begin{equation}\label{eq:f-near}
f=\frac{\rme^{-\rmi kr_{\rm s}}}{r_{\rm s}}\,r\,R_{\rm in}(r/r_{\rm s})
\qquad(\mathrm{Im}\,k>-1/2r_{\rm s}).
\end{equation}

\begin{alemma}[Near zone]\label{lem:near}
Let $K\subset\{\mathrm{Re}\,\kappa>0\}$ be compact and $X\geq2$. As
$c\to\infty$, $R_{\rm in}\to\mathrm P_l(2x-1)$ in $C^1([2,X])$,
uniformly for $\kappa\in K$, where $\mathrm P_l$ is the Legendre
polynomial.
\end{alemma}

\begin{proof}
As $c\to\infty$, uniformly for $\kappa\in K$, $\tau\to0$,
$\varpi^2\to0$ and $p\to0$ uniformly on $[1,X]$. The Taylor coefficients
$a_j$ of $g$ at $x=1$ obey
\[
(j+1)(j+1+2\tau)a_{j+1}=-\bigl[j(j+1+2\tau)+c_0\bigr]a_j-c_1a_{j-1}
-c_2a_{j-2},
\]
with $c_0=\tau-\varpi^2-l(l+1)+4r_{\rm s}-3\kappa^2r_{\rm s}^2$,
$c_1=2r_{\rm s}-3\kappa^2r_{\rm s}^2$ and $c_2=-\kappa^2r_{\rm s}^2$. For
$|\tau|\leq\frac14$ and $|c_0|+|c_1|+|c_2|\leq l(l+1)+1$ this gives
$\max_{i\leq j+1}|a_i|\leq(1+C/(j+1)^2)\max_{i\leq j}|a_i|$, so the $a_j$
are bounded uniformly in $j$ and in the parameters, and each is continuous
in them. Hence $g$ converges in $C^1([1,\frac32])$ to the solution $g_0$,
analytic at $1$ with $g_0(1)=1$, of the limiting equation
$x(x-1)g''+(2x-1)g'-l(l+1)g=0$. This is Legendre's equation in $2x-1$, so
$g_0(x)=\mathrm P_l(2x-1)$. On $[\frac32,X]$ the equation~\eqref{eq:g} is
regular and its coefficients converge, so $g\to g_0$ in $C^1([1,X])$.
Finally $(x-1)^{-\rmi\varpi}\to1$ in $C^1([2,X])$.
\end{proof}

\begin{alemma}[Matching]\label{lem:matching}
Let $K\subset\{\mathrm{Re}\,\kappa>0\}$ be compact and let
$u_{\rm reg}(\cdot,\kappa)$ be the solution of the hydrogen equation
$u''+(2/r-\kappa^2-l(l+1)/r^2)u=0$ with $u_{\rm reg}(r)=r^{l+1}(1+O(r))$ as
$r\to0$. Put
\[
\mathfrak a_c:=\frac{(2l)!}{(l!)^2}\,r_{\rm s}^{-l},\qquad
\varphi_c(r):=\mathfrak a_c^{-1}\,rN_c\,R_{\rm in}(r/r_{\rm s}) .
\]
For every $\bar r\geq1$ there is $C_{\bar r}$ such that for every $X\geq2$
\[
\limsup_{c\to\infty}\ \sup_{\kappa\in K}\ \sup_{Xr_{\rm s}\leq r\leq\bar r}
\ \frac{|\varphi_c(r)-u_{\rm reg}(r)|+r|\varphi_c'(r)-u_{\rm reg}'(r)|}
{r^{l+1}}\leq\frac{C_{\bar r}}{X}.
\]
In particular $\varphi_c\to u_{\rm reg}$ in $C^1$ on compact subsets of
$(0,\infty)$, uniformly for $\kappa\in K$.
\end{alemma}

\begin{proof}
The function $v:=rN_cR$ turns~\eqref{eq:radial-R} into
$v''+V_cv=0$ with
\[
V_c=\frac{2r-\kappa^2r^2}{(r-r_{\rm s})^2}-\frac{l(l+1)}{r(r-r_{\rm s})}
+\frac{r_{\rm s}^2}{4r^2(r-r_{\rm s})^2},
\]
and $V_\infty:=2/r-\kappa^2-l(l+1)/r^2$. For $r\geq2r_{\rm s}$,
$(1-r_{\rm s}/r)^{-2}-1\leq8r_{\rm s}/r$ gives
$|V_c-V_\infty|\leq Cr_{\rm s}(1+r)^2/r^3$, with $C$ depending on $K$. Let
$u_{\rm irr}$ be a hydrogen solution with
$W(u_{\rm reg},u_{\rm irr})=-(2l+1)$, analytic in $\kappa$ (for instance
the one whose expansion at $0$ has no term in $r^{l+1}$). On $(0,\bar r]$, uniformly for
$\kappa\in K$, $|u_{\rm reg}|+r|u_{\rm reg}'|\leq Cr^{l+1}$ and
$|u_{\rm irr}|+r|u_{\rm irr}'|\leq Cr^{-l}$; for $l=0$,
$u_{\rm irr}=1+O(r\ln r)$. Let $r_0:=Xr_{\rm s}$ and let $v_0$ be the
hydrogen solution with the Cauchy data of $\varphi_c$ at $r_0$. Variation
of constants gives
\[
\varphi_c(r)=v_0(r)+\int_{r_0}^r\Xi(r,t)\,(V_c-V_\infty)(t)\,
\varphi_c(t)\,\rmd t,\qquad
\Xi(r,t):=\frac{u_{\rm irr}(r)u_{\rm reg}(t)-u_{\rm reg}(r)u_{\rm irr}(t)}
{2l+1},
\]
and $|\Xi(r,t)|+r|\partial_r\Xi(r,t)|\leq Cr^{l+1}t^{-l}$ for
$t\leq r\leq\bar r$. Since $\int_{r_0}^{\bar r} t|V_c-V_\infty|\,\rmd t\leq
C_{\bar r}/X$ for $c$ large, Gronwall's inequality gives
\[
|\varphi_c-v_0|+r|\varphi_c'-v_0'|\leq\frac{C_{\bar r}}{X}\,\Theta_0\,r^{l+1}
\quad\text{on }[r_0,\bar r],\qquad
\Theta_0:=\frac{|\varphi_c(r_0)|+r_0|\varphi_c'(r_0)|}{r_0^{l+1}} .
\]
By Lemma~\ref{lem:near} and
$\mathrm P_l(2x-1)=\frac{(2l)!}{(l!)^2}x^l(1+O(1/x))$, the data of
$\varphi_c$ at $r_0$ are $\varphi_c(r_0)=r_0^{l+1}(1+e_1)$ and
$r_0\varphi_c'(r_0)=(l+1)r_0^{l+1}(1+e_2)$ with
$\limsup_{c\to\infty}(|e_1|+|e_2|)\leq C/X$; in particular
$\Theta_0\leq C$. Writing $v_0=a\,u_{\rm reg}+b\,u_{\rm irr}$ and computing
$a$ and $b$ from Wronskians at $r_0$ gives $|a-1|+r_0^{-2l-1}|b|\leq
C(|e_1|+|e_2|)+o(1)$ as $r_0\to0$. Since
$r_0^{2l+1}r^{-l}\leq r^{l+1}$ for $r\geq r_0$, this bounds
$|v_0-u_{\rm reg}|+r|v_0'-u_{\rm reg}'|$ by the same quantity times
$Cr^{l+1}$ on $[r_0,\bar r]$, and the claim follows.
\end{proof}

\begin{alemma}[Far zone]\label{lem:far}
Let $R_1>0$, $R_2:=R_1+2$, and let $\mathcal N$ be an open set with compact closure
in $\{\kappa:\mathrm{Re}\,\kappa>0,\ \kappa^2\notin(-\infty,2/R_1]\}$. For
$c$ large and $\kappa\in\mathcal N$ there is a solution $\psi_c(\cdot,\kappa)$ of
$v''+V_cv=0$ on $(r_{\rm s},\infty)$ that is square integrable near
$\infty$, with $\psi_c(R_2,\kappa)=1$, analytic in $\kappa$ and real for
real $\kappa$. As
$c\to\infty$, $\psi_c\to\psi_\infty$, the decaying hydrogen solution
with $\psi_\infty(R_2,\kappa)=1$, in $C^1$ on compact subsets of
$(0,\infty)$ and in $H^1((R_2,\infty))$, uniformly for $\kappa\in\mathcal N$.
\end{alemma}

\begin{proof}
Write $V_c=V_c^0-\kappa^2N_c^{-4}$, where
$V_c^0:=2r/(r-r_{\rm s})^2-l(l+1)/(r(r-r_{\rm s}))+r_{\rm s}^2/(4r^2(r-r_{\rm s})^2)$ does not
depend on $\kappa$. On $(R_1,\infty)$ the equation $v''+V_cv=0$ is
$T_cy=-\kappa^2y$ for $y:=N_c^{-2}v$, where
$T_c:=N_c^2(-\partial_r^2-V_c^0)N_c^2$ with a Dirichlet condition at $R_1$.
The map $w\mapsto N_c^{-1}w$ carries $L^2([r_*(R_1),\infty),\rmd r_*)$
unitarily onto $L^2((R_1,\infty),\rmd r)$ and the Dirichlet realisation
of $2h_l$ onto $T_c$. Since $W_l\geq-1/r$, $T_c\geq-2/R_1$, and so is
$T_\infty:=-\partial_r^2-2/r+l(l+1)/r^2$ with the same condition. Since
$-\kappa^2\notin[-2/R_1,\infty)$ on $\overline{\mathcal N}$,
$\|(T_c+\kappa^2)^{-1}\|\leq1/\eta_0$ on $\mathcal N$ for some $\eta_0>0$. For the same
reason no solution that is square integrable near $\infty$ vanishes at
a point $y\geq R_1$ when $\kappa\in\overline{\mathcal N}$: it would be an eigenfunction
of the Dirichlet problem on $[y,\infty)$, with the eigenvalue
$-\kappa^2\notin[-2/R_1,\infty)$. The difference $T_c-T_\infty$ is a second-order
operator whose coefficients are $O(r_{\rm s})$ uniformly on
$[R_1,\infty)$, so $\|(T_c-T_\infty)(T_\infty+2/R_1+1)^{-1}\|=O(r_{\rm s})$ and
$(T_c+\kappa^2)^{-1}\to(T_\infty+\kappa^2)^{-1}$ in
$\mathcal B(L^2,H^2)$, uniformly for $\kappa\in\mathcal N$. Fix $y_1\in(R_1,
R_1+1)$ and a non-negative $\zeta\in C_c^\infty((R_1,R_1+1))$ supported so
close to $y_1$ that $\int\zeta\vartheta\neq0$ on $\overline{\mathcal N}$, where $\vartheta$
is the solution of the limit equation with $\vartheta(R_1)=0$,
$\vartheta'(R_1)=1$, which does not vanish on $(R_1,\infty)$ by the same
argument. On $(R_1+1,\infty)$, $y_c:=(T_c+\kappa^2)^{-1}\zeta$ solves the
homogeneous equation and is square integrable; it is real for real
$\kappa$, because $T_c$ and $\zeta$ are real; and $y_c\to y_\infty$ in
$H^2((R_1+1,\infty))$; $y_\infty$ is a multiple of the decaying solution
with factor proportional to $\int\zeta\vartheta\neq0$, so $y_c(R_2)\neq0$
for $\kappa\in\mathcal N$ and $c$ large. Put $\psi_c:=N_c^2y_c/(N_c^2y_c)(R_2)$
on $(R_1+1,\infty)$ and continue it as a solution. Since $V_c\to V_\infty$
locally uniformly on $(0,\infty)$, the convergence at $R_2$ propagates
to compact subsets of $(0,\infty)$.
\end{proof}

\begin{proof}[Proof of Theorem~\ref{thm:resonances}(b)]
Since $rR_{\rm in}=\mathfrak a_c\varphi_c/N_c$, \eqref{eq:f-near} reads
$f=\rme^{-\rmi kr_{\rm s}}r_{\rm s}^{-1}\mathfrak a_c\,\varphi_c/N_c$, and the
solutions of $-\frac12w''+(W_l-s)w=0$ that are square integrable near
$+\infty$ are the multiples of $\psi_c/N_c$. The Wronskian of
$v_1/N_c$ and $v_2/N_c$ in $r_*$ equals that of $v_1$ and $v_2$ in $r$,
because $\partial_*=N_c^2\partial_r$. So the resonances correspond to the
zeros of
\[
\mathcal J_c(\kappa):=W\bigl(\varphi_c(\cdot,\kappa),\psi_c(\cdot,\kappa)\bigr)
\]
under the map $\kappa\mapsto z$ inverse to $z\mapsto(-2z-z^2/c^2)^{1/2}$,
which converges locally uniformly to $\kappa\mapsto-\kappa^2/2$. Choose
$\delta'\in(\delta,\frac12\operatorname{dist}(E_n,\sigma(H_S)\setminus
\{E_n\}))$, let $\mathcal N:=\{\kappa:\mathrm{Re}\,\kappa>0,\ |\kappa^2/2+E_n|<
\delta'\}$ and $R_1>2/(|E_n|-\delta')$, so that $\mathrm{Re}\,\kappa^2>
2/R_1$ on $\overline{\mathcal N}$. By Lemmas~\ref{lem:matching} and~\ref{lem:far},
evaluated at $r=1$, $\mathcal J_c\to \mathcal J_\infty:=W(u_{\rm reg},\psi_\infty)$
uniformly on $\mathcal N$. Integrating
$\partial_rW\bigl(u_{\rm reg}(\cdot,\kappa_n),\psi_\infty(\cdot,\kappa)
\bigr)=(\kappa^2-\kappa_n^2)\,u_{\rm reg}(\cdot,\kappa_n)\,
\psi_\infty(\cdot,\kappa)$ over $(0,\infty)$, where this Wronskian tends
to $0$ at $\infty$ and to $\mathcal J_\infty(\kappa)$ at $0$, gives
\begin{equation}\label{eq:Dinf}
\mathcal J_\infty(\kappa)=-(\kappa^2-\kappa_n^2)\int_0^\infty u_{\rm reg}(r,
\kappa_n)\,\psi_\infty(r,\kappa)\,\rmd r .
\end{equation}
The zeros of $\mathcal J_\infty$ are the $\kappa$ at which $u_{\rm reg}$ decays,
that is, the eigenvalues of $H_S$ with angular momentum $l$; in $\overline{\mathcal N}$
there is only $\kappa_n$, and it is simple, since
$\psi_\infty(\cdot,\kappa_n)=u_{\rm reg}(\cdot,\kappa_n)/u_{\rm reg}(R_2,
\kappa_n)$ and~\eqref{eq:Dinf} give
\[
\mathcal J_\infty'(\kappa_n)=-\frac{2\kappa_n\|u_{\rm reg}\|^2}
{u_{\rm reg}(R_2,\kappa_n)}\neq0,\qquad
\|u_{\rm reg}\|:=\|u_{\rm reg}(\cdot,\kappa_n)\|_{L^2(0,\infty)} .
\]
By Hurwitz's theorem, for $c$ large $\mathcal J_c$ has exactly one zero
$\kappa_c$ in $\mathcal N$; it is simple, and $\kappa_c\to\kappa_n$. The map
$z\mapsto(-2z-z^2/c^2)^{1/2}$ converges uniformly on the closed disc
$\{|z-E_n|\leq\delta\}$ to $z\mapsto(-2z)^{1/2}$, which maps it into the
compact set $\{|\kappa^2/2+E_n|\leq\delta\}\subset\mathcal N$. So for $c$
large the resonances in $\{|z-E_n|<\delta\}$ are images of zeros of
$\mathcal J_c$ in $\mathcal N$, and since the image $z(\kappa_c)$ of $\kappa_c$ tends
to $E_n$, there is exactly one. The map $\kappa\mapsto z=-c^2+c(c^2-
\kappa^2)^{1/2}$ has non-zero derivative on $\mathcal N$, so the pole has order
one, like the zero. By the proof of~(a) its residue is a non-zero
multiple of the operator with kernel $\chi(x)\chi(y)F(|x|)F(|y|)
\sum_mY_{lm}(\hat x)\overline{Y_{lm}(\hat y)}$, where
$F(r):=f(r_*(r))/(N_c(r)r)$. Since $\chi F\neq0$ on a non-empty open
set, whose directions form an open subset of $S^2$ on which the
$Y_{lm}$ are linearly independent, the residue has rank $2l+1$.
\end{proof}

\begin{proof}[Proof of Theorem~\ref{thm:resonances}(c), imaginary part]
For real $\kappa$ the coefficients of $v''+V_cv=0$ are real. Put
$\varphi_c^\sharp(\cdot,\kappa):=\overline{\varphi_c(\cdot,\bar\kappa)}$,
$F_1:=(\varphi_c+\varphi_c^\sharp)/2$ and
$F_2:=(\varphi_c-\varphi_c^\sharp)/2\rmi$. They are solutions, analytic
in $\kappa\in\mathcal N$ and real for real $\kappa$, with $\varphi_c=F_1+\rmi
F_2$; since $u_{\rm reg}$ is real for real $\kappa$, $F_1\to u_{\rm reg}$
and $F_2\to0$ as in Lemma~\ref{lem:matching}. Hence $\mathcal J_c=J_1+\rmi J_2$
with $J_j:=W(F_j,\psi_c)$ real on $\mathcal N\cap\R$, $J_1\to \mathcal J_\infty$ and
$J_2\to0$ uniformly on $\mathcal N$. By Hurwitz's theorem $J_1$ has exactly one
zero $\kappa_0$ in $\mathcal N$; it is real, because the zeros of $J_1$ come in
conjugate pairs, and $\kappa_0\to\kappa_n$.

For real $\kappa$, $k$ is real and $W(f,\bar f)=2\rmi k$, as one sees at
$r_*\to-\infty$ from Lemma~\ref{lem:continuation}(i). By~\eqref{eq:f-near}
and the invariance of Wronskians used above,
$W(\varphi_c,\bar\varphi_c)=2\rmi kr_{\rm s}^2/\mathfrak a_c^2$, that is
\begin{equation}\label{eq:flux}
W(F_1,F_2)=-\Upsilon_c,\qquad
\Upsilon_c:=\frac{kr_{\rm s}^2}{\mathfrak a_c^2}=\Bigl(\frac{(l!)^2}{(2l)!}
\Bigr)^2k\,r_{\rm s}^{2l+2} .
\end{equation}
At $\kappa_0$, $F_1=\lambda\psi_c$ with
$\lambda:=F_1(R_2,\kappa_0)\to u_{\rm reg}(R_2,\kappa_n)\neq0$,
so~\eqref{eq:flux} gives $J_2(\kappa_0)=\Upsilon_c/\lambda$. Let
$\kappa_*:=\kappa_0-\rmi J_2(\kappa_0)/J_1'(\kappa_0)$,
$d_c:=|\kappa_*-\kappa_0|$ and
$\mathcal L(\kappa):=J_1'(\kappa_0)(\kappa-\kappa_0)+\rmi J_2(\kappa_0)$. On a
smaller disc $\mathcal N'$ about $\kappa_n$, $|J_1''|\leq C$ and
$\sup_{\mathcal N'}|J_2'|\to0$ by Cauchy's estimates. On the circle
$|\kappa-\kappa_*|=\eta d_c$, $0<\eta<1$,
\[
|\mathcal J_c-\mathcal L|\leq\bigl[C(1+\eta)^2d_c+(1+\eta)\sup_{\mathcal N'}|J_2'|\bigr]d_c
<\eta|J_1'(\kappa_0)|d_c=|\mathcal L|
\]
for $c$ large. By Rouch\'e's theorem $\kappa_c$ lies in the disc of radius
$\eta d_c$ about $\kappa_*$, and $\eta$ is arbitrary. Since
$\kappa_*-\kappa_0$ is imaginary,
\[
\mathrm{Re}\,\kappa_c=\kappa_0+o(\Upsilon_c),\qquad
\mathrm{Im}\,\kappa_c=-\frac{J_2(\kappa_0)}{J_1'(\kappa_0)}+
o(\Upsilon_c)=\frac{\Upsilon_c}{2\kappa_n\|u_{\rm reg}\|^2}
\bigl(1+o(1)\bigr).
\]
The resonance of $S_c$ is $s_c=-\kappa_c^2/2$, so
$\mathrm{Im}\,s_c=-\Upsilon_c/(2\|u_{\rm reg}\|^2)\,(1+o(1))$, and
inverting $s=z+z^2/2c^2$ gives $\mathrm{Im}\,z_{nl}=\mathrm{Im}\,s_c\,
(1+O(c^{-2}))$. Finally $\|u_{\rm reg}\|^{-2}=a_{nl}^2$, where
$a_{nl}:=\lim_{r\to0}r^{-l}R_{nl}(r)$, with $R_{nl}$ the normalised
radial eigenfunction of Sec.~\ref{sec:atoms}, and the closed form of the hydrogen
eigenfunctions~\cite{BetheSalpeter1957} gives $a_{nl}^2=2^{2l+2}(n+l)!/
\bigl(n^{2l+4}(n-l-1)!((2l+1)!)^2\bigr)$. With $k=c(1+O(c^{-2}))$,
$r_{\rm s}=2\alpha^2$ and $c=1/\alpha$ this is~\eqref{eq:res-imag}.
\end{proof}

\begin{proof}[Proof of Theorem~\ref{thm:resonances}(c), real part]
Let $R_c:=F_1(\cdot,\kappa_0)/(rN_c)$ and
$R_n:=u_{\rm reg}(\cdot,\kappa_n)/r$. Both are real, $R_c$
solves~\eqref{eq:radial-R} at $\kappa_0$ and decays at $\infty$, and $R_n$
solves the hydrogen equation $(r^2R')'+(2r-\kappa_n^2r^2-l(l+1))R=0$. Since
\[
\frac{2r^2-\kappa^2r^3}{r-r_{\rm s}}=\frac{r_{\rm s}\varpi^2}{r-r_{\rm s}}
+2(r+r_{\rm s})-\kappa^2(r^2+rr_{\rm s}+r_{\rm s}^2),
\]
the function $\mathcal W:=r^2N_c^2R_c'R_n-r^2R_n'R_c$ satisfies
\[
\mathcal W'=-\Bigl[\frac{r_{\rm s}\varpi_0^2}{r-r_{\rm s}}+2r_{\rm s}
-\kappa_0^2r_{\rm s}(r+r_{\rm s})\Bigr]R_cR_n+(\kappa_0^2-\kappa_n^2)
r^2R_cR_n-r_{\rm s}rR_c'R_n',
\]
with $\varpi_0:=\varpi(\kappa_0)$. Integrating over $(2r_{\rm s},\infty)$,
where $\mathcal W\to0$ at $\infty$,
\begin{equation}\label{eq:real-identity}
(\kappa_0^2-\kappa_n^2)\int_{2r_{\rm s}}^\infty r^2R_cR_n\,\rmd r=
-\mathcal W(2r_{\rm s})+\int_{2r_{\rm s}}^\infty\Bigl\{\Bigl[
\frac{r_{\rm s}\varpi_0^2}{r-r_{\rm s}}+2r_{\rm s}-\kappa_0^2r_{\rm s}
(r+r_{\rm s})\Bigr]R_cR_n+r_{\rm s}rR_c'R_n'\Bigr\}\rmd r .
\end{equation}
Divide by $r_{\rm s}$ and let $c\to\infty$. Fix $X\geq2$ and take
$\bar r:=R_2$ in Lemma~\ref{lem:matching}. On
$[2r_{\rm s},Xr_{\rm s}]$,
$R_c=\mathfrak a_c^{-1}\,\mathrm{Re}\,R_{\rm in}(r/r_{\rm s})$ by the definition
of $\varphi_c$, so Lemma~\ref{lem:near} gives
$|R_c|+r_{\rm s}|R_c'|\leq Cr_{\rm s}^l$ there, and
$r_{\rm s}R_c'(2r_{\rm s})\to0$ for $l=0$, because $\mathrm P_0$ is
constant. With $r^2N_c^2=2r_{\rm s}^2$ at $r=2r_{\rm s}$, $R_n=O(r^l)$ and
$R_n'=O(r^{\max(l-1,0)})$, this gives $\mathcal W(2r_{\rm s})=o(r_{\rm s})$, and
the contributions of $[2r_{\rm s},Xr_{\rm s}]$ to the integrals divided
by $r_{\rm s}$ tend to $0$. On $[Xr_{\rm s},\bar r]$,
Lemma~\ref{lem:matching} gives $|R_c-R_n|+r|R_c'-R_n'|\leq(C_{\bar r}/X+
o(1))r^l$, and $R_c\to R_n$ with its derivative, locally uniformly on
$(0,\infty)$. On $[R_2,\infty)$, $F_1(\cdot,\kappa_0)=\lambda\psi_c(\cdot,
\kappa_0)\to u_{\rm reg}(\cdot,\kappa_n)$ in $H^1$ by
Lemma~\ref{lem:far}. The term with $\varpi_0^2=O(r_{\rm s})$ contributes
$O(\varpi_0^2\ln(1/r_{\rm s}))=o(1)$ after division by $r_{\rm s}$.
Letting $c\to\infty$, then $X\to\infty$,
\[
\frac{\kappa_0^2-\kappa_n^2}{r_{\rm s}}\longrightarrow
\frac{\int_0^\infty\bigl[(2-\kappa_n^2r)R_n^2+rR_n'^2\bigr]\rmd r}
{\int_0^\infty r^2R_n^2\,\rmd r}=-D_1(u_{nl}),
\]
by~\eqref{eq:D1-schw} in these units, for $u_{nl}=R_nY/\|u_{\rm reg}\|$,
with $Y$ as in Sec.~\ref{sec:atoms}. Hence $-\kappa_0^2/2=E_n+(r_{\rm
s}/2)D_1(u_{nl})+o(r_{\rm s})$, and $r_{\rm s}/2=c^{-2}$. From the proof of
the imaginary part, $\mathrm{Re}\,s_c=-\bigl((\mathrm{Re}\,\kappa_c)^2-
(\mathrm{Im}\,\kappa_c)^2\bigr)/2=-\kappa_0^2/2+o(\Upsilon_c)$, and
$\Upsilon_c=o(c^{-2})$. Inverting $s=z+z^2/2c^2$ gives
$\mathrm{Re}\,z_{nl}=\mathrm{Re}\,s_c-(\mathrm{Re}\,s_c)^2/2c^2+
o(c^{-2})$, which is~\eqref{eq:res-real-D1}; the computation of $D_1$ in
Sec.~\ref{sec:atoms} turns it into~\eqref{eq:res-real}.
\end{proof}

\section{Numerical methods}\label{app:numerics}

All computations use units $\hbar=m=G=M=1$, in which the gravitational
Bohr radius is $1$ and $c=\alpha^{-1}$.

\emph{Stars (Table~\ref{tab:stars}).} For the uniform-density star, with
$u=f(r)Y_{lm}$, the eigenvalue problem of $S_c$ on $\K_c$ reads
$-\frac12(pf')'+Qf=s\,wf$, with $p=N_ch_{rr}^{-1/2}r^2$,
$w=N_c^{-1}h_{rr}^{1/2}r^2$ and $Q=w[N_c^2l(l+1)/2r^2+\Phi_c]$. It is
solved by shooting in $r$ with an eighth-order Runge--Kutta method,
separately on each side of $r=R$, where the coefficients are smooth, and
the eigenvalues $s_n$ are converted to $E_n(c)=F_c(s_n)$. They are stable
to $10^{-13}$ under changes of the tolerance and of the matching point;
since $r$ does not depend on $c$, this contributes less than $10^{-8}$ to
$c^2(E(c)-E)$ at $c=128$. A cell-centred finite-difference scheme on a
grid with $r=R$ on a cell face, extrapolated in the mesh size, agrees to
$2\times10^{-5}$. The coefficient $D_1(u)$ is evaluated by quadrature
from the $O(c^{-2})$ parts of $p$, $w$ and $Q$.

\emph{Resolvents (Table~\ref{tab:schw}).} In each partial wave the
Green's function of $-\frac12\partial_*^2+W_l-z$ is built from the
solution that is ingoing at the horizon and the solution that decays at
infinity, integrated with an eighth-order Runge--Kutta method (relative
tolerance $10^{-10}$) from the asymptotic regions to a matching point,
and applied to the test function by trapezoidal quadrature on a grid of
$160$ points per local wavelength, with step at most $5\times10^{-4}$ in
$r_*/a_0$. To the right of the support of the test function
Green's formula involves only the decaying solution, and it is used in
that form. The result is compared with the resolvent of the gravitational
hydrogen atom in the same partial wave, computed by second-order finite
differences on $[0,80]$ with the Friedrichs condition at the origin and
step $2\times10^{-4}$, whose error is of order $10^{-8}$. The test
function in the partial wave $l$ is $g(x)=b(|x|)Y_{lm}(x/|x|)$, with
$b(r)=\exp(-1/(1-t^2))$ for $|t|<1$ and $b=0$ otherwise,
$t=(r/a_0-7/4)/(5/4)$, in units of $a_0^{-3/2}$; it is supported in
$0.5\leq r/a_0\leq3$, and $\|g\|^2=0.539$. With $20$,
$40$ and $80$ points per wavelength, and the maximal step scaled
accordingly, the entries differ from those of the table by at most
$7\times10^{-3}$, $1.7\times10^{-3}$ and $3.3\times10^{-4}$ in relative
terms, so the quadrature error is of second order and that of the table
is about $10^{-4}$. The same code gives the part of the error carried by the
region near the horizon; for $s$-waves it accounts for almost all of
$\mathcal{R}(c)$, and for $p$-waves it is negligible.
The spectral density of Sec.~\ref{sec:schw-numerics} is the imaginary
part of the same Green's function just above the real axis, fitted by a
Lorentzian plus a linear background.

\emph{Quasi-bound states (Table~\ref{tab:atoms}).} Leaver's
continued-fraction method~\cite{Leaver1985}, in the form of
Ref.~\onlinecite{Dolan2007} for the massive scalar field, was rederived
symbolically for Schwarzschild and checked against
Ref.~\onlinecite{Dolan2007} with vanishing spin. The
continued fraction is solved for the decay constant
$q=-(\mu^2-\omega^2)^{1/2}$ rather than for $\omega$, to avoid the
cancellation between $\omega$ and $\mu$ at small $\alpha$. In double
precision, truncation depths $N_{\rm t}$ and $2N_{\rm t}$ agree to a
relative precision of $2\times10^{-16}$ or better. For $(n,l)=(3,2)$ the
ratio $|\mathrm{Im}\,q|/|q|$ is below $10^{-15}$, at the rounding level
of double precision, and the root is computed with $40$ significant
digits; the same code reproduces the double-precision width for
$(n,l)=(2,1)$ at $\alpha=0.02$ to $10^{-5}$.

\emph{Coulomb diffraction (Remark~\ref{rem:coulomb}).} For $|x|=|y|=1$
at an angle of $0.6$ the function $\mathcal{A}(\xi)$ is computed by
quadrature; one-sided fits give a jump of $\mathcal{A}'$ at $\xi_0$ of
$-4.0658$, against
$-4.0690$ from the closed form, while a Coulomb potential smoothed on
the scale $0.05$ gives $-0.007$. The diffracted component of
$\int\rme^{\rmi kR_{xy}\xi}\mathcal{A}(\xi)\,\rmd\xi$ decays as $k^{-2}$
with the
predicted amplitude for the Coulomb potential and much faster for the
smoothed one.

\begin{acknowledgments}
This work was supported by the R+D+I efforts from guane Enterprises.
\end{acknowledgments}

\section*{Data Availability}

The scripts that produce the numerical results of
Tables~\ref{tab:stars}--\ref{tab:atoms} and Appendix~\ref{app:numerics}
are available from the author upon reasonable request.

\bibliography{newton_cartan_limit}

@article{AlbeverioGesztesyHoeghKrohnStreit1983,
  author  = {Albeverio, Sergio and Gesztesy, Fritz and H{\o}egh-Krohn, Raphael and Streit, Ludwig},
  title   = {Charged particles with short range interactions},
  journal = {Annales de l'Institut Henri Poincar{\'e}, Section A, Physique Th{\'e}orique},
  volume  = {38},
  number  = {3},
  pages   = {263--293},
  year    = {1983},
  url     = {http://www.numdam.org/item/AIHPA_1983__38_3_263_0/}
}

@article{Araki1963,
  author  = {Araki, Huzihiro},
  title   = {A lattice of von {Neumann} algebras associated with the quantum theory of a free {Bose} field},
  journal = {Journal of Mathematical Physics},
  volume  = {4},
  pages   = {1343--1362},
  year    = {1963},
  doi     = {10.1063/1.1703912}
}

@article{Araki1964b,
  author  = {Araki, Huzihiro},
  title   = {Type of von {Neumann} algebra associated with free field},
  journal = {Progress of Theoretical Physics},
  volume  = {32},
  pages   = {956--965},
  year    = {1964},
  doi     = {10.1143/PTP.32.956}
}

@article{ArakiWoods1963,
  author  = {Araki, H. and Woods, E. J.},
  title   = {Representations of the canonical commutation relations
             describing a nonrelativistic infinite free {B}ose gas},
  journal = {Journal of Mathematical Physics},
  volume  = {4},
  number  = {5},
  pages   = {637--662},
  year    = {1963},
  doi     = {10.1063/1.1704002}
}

@article{ArvanitakiDubovsky2011,
  author  = {Arvanitaki, Asimina and Dubovsky, Sergei},
  title   = {Exploring the string axiverse with precision black hole physics},
  journal = {Physical Review D},
  volume  = {83},
  pages   = {044026},
  year    = {2011},
  doi     = {10.1103/physrevd.83.044026}
}

@article{AshbaughHarrell1982,
  author  = {Ashbaugh, Mark S. and Harrell, Evans M.},
  title   = {Perturbation theory for shape resonances and large barrier potentials},
  journal = {Communications in Mathematical Physics},
  volume  = {83},
  number  = {2},
  pages   = {151--170},
  year    = {1982},
  doi     = {10.1007/BF01976039}
}

@article{BachelotMotetBachelot1993,
  author  = {Bachelot, Alain and Motet-Bachelot, Agn{\`e}s},
  title   = {Les r{\'e}sonances d'un trou noir de {Schwarzschild}},
  journal = {Annales de l'Institut Henri Poincar{\'e}, Physique Th{\'e}orique},
  volume  = {59},
  number  = {1},
  pages   = {3--68},
  year    = {1993},
  url     = {http://www.numdam.org/item/AIHPA_1993__59_1_3_0/}
}

@article{Bachelot1994,
  author  = {Bachelot, Alain},
  title   = {Asymptotic completeness for the {Klein--Gordon} equation on the {Schwarzschild} metric},
  journal = {Annales de l'Institut Henri Poincar{\'e}, Physique Th{\'e}orique},
  volume  = {61},
  number  = {4},
  pages   = {411--441},
  year    = {1994},
  url     = {http://www.numdam.org/item/AIHPA_1994__61_4_411_0/}
}

@article{Banerjee2020,
  author  = {Banerjee, Abhishek and Budker, Dmitry and Eby, Joshua and Flambaum, Victor V. and Kim, Hyungjin and Matsedonskyi, Oleksii and Perez, Gilad},
  title   = {Searching for {Earth}/{Solar} axion halos},
  journal = {Journal of High Energy Physics},
  volume  = {2020},
  number  = {9},
  pages   = {004},
  year    = {2020},
  doi     = {10.1007/jhep09(2020)004}
}

@article{BaoXuZhang2022,
  author  = {Bao, Shou-Shan and Xu, Qi-Xuan and Zhang, Hong},
  title   = {Improved analytic solution of black hole superradiance},
  journal = {Physical Review D},
  volume  = {106},
  number  = {6},
  pages   = {064016},
  year    = {2022},
  doi     = {10.1103/PhysRevD.106.064016},
  eprint  = {2201.10941},
  archiveprefix = {arXiv}
}

@article{BaumannChiaPorto2019,
  author  = {Baumann, Daniel and Chia, Horng Sheng and Porto, Rafael A.},
  title   = {Probing ultralight bosons with binary black holes},
  journal = {Physical Review D},
  volume  = {99},
  pages   = {044001},
  year    = {2019},
  doi     = {10.1103/physrevd.99.044001}
}

@article{BaumannChiaStoutTerHaar2019,
  author  = {Baumann, Daniel and Chia, Horng Sheng and Stout, John and ter Haar, Lotte},
  title   = {The spectra of gravitational atoms},
  journal = {Journal of Cosmology and Astroparticle Physics},
  volume  = {2019},
  number  = {12},
  pages   = {006},
  year    = {2019},
  doi     = {10.1088/1475-7516/2019/12/006}
}

@article{BerryKlein1996,
  author  = {Berry, M. V. and Klein, S.},
  title   = {Integer, fractional and fractal {T}albot effects},
  journal = {Journal of Modern Optics},
  volume  = {43},
  number  = {10},
  pages   = {2139--2164},
  year    = {1996},
  doi     = {10.1080/09500349608232876}
}

@book{BetheSalpeter1957,
  author    = {Bethe, Hans A. and Salpeter, Edwin E.},
  title     = {Quantum mechanics of one- and two-electron atoms},
  publisher = {Springer},
  address   = {Berlin, Heidelberg},
  year      = {1957},
  doi       = {10.1007/978-3-662-12869-5}
}

@article{Borchers1961,
  author  = {Borchers, H. J.},
  title   = {{{\"U}ber} die {Vollst{\"a}ndigkeit} lorentzinvarianter {Felder} in einer zeitartigen {R{\"o}hre}},
  journal = {Il Nuovo Cimento},
  volume  = {19},
  pages   = {787--793},
  year    = {1961},
  doi     = {10.1007/bf02733373}
}

@article{Boulware1975,
  author  = {Boulware, David G.},
  title   = {Quantum field theory in {Schwarzschild} and {Rindler} spaces},
  journal = {Physical Review D},
  volume  = {11},
  pages   = {1404--1423},
  year    = {1975},
  doi     = {10.1103/physrevd.11.1404}
}

@book{BritoCardosoPani2020,
  author    = {Brito, Richard and Cardoso, Vitor and Pani, Paolo},
  title     = {Superradiance: new frontiers in black hole physics},
  series    = {Lecture Notes in Physics},
  publisher = {Springer International Publishing},
  year      = {2020},
  doi       = {10.1007/978-3-030-46622-0}
}

@article{BuchholzDAntoniFredenhagen1987,
  author  = {Buchholz, D. and D'Antoni, C. and Fredenhagen, K.},
  title   = {The universal structure of local algebras},
  journal = {Communications in Mathematical Physics},
  volume  = {111},
  pages   = {123--135},
  year    = {1987},
  doi     = {10.1007/bf01239019}
}

@article{CirincioneChernoff1981,
  author  = {Cirincione, Ross J. and Chernoff, Paul R.},
  title   = {Dirac and {K}lein--{G}ordon equations: convergence of
             solutions in the nonrelativistic limit},
  journal = {Communications in Mathematical Physics},
  volume  = {79},
  number  = {1},
  pages   = {33--46},
  year    = {1981},
  doi     = {10.1007/BF01208284}
}

@article{Detweiler1980,
  author  = {Detweiler, Steven},
  title   = {{Klein--Gordon} equation and rotating black holes},
  journal = {Physical Review D},
  volume  = {22},
  pages   = {2323--2326},
  year    = {1980},
  doi     = {10.1103/physrevd.22.2323}
}

@article{Dimock1977,
  author  = {Dimock, J.},
  title   = {The non-relativistic limit of {$P(\phi)_2$} quantum field
             theories: two-particle phenomena},
  journal = {Communications in Mathematical Physics},
  volume  = {57},
  number  = {1},
  pages   = {51--66},
  year    = {1977},
  doi     = {10.1007/BF01651693}
}

@article{Dimock1980,
  author  = {Dimock, J.},
  title   = {Algebras of local observables on a manifold},
  journal = {Communications in Mathematical Physics},
  volume  = {77},
  number  = {3},
  pages   = {219--228},
  year    = {1980},
  doi     = {10.1007/BF01269921}
}

@article{Dolan2007,
  author  = {Dolan, Sam R.},
  title   = {Instability of the massive {Klein--Gordon} field on the {Kerr} spacetime},
  journal = {Physical Review D},
  volume  = {76},
  pages   = {084001},
  year    = {2007},
  doi     = {10.1103/physrevd.76.084001}
}

@book{DyatlovZworski2019,
  author    = {Dyatlov, Semyon and Zworski, Maciej},
  title     = {Mathematical theory of scattering resonances},
  series    = {Graduate Studies in Mathematics},
  publisher = {American Mathematical Society},
  year      = {2019},
  doi       = {10.1090/gsm/200}
}

@article{FalconeConti2024,
  author        = {Falcone, Riccardo and Conti, Claudio},
  title         = {Localization in quantum field theory},
  journal       = {Reviews in Physics},
  volume        = {12},
  pages         = {100095},
  year          = {2024},
  doi           = {10.1016/j.revip.2024.100095},
  eprint        = {2312.15348},
  archivePrefix = {arXiv},
  primaryClass  = {hep-th},
  note          = {arXiv:2312.15348 [hep-th]}
}

@article{FeshbachVillars1958,
  author  = {Feshbach, Herman and Villars, Felix},
  title   = {Elementary relativistic wave mechanics of spin 0 and spin 1/2 particles},
  journal = {Reviews of Modern Physics},
  volume  = {30},
  number  = {1},
  pages   = {24--45},
  year    = {1958},
  doi     = {10.1103/RevModPhys.30.24}
}

@article{FullingNarcowichWald1981,
  author  = {Fulling, Stephen A. and Narcowich, Francis J. and Wald, Robert M.},
  title   = {Singularity structure of the two-point function in quantum field theory in curved spacetime, {II}},
  journal = {Annals of Physics},
  volume  = {136},
  number  = {2},
  pages   = {243--272},
  year    = {1981},
  doi     = {10.1016/0003-4916(81)90098-1}
}

@article{GaveauSchulman1986,
  author  = {Gaveau, B. and Schulman, L. S.},
  title   = {Explicit time-dependent {Schrodinger} propagators},
  journal = {Journal of Physics A: Mathematical and General},
  volume  = {19},
  pages   = {1833--1846},
  year    = {1986},
  doi     = {10.1088/0305-4470/19/10/024}
}

@article{GesztesyGrosseThaller1985,
  author  = {Gesztesy, F. and Grosse, H. and Thaller, B.},
  title   = {First-order relativistic corrections and spectral concentration},
  journal = {Advances in Applied Mathematics},
  volume  = {6},
  pages   = {159--176},
  year    = {1985},
  doi     = {10.1016/0196-8858(85)90009-0}
}

@book{GilbargTrudinger2001,
  author    = {Gilbarg, David and Trudinger, Neil S.},
  title     = {Elliptic Partial Differential Equations of Second Order},
  series    = {Classics in Mathematics},
  publisher = {Springer},
  year      = {2001},
  doi       = {10.1007/978-3-642-61798-0}
}

@article{GiuliniGrossardt2012,
  author  = {Giulini, Domenico and Gro{\ss}ardt, Andr{\'e}},
  title   = {The {S}chr{\"o}dinger--{N}ewton equation as a non-relativistic
             limit of self-gravitating {K}lein--{G}ordon and {D}irac fields},
  journal = {Classical and Quantum Gravity},
  volume  = {29},
  number  = {21},
  pages   = {215010},
  year    = {2012},
  doi     = {10.1088/0264-9381/29/21/215010}
}

@article{HaagHugenholtzWinnink1967,
  author  = {Haag, R. and Hugenholtz, N. M. and Winnink, M.},
  title   = {On the equilibrium states in quantum statistical mechanics},
  journal = {Communications in Mathematical Physics},
  volume  = {5},
  number  = {3},
  pages   = {215--236},
  year    = {1967},
  doi     = {10.1007/BF01646342}
}

@article{HartleHawking1976,
  author  = {Hartle, J. B. and Hawking, S. W.},
  title   = {Path-integral derivation of black-hole radiance},
  journal = {Physical Review D},
  volume  = {13},
  number  = {8},
  pages   = {2188--2203},
  year    = {1976},
  doi     = {10.1103/PhysRevD.13.2188}
}

@article{HartongObersOling2023,
  author  = {Hartong, Jelle and Obers, Niels A. and Oling, Gerben},
  title   = {Review on non-relativistic gravity},
  journal = {Frontiers in Physics},
  volume  = {11},
  pages   = {1116888},
  year    = {2023},
  doi     = {10.3389/fphy.2023.1116888}
}

@misc{HassellJiaSussmanVasy2025a,
  author        = {Hassell, Andrew and Jia, Qiuye and Sussman, Ethan and Vasy, Andras},
  title         = {Microlocal analysis of the non-relativistic limit of the {Klein--Gordon} equation: estimates},
  year          = {2025},
  eprint        = {2509.09518},
  archiveprefix = {arXiv}
}

@misc{HassellJiaSussmanVasy2025b,
  author        = {Hassell, Andrew and Jia, Qiuye and Sussman, Ethan and Vasy, Andras},
  title         = {Microlocal analysis of the non-relativistic limit of the {Klein--Gordon} equation: asymptotics},
  year          = {2025},
  eprint        = {2511.08724},
  archiveprefix = {arXiv}
}

@article{HintzXie2022,
  author  = {Hintz, Peter and Xie, YuQing},
  title   = {Quasinormal modes of small {Schwarzschild--de Sitter} black holes},
  journal = {Journal of Mathematical Physics},
  volume  = {63},
  pages   = {011509},
  year    = {2022},
  doi     = {10.1063/5.0062985}
}

@book{HollandsSanders2018,
  author    = {Hollands, Stefan and Sanders, Ko},
  title     = {Entanglement measures and their properties in quantum field theory},
  series    = {SpringerBriefs in Mathematical Physics},
  publisher = {Springer International Publishing},
  year      = {2018},
  doi       = {10.1007/978-3-319-94902-4}
}

@article{Hostler1964,
  author  = {Hostler, Levere},
  title   = {{Coulomb} {Green's} functions and the {Furry} approximation},
  journal = {Journal of Mathematical Physics},
  volume  = {5},
  pages   = {591--611},
  year    = {1964},
  doi     = {10.1063/1.1704153}
}

@article{Hunziker1975,
  author  = {Hunziker, W.},
  title   = {On the nonrelativistic limit of the {Dirac} theory},
  journal = {Communications in Mathematical Physics},
  volume  = {40},
  pages   = {215--222},
  year    = {1975},
  doi     = {10.1007/bf01609998}
}

@article{Hunziker1990,
  author  = {Hunziker, Walter},
  title   = {Resonances, metastable states and exponential decay laws in perturbation theory},
  journal = {Communications in Mathematical Physics},
  volume  = {132},
  pages   = {177--188},
  year    = {1990},
  doi     = {10.1007/bf02278006}
}

@article{ItoYamada2005,
  author  = {Ito, Hiroshi T. and Yamada, Osanobu},
  title   = {A note on the nonrelativistic limit of {Dirac} operators and spectral concentration},
  journal = {Proceedings of the Japan Academy, Series A, Mathematical Sciences},
  volume  = {81},
  pages   = {157--161},
  year    = {2005},
  doi     = {10.3792/pjaa.81.157}
}

@book{Kato1995,
  author    = {Kato, Tosio},
  title     = {Perturbation theory for linear operators},
  series    = {Classics in Mathematics},
  publisher = {Springer Berlin Heidelberg},
  year      = {1995},
  doi       = {10.1007/978-3-642-66282-9}
}

@article{Kay1978,
  author  = {Kay, Bernard S.},
  title   = {Linear spin-zero quantum fields in external gravitational and scalar fields. {I}. {A} one particle structure for the stationary case},
  journal = {Communications in Mathematical Physics},
  volume  = {62},
  pages   = {55--70},
  year    = {1978},
  doi     = {10.1007/bf01940330}
}

@article{KayWald1991,
  author  = {Kay, Bernard S. and Wald, Robert M.},
  title   = {Theorems on the uniqueness and thermal properties of
             stationary, nonsingular, quasifree states on spacetimes
             with a bifurcate {Killing} horizon},
  journal = {Physics Reports},
  volume  = {207},
  number  = {2},
  pages   = {49--136},
  year    = {1991},
  doi     = {10.1016/0370-1573(91)90015-E}
}

@article{KieferSingh1991,
  author  = {Kiefer, Claus and Singh, Tejinder P.},
  title   = {Quantum gravitational corrections to the functional {Schr{\"o}dinger} equation},
  journal = {Physical Review D},
  volume  = {44},
  pages   = {1067--1076},
  year    = {1991},
  doi     = {10.1103/physrevd.44.1067}
}

@article{Kuchar1980,
  author  = {Kucha{\v{r}}, Karel},
  title   = {Gravitation, geometry, and nonrelativistic quantum theory},
  journal = {Physical Review D},
  volume  = {22},
  number  = {6},
  pages   = {1285--1299},
  year    = {1980},
  doi     = {10.1103/PhysRevD.22.1285}
}

@article{KuwaeShioya2003,
  author  = {Kuwae, Kazuhiro and Shioya, Takashi},
  title   = {Convergence of spectral structures: a functional analytic theory and its applications to spectral geometry},
  journal = {Communications in Analysis and Geometry},
  volume  = {11},
  pages   = {599--673},
  year    = {2003},
  doi     = {10.4310/cag.2003.v11.n4.a1}
}

@article{Lammerzahl1995,
  author  = {L{\"a}mmerzahl, Claus},
  title   = {A {Hamilton} operator for quantum optics in gravitational fields},
  journal = {Physics Letters A},
  volume  = {203},
  pages   = {12--17},
  year    = {1995},
  doi     = {10.1016/0375-9601(95)00345-4}
}

@article{Leaver1985,
  author  = {Leaver, E. W.},
  title   = {An analytic representation for the quasi-normal modes of {Kerr} black holes},
  journal = {Proceedings of the Royal Society of London. A. Mathematical and Physical Sciences},
  volume  = {402},
  pages   = {285--298},
  year    = {1985},
  doi     = {10.1098/rspa.1985.0119}
}

@article{Masuda1967,
  author  = {Masuda, Ky\^uya},
  title   = {A unique continuation theorem for solutions of the {Schr\"odinger} equations},
  journal = {Proceedings of the Japan Academy},
  volume  = {43},
  number  = {5},
  pages   = {361--364},
  year    = {1967},
  doi     = {10.3792/pja/1195521603}
}

@article{NamjooGuthKaiser2018,
  author  = {Namjoo, Mohammad Hossein and Guth, Alan H. and Kaiser, David I.},
  title   = {Relativistic corrections to nonrelativistic effective field theories},
  journal = {Physical Review D},
  volume  = {98},
  number  = {1},
  pages   = {016011},
  year    = {2018},
  doi     = {10.1103/PhysRevD.98.016011},
  eprint  = {1712.00445},
  archiveprefix = {arXiv}
}

@misc{pachon2026algebraicquantumkinematicssrselection,
      title={Algebraic quantum kinematics: {G}alilean covariance with positive energy confines unequal-time commutation to null sets},
      author={Leonardo A. Pachon},
      year={2026},
      eprint={2604.26267},
      archivePrefix={arXiv},
      primaryClass={quant-ph},
      url={https://arxiv.org/abs/2604.26267},
}

@misc{pachon2026galileanreehschliederobstruction,
      title={{G}alilean {Reeh--Schlieder} obstruction in the vacuum and the thermal {Reeh--Schlieder} property},
      author={Leonardo A. Pachon},
      year={2026},
      eprint={2604.26271},
      archivePrefix={arXiv},
      primaryClass={quant-ph},
      url={https://arxiv.org/abs/2604.26271},
}

@article{PadmanabhanPadmanabhan2011,
  author  = {Padmanabhan, Hamsa and Padmanabhan, T.},
  title   = {Nonrelativistic limit of quantum field theory in inertial and
             noninertial frames and the principle of equivalence},
  journal = {Physical Review D},
  volume  = {84},
  number  = {8},
  pages   = {085018},
  year    = {2011},
  note    = {arXiv:1110.1314},
  doi     = {10.1103/PhysRevD.84.085018}
}

@article{PaniCardosoGualtieriBertiIshibashi2012,
  author  = {Pani, Paolo and Cardoso, Vitor and Gualtieri, Leonardo and Berti, Emanuele and Ishibashi, Akihiro},
  title   = {Perturbations of slowly rotating black holes: massive vector fields in the {Kerr} metric},
  journal = {Physical Review D},
  volume  = {86},
  number  = {10},
  pages   = {104017},
  year    = {2012},
  doi     = {10.1103/PhysRevD.86.104017},
  eprint  = {1209.0773},
  archiveprefix = {arXiv}
}

@misc{PasqualottoShlapentokhRothmanVandeMoortel2023,
  author        = {Pasqualotto, Federico and Shlapentokh-Rothman, Yakov and {Van de Moortel}, Maxime},
  title         = {The asymptotics of massive fields on stationary spherically symmetric black holes for all angular momenta},
  year          = {2023},
  eprint        = {2303.17767},
  archiveprefix = {arXiv}
}

@book{Post2012,
  author    = {Post, Olaf},
  title     = {Spectral analysis on graph-like spaces},
  series    = {Lecture Notes in Mathematics},
  publisher = {Springer Berlin Heidelberg},
  year      = {2012},
  doi       = {10.1007/978-3-642-23840-6}
}

@article{Radzikowski1996,
  author  = {Radzikowski, Marek J.},
  title   = {Micro-local approach to the {Hadamard} condition in quantum field theory on curved space-time},
  journal = {Communications in Mathematical Physics},
  volume  = {179},
  pages   = {529--553},
  year    = {1996},
  doi     = {10.1007/bf02100096}
}

@article{RauchTaylor1975,
  author  = {Rauch, Jeffrey and Taylor, Michael},
  title   = {Potential and scattering theory on wildly perturbed domains},
  journal = {Journal of Functional Analysis},
  volume  = {18},
  pages   = {27--59},
  year    = {1975},
  doi     = {10.1016/0022-1236(75)90028-2}
}

@book{ReedSimon1972,
  author    = {Reed, Michael and Simon, Barry},
  title     = {Methods of modern mathematical physics {I}: functional analysis},
  publisher = {Academic Press},
  year      = {1972},
  doi       = {10.1016/b978-0-12-585001-8.x5001-6}
}

@article{Requardt1982,
  author  = {Requardt, Manfred},
  title   = {Spectrum condition, analyticity, {R}eeh--{S}chlieder and
             cluster properties in non-relativistic {G}alilei-invariant
             quantum theory},
  journal = {Journal of Physics A: Mathematical and General},
  volume  = {15},
  number  = {12},
  pages   = {3715--3723},
  year    = {1982},
  doi     = {10.1088/0305-4470/15/12/023}
}

@article{Requardt1986,
  author  = {Requardt, Manfred},
  title   = {{Reeh--Schlieder}-type density results in one- and $n$-body {Schr{\"o}dinger} theory and the ``unique continuation problem''},
  journal = {Journal of Mathematical Physics},
  volume  = {27},
  pages   = {1571--1577},
  year    = {1986},
  doi     = {10.1063/1.527068}
}

@article{SaBarretoZworski1997,
  author  = {{S{\'a} Barreto}, Ant{\^o}nio and Zworski, Maciej},
  title   = {Distribution of resonances for spherical black holes},
  journal = {Mathematical Research Letters},
  volume  = {4},
  pages   = {103--121},
  year    = {1997},
  doi     = {10.4310/mrl.1997.v4.n1.a10}
}

@article{SahlmannVerch2000,
  author  = {Sahlmann, Hanno and Verch, Rainer},
  title   = {Passivity and microlocal spectrum condition},
  journal = {Communications in Mathematical Physics},
  volume  = {214},
  pages   = {705--731},
  year    = {2000},
  doi     = {10.1007/s002200000297}
}

@article{Schoene1979,
  author  = {Schoene, Andrew Y.},
  title   = {On the nonrelativistic limits of the {Klein--Gordon} and {Dirac} equations},
  journal = {Journal of Mathematical Analysis and Applications},
  volume  = {71},
  pages   = {36--47},
  year    = {1979},
  doi     = {10.1016/0022-247x(79)90216-6}
}

@article{SchwartzGiulini2019,
  author  = {Schwartz, Philip K. and Giulini, Domenico},
  title   = {Post-{Newtonian} corrections to {Schr{\"o}dinger} equations in gravitational fields},
  journal = {Classical and Quantum Gravity},
  volume  = {36},
  pages   = {095016},
  year    = {2019},
  doi     = {10.1088/1361-6382/ab0fbd}
}

@article{Shale1962,
  author  = {Shale, David},
  title   = {Linear symmetries of free boson fields},
  journal = {Transactions of the American Mathematical Society},
  volume  = {103},
  pages   = {149--167},
  year    = {1962},
  doi     = {10.1090/s0002-9947-1962-0137504-6}
}

@article{ShlapentokhRothman2014,
  author  = {Shlapentokh-Rothman, Yakov},
  title   = {Exponentially growing finite energy solutions for the {Klein--Gordon} equation on sub-extremal {Kerr} spacetimes},
  journal = {Communications in Mathematical Physics},
  volume  = {329},
  pages   = {859--891},
  year    = {2014},
  doi     = {10.1007/s00220-014-2033-x}
}

@misc{ShlapentokhRothmanVandeMoortel2023,
  author        = {Shlapentokh-Rothman, Yakov and {Van de Moortel}, Maxime},
  title         = {Polynomial time decay for solutions of the {Klein--Gordon} equation on a subextremal {Reissner--Nordstr{\"o}m} black hole},
  year          = {2023},
  eprint        = {2401.00048},
  archiveprefix = {arXiv}
}

@article{Simon1973,
  author  = {Simon, Barry},
  title   = {Resonances in {$n$}-body quantum systems with dilatation analytic potentials and the foundations of time-dependent perturbation theory},
  journal = {Annals of Mathematics},
  volume  = {97},
  pages   = {247--274},
  year    = {1973},
  doi     = {10.2307/1970847}
}

@article{SofferWeinstein1998,
  author  = {Soffer, A. and Weinstein, M. I.},
  title   = {Time dependent resonance theory},
  journal = {Geometric and Functional Analysis},
  volume  = {8},
  pages   = {1086--1128},
  year    = {1998},
  doi     = {10.1007/s000390050124}
}

@article{StolzWeidmann1993,
  author  = {Stolz, G{\"u}nter and Weidmann, Joachim},
  title   = {Approximation of isolated eigenvalues of ordinary differential operators},
  journal = {Journal f{\"u}r die reine und angewandte Mathematik},
  volume  = {445},
  pages   = {31--44},
  year    = {1993},
  doi     = {10.1515/crll.1993.445.31}
}

@article{StrohmaierWitten2023,
  author        = {Strohmaier, Alexander and Witten, Edward},
  title         = {The timelike tube theorem in curved spacetime},
  journal       = {Communications in Mathematical Physics},
  volume        = {405},
  pages         = {153},
  year          = {2024},
  doi           = {10.1007/s00220-024-05009-3},
  eprint        = {2303.16380},
  archiveprefix = {arXiv}
}

@article{Sussman2024,
  author  = {Sussman, Ethan},
  title   = {Massive waves gravitationally bound to static bodies},
  journal = {Proceedings of the American Mathematical Society},
  volume  = {152},
  number  = {8},
  pages   = {3319--3337},
  year    = {2024},
  eprint  = {2312.00959},
  archiveprefix = {arXiv}
}

@book{Takesaki1970,
  author    = {Takesaki, M.},
  title     = {Tomita's Theory of Modular {Hilbert} Algebras and its Applications},
  series    = {Lecture Notes in Mathematics},
  volume    = {128},
  publisher = {Springer},
  address   = {Berlin},
  year      = {1970},
  doi       = {10.1007/BFb0065832}
}

@book{Thaller1992,
  author    = {Thaller, Bernd},
  title     = {The {Dirac} equation},
  publisher = {Springer Berlin Heidelberg},
  year      = {1992},
  doi       = {10.1007/978-3-662-02753-0}
}

@article{Titchmarsh1962,
  author  = {Titchmarsh, Edward Charles},
  title   = {On the relation between the eigenvalues in relativistic and non-relativistic quantum mechanics},
  journal = {Proceedings of the Royal Society of London. Series A. Mathematical and Physical Sciences},
  volume  = {266},
  pages   = {33--46},
  year    = {1962},
  doi     = {10.1098/rspa.1962.0045}
}

@article{Unruh1976,
  author  = {Unruh, W. G.},
  title   = {Notes on black-hole evaporation},
  journal = {Physical Review D},
  volume  = {14},
  number  = {4},
  pages   = {870--892},
  year    = {1976},
  doi     = {10.1103/PhysRevD.14.870}
}

@article{Verch1997,
  author  = {Verch, Rainer},
  title   = {Continuity of symplectically adjoint maps and the algebraic
             structure of {Hadamard} vacuum representations for quantum
             fields on curved spacetime},
  journal = {Reviews in Mathematical Physics},
  volume  = {9},
  number  = {5},
  pages   = {635--674},
  year    = {1997},
  doi     = {10.1142/S0129055X97000233}
}

@article{Veselic1970,
  author  = {Veseli{\'c}, K.},
  title   = {A spectral theory for the {Klein--Gordon} equation with an external electrostatic potential},
  journal = {Nuclear Physics A},
  volume  = {147},
  pages   = {215--224},
  year    = {1970},
  doi     = {10.1016/0375-9474(70)90524-5}
}

@article{Veselic1971,
  author  = {Veseli{\'c}, K.},
  title   = {Perturbation of pseudoresolvents and analyticity in {$1/c$} in relativistic quantum mechanics},
  journal = {Communications in Mathematical Physics},
  volume  = {22},
  pages   = {27--43},
  year    = {1971},
  doi     = {10.1007/bf01651582}
}

@article{Veselic1983,
  author  = {Veseli{\'c}, K.},
  title   = {On the nonrelativistic limit of the bound states of the {Klein--Gordon} equation},
  journal = {Journal of Mathematical Analysis and Applications},
  volume  = {96},
  pages   = {63--84},
  year    = {1983},
  doi     = {10.1016/0022-247x(83)90028-8}
}

@book{Weidmann1987,
  author    = {Weidmann, Joachim},
  title     = {Spectral theory of ordinary differential operators},
  series    = {Lecture Notes in Mathematics},
  publisher = {Springer Berlin Heidelberg},
  year      = {1987},
  doi       = {10.1007/bfb0077960}
}

\end{document}